\documentclass[3p,12pt]{elsarticle}

\usepackage{amssymb,latexsym,amsmath,color,subfig,amsthm,url,multirow}

\journal{}

\newcommand{\eps}{\varepsilon}
\newcommand{\abs}[1]{\left\vert#1\right\vert}
\newcommand{\set}[1]{\left\{#1\right\}}
\newcommand{\norm}[1]{\left|#1\right|}
\newcommand{\cL}{\mathcal{L}}
\newcommand{\p}{\partial}
\newcommand{\mx}{\mathbf{x}}
\newcommand{\my}{\mathbf{y}}
\newcommand{\mz}{\mathbf{z}}
\newcommand{\mH}{\mathbf{H}}
\newcommand{\mU}{\mathbf{U}}
\newcommand{\mV}{\mathbf{V}}
\newcommand{\mS}{\mathbf{S}}
\newcommand{\vn}{\boldsymbol{\nu}}
\newcommand{\vt}{\boldsymbol{\theta}}
\newcommand{\vx}{\boldsymbol{\xi}}

\newtheorem{thm}{Theorem}[section]

\newtheorem{rem}[thm]{Remark}

\begin{document}

\begin{frontmatter}



\title{Multi-frequency subspace migration for imaging of perfectly conducting, arc-like cracks}


\author{Won-Kwang Park}
\ead{parkwk@kookmin.ac.kr}
\address{Department of Mathematics, Kookmin University, Seoul, 136-702, Korea.}

\begin{abstract}
Multi-frequency subspace migration imaging technique are usually adopted for the non-iterative imaging of unknown electromagnetic targets such as cracks in the concrete walls or bridges, anti-personnel mines in the ground, etc. in the inverse scattering problems. It is confirmed that this technique is very fast, effective, robust, and can be applied not only full- but also limited-view inverse problems if suitable number of incident and corresponding scattered field are applied and collected. But in many works, the application of such technique is somehow heuristic. Under the motivation of such heuristic application, this contribution analyzes the structure of imaging functional employed in the subspace migration imaging technique in two-dimensional full- and limited-view inverse scattering when the unknown target is arbitrary shaped, arc-like perfectly conducting cracks located in the homogeneous two-dimensional space. Opposite to the Statistical approach based on the Statistical Hypothesis Testing, our approach is based on the fact that subspace migration imaging functional can be expressed by a linear combination of Bessel functions of integer order of the first kind. This is based on the structure of the Multi-Static Response (MSR) matrix collected in the far-field at nonzero frequency in either Transverse Magnetic (TM) mode (Dirichlet boundary condition) or Transverse Electric (TE) mode (Neumann boundary condition). Explored expression of imaging functionals gives us certain properties of subspace migration and an answer of why multi-frequency enhances imaging resolution. Particularly, we carefully analyze the subspace migration and confirm some properties of imaging when a small number of incident field is applied. Consequently, we simply introduce a weighted multi-frequency imaging functional and confirm that which is an improved version of subspace migration in TM mode. Various results of numerical simulations via the far-field data affected by large amount of random noise are well matched with the analytical results derived herein, and give some ideas of future studies.
\end{abstract}

\begin{keyword}
Full- and limited-view inverse scattering problems \sep multi-frequency subspace migration \sep perfectly conducting cracks \sep Multi-Static Response (MSR) matrix \sep numerical experiments



\end{keyword}

\end{frontmatter}





\section{Introduction}\label{Sec1}
The main purpose of inverse scattering problems like a non-destructive evaluation is identifying unknown characteristics of defects such as sizes, locations, shapes, electric and magnetic properties. Among them, identification of shape of arbitrary shaped cracks in a structure such as bridges, concrete walls, machines, etc., is an interesting problems that can be easily faced in human life. Unfortunately, due to the intrinsic difficulties of its ill-posedness and nonlinearity, a successful accomplish of this problem cannot be easily performed.

Nowadays, many remarkable inversion techniques and corresponding computational environments are developed and established to solve this problem. The main approach of solving is based on the Newton-type iteration method that is finding the shape of target (minimizer) which minimize the discrete norm (generally $L^2-$norm) between measured scattered or far-fields in the presence of true and man-made targets. This iterative-based technique such as the level-set method and optimization algorithm have been applied successfully to identify the number, locations, shapes, and topological properties of cracks with a small number of directions of incident and scattered field data as exhibited in many works \cite{ADIM,CR,DL,K,PL4,VXB}. Nevertheless, finding a good initial guess close to the target, estimating \textit{a priori} information such as length, locations, material properties, selecting an appropriate regularization terms highly dependent on the problem, and evaluation of so-called Fr{\'e}chet (or domain) derivative are must be considered beforehand as described in \cite{KSY}. If one of these conditions is not full-filled, one shall encounters various problems such as phenomenon of non-convergence, occurrence of local minimizer problem, and requirement of large computational costs due to the large number of iteration procedure. In order to overcome, a simultaneous reconstruction algorithm is developed \cite{GH,RLLZ}. However, this kind of algorithm still has a difficulty that it must be performed with a good initial guess and generally very slow hence this fact indicate that it is desperately required the development of alternative fast algorithm and corresponding rigorous mathematical theory for finding a good initial guess in the the beginning stage of iteration procedure.

Correspondingly, for an alternative, various non-iterative imaging algorithms have been developed and successfully applied to various inverse problems. Based on the calculations of inverse Fourier transform, a variational algorithm has been proposed in \cite{AIM,AK,AMV}. Based on these nobel works, this algorithm produces very good results but it is still restricted for identifying locations of small inclusions so that extension to the reconstruction of the arbitrary shaped electromagnetic targets is further research topic. In \cite{CD,CHM,KR}, a linear sampling method is developed and applied for determining locations and shapes of unknown scatterers, but this approach requires a large number of directions of incident and scattered field, and does not considered in the limited-view inverse scattering problems. MUltiple SIgnal Classification, which is closely linked to the linear sampling method (see \cite{C}), is also applied various inverse scattering problems for full- and limited-view problems for imaging arbitrary shaped thin penetrable electromagnetic scatterers, cracks, and extended targets in two- and three-dimensional spaces, refer to \cite{AIL,AILP,AGKPS,AKLP,HSZ,JP,PhDWKP,P4,PL1,PL2,PL3}. Based on results in \cite{ABC,PL1,PL3}, it still requires huge amount of directions of incident and scattered field data in order to obtain an acceptable result and does not guarantee complete shape of targets due to the intrinsic resolution limit related to the half of applied wavelength. In recent works \cite{AIL,P4,PL2}, MUSIC algorithm is applied to the limited-view inverse problems however, it yields incorrect location of small inhomogeneities and shape of extended targets (for example, see Figure \ref{FigMUSIC}), and throughout the structure of MUSIC imaging function derived in \cite{JP}, the reason of this phenomenon is mathematically proved. Originally, topological derivative strategy is applied in shape optimization problems but throughout recent works, it has been confirm that it is a non-iterative imaging technique and successfully applied to the various inverse scattering problems, see \cite{AGJK,B,CR,MKP,NB,P2,P5,P6,SZ} and references therein. The remarkable advantages of topological derivative based imaging technique are that it produces good results even in the small number of directions of incident field data and robust with respect to the huge amount of random noise. However, throughout the derivation of topological derivatives in \cite{AKLP,P2,P5,P6}, this strategy covers only the case of full-view inverse scattering problems and small number of directions of incident field must span unit circle $\mathbb{S}^1$ it means that is is especially vulnerable in the limited-view inverse problems.

\begin{figure}[!ht]
\begin{center}
\subfloat[full-view case]{\label{FigMUSIC1}\includegraphics[width=0.32\textwidth]{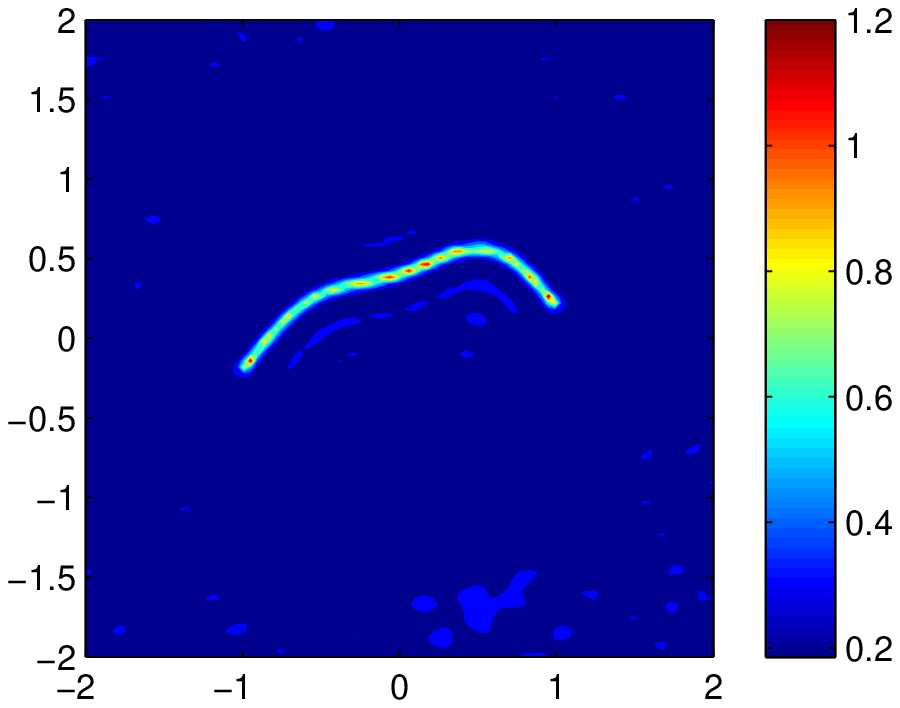}}
\subfloat[limited-view case]{\label{FigMUSIC2}\includegraphics[width=0.32\textwidth]{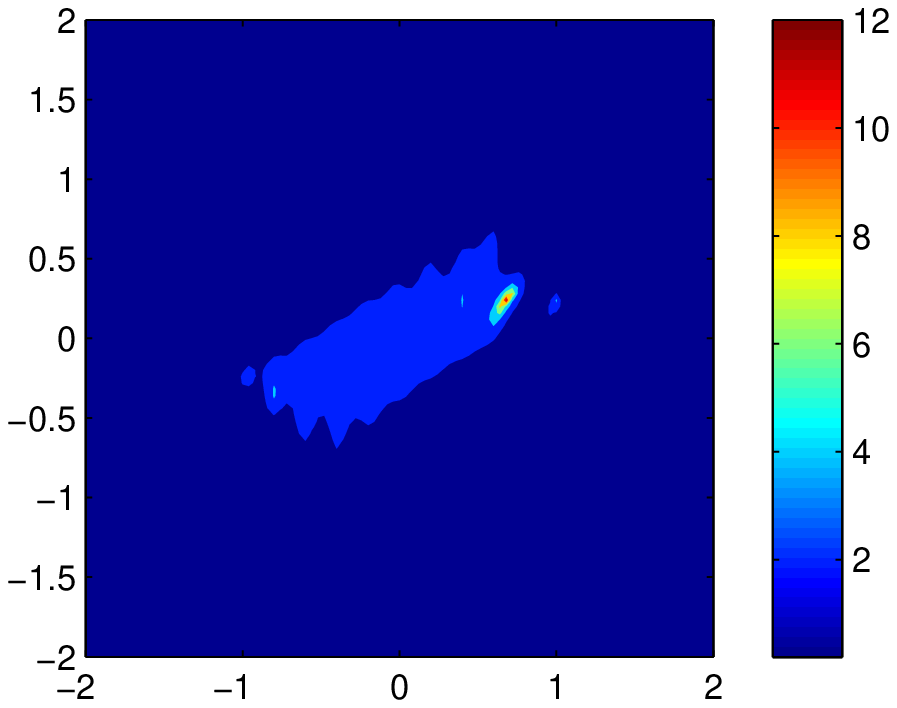}}
\subfloat[true shape]{\label{FigMUSIC3}\includegraphics[width=0.32\textwidth]{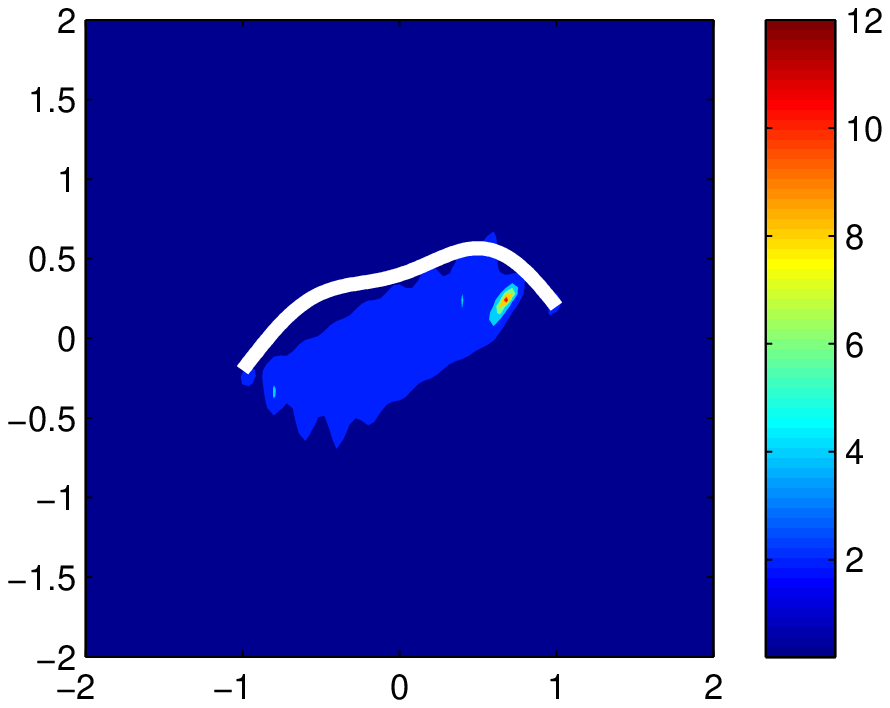}}
\caption{\label{FigMUSIC}Reconstruction of arbitrary shaped extended perfectly conducting crack via MUSIC-type algorithm \cite{PL1}. In the full-view problem, MUSIC offers very good result but very poor result is appeared in the limited-view problem.}
\end{center}
\end{figure}

Kirchhoff and subspace migration imaging techniques operated at single and multiple time-harmonic frequency has been applied successfully not only for full- but also limited-view inverse scattering problem if the number of directions of incident field and corresponding scattered field are sufficiently large enough. Related works can be found in \cite{AGKPS,HHSZ,JKHP,KP,P1,P3,P4,PL2,PP} and reference therein. Although these techniques are very robust with respect to the random noise and media, and offer very good results for imaging of small and extended scatterers, it is still used in many research field without rigorous mathematical theory. Based on the statistical hypothesis testing \cite{AGKPS} and relationship between Bessel functions of integer order of the first kind \cite{JKHP}, it is confirm that why multi-frequency Kirchhoff and subspace migration technique gives better results than single-frequency ones in the full-view inverse scattering problems. Recently, in \cite{KP}, it is proved that why they can be applied in the limited-view problems, and discovered some necessary conditions for obtaining proper results. However, this work was restricted in the imaging of small targets so that analysis of imaging functions for arbitrary shaped extended target is still remaining.

The main purpose of this paper is a rigorous mathematical analysis and identification of the structure of subspace migration for imaging of arbitrary shaped perfectly conducting cracks in either Transverse Magnetic (TM) mode (Dirichlet boundary condition) or Transverse Electric (TE) mode (Neumann boundary condition) for full- and limited-view inverse scattering problems. This is based on a factorization of the Multi-Static Response (MSR) matrix collected in the far-field data, and the fact that structure of singular vectors of MSR matrix is linked to the known incident field data. Based on the orthonormal property of left- and right-singular vectors of MSR matrix, it is clear that why subspace migration imaging produces an imaging of targets but this fact cannot explain some facts for example, appearance of unexpected replicas and the reason behind enhancement in the imaging performance by applying multiple frequencies. Fortunately, structure of singular vectors of MSR matrix leads us that subspace migration imaging functional is highly related to the integral representations of the Bessel functions of integer order of the first kind so that some definite integral formulas of Bessel functions is needed. Unfortunately, some of such formulas derived in the full-view case and does not considered in the limited-view case because there is no finite representation linked to the considering problem so that we will evaluate an approximation of such integrals by applying well-known Jacobi-Anger expansion formula and asymptotic properties of Bessel functions. From the derived structure of imaging functional, we can explore certain properties (such that why two curves are appeared instead of true curve for TE case), fundamental limitations, and a clue of improvements (for example, multi-frequency subspace migration weighted by each frequency).

For numerical simulations, two different approaches -- the Nystr\"{o}m method in \cite{K,M1} and the second-kind Fredholm integral equation \cite{N} -- are adopted for evaluating the far-field pattern data in order to avoid committing \textit{inverse crimes}. The numerical results of images from far-field data corrupted by the large amount of random noise appear almost indistinguishable from those. Let us signal that the extension of several cracks is available without any additional configuration. For simplicity, its mathematical derivation is not considered; only some numerical results for two cracks are illustrated. Once the shape of interested crack is mapped, it can be accepted as an initial guess and one can evolve it in order to retrieve a better shape via a Newton-type iteration algorithm \cite{K,KS,M2}, level-set methodology \cite{ADIM,DL,PL2,RLLZ,VXB} and optimization concept \cite{AGJKLY}, and it is possible to observe that only a few iteration procedure is required due to the closeness of initial guess so that it does not requires a great deal of computation.

We will organize this paper as follows. In section \ref{Sec2}, the two-dimensional direct scattering problem is briefly discussed in some detail, here mostly for the sake of completeness. In section \ref{Sec3}, multi-frequency based subspace migration imaging functional is sketched and mathematical analysis of the structure of imaging functional is provided in full- and limited-view cases. In section \ref{Sec4}, a set of results of numerical simulations from noisy, discrete TM and TE data, such data being computed from the application of
the Nystr\"{o}m method in \cite{K,M1} or by the second-kind Fredholm integral equation \cite{N}, is exhibited for supporting our investigations. A short conclusion including an outline of current and future work is mentioned in section \ref{Sec5}.

\section{Mathematical survey on two-dimensional direct scattering problem and multi-frequency subspace migration imaging}\label{Sec2}
\subsection{Direct scattering problem}
In this section, we consider the two-dimensional electromagnetic scattering by a perfectly conducting crack, denoted by $\Gamma$, located in the homogeneous space $\mathbb{R}^2$. The crack is an oriented piecewise smooth nonintersecting arc without cusp that can be represented as
\begin{equation}\label{crack}
\Gamma=\set{\mz(s):s\in[-1,1]}
\end{equation}
where $\mz:[-1,1]\longrightarrow\mathbb{R}^2$ is an injective piecewise $\mathcal{C}^3$ function (see Figure \ref{figureCrack}).

\begin{figure}[!ht]
\begin{center}
\includegraphics[width=0.25\textwidth]{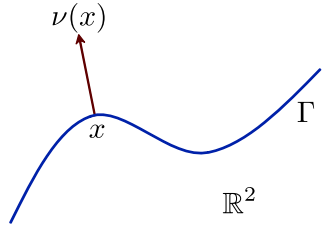}
\caption{\label{figureCrack}Illustration of two-dimensional perfectly conducting crack $\Gamma$.}
\end{center}
\end{figure}

First, let us consider the so-called Transverse Magnetic polarization case, letting $u(\mx,\vt;k)$ be the (single-component) electric field that satisfies the two-dimensional Helmholtz wave equation
\begin{equation}\label{HelmEQ}
\Delta u(\mx,\vt;k)+k^2 u(\mx,\vt;k)=0\quad\mbox{in}\quad\mathbb{R}^2\backslash\Gamma
\end{equation}
with incident direction $\vt$ and strictly positive wave number $k=\omega\sqrt{\mu\eps}$, letting $\eps$ be the electric permittivity and $\mu$ the magnetic permeability. Throughout this paper, applied wave number $k$ is of the form $k=2\pi/\lambda$, where $\lambda$ denotes the wavelength. The field cannot penetrate into $\Gamma$, i.e., $u$ satisfies the Dirichlet boundary condition
\begin{equation}\label{HelmBC}
u(\mx,\vt;k)=0\quad\mbox{on}\quad\Gamma.
\end{equation}
Conversely, let us consider the Transverse Electric polarization case, letting $u$ be the (single-component) magnetic field that satisfies the two-dimensional Helmholtz wave equation  (\ref{HelmEQ}) yet now with the following
Neumann boundary condition on $\Gamma$:
\begin{equation}\label{HelmNC}
\frac{\p u(\mx,\vt;k)}{\p\vn(\mx)}=0\quad\mbox{on}\quad\Gamma\backslash\set{\mz(-1),\mz(1)},
\end{equation}
where $\vn(\mx)$ is a unit normal vector to $\Gamma$ at $\mx$.

Let us notice that the total field can always be decomposed as $u(\mx,\vt;k)=u_{\mathrm{inc}}(\mx,\vt;k)+u_{\mathrm{scat}}(\mx,\vt;k)$. In this paper, we will consider the plane-wave illumination hence, $u_{\mathrm{inc}}(\mx,\vt;k)=e^{ik\vt\cdot\mx}$ be the given incident field for incident direction $\vt\in\mathbb{S}^1$ (unit circle), and $u_{\mathrm{scat}}(\mx,\vt;k)$ be the unknown scattered field, which is required to satisfy the Sommerfeld radiation condition
\[\lim_{\abs{\mx}\to\infty}\sqrt{\abs{\mx}}\left(\frac{\p u_{\mathrm{scat}}(\mx,\vt;k)}{\p\abs{\mx}}-iku_{\mathrm{scat}}(\mx,\vt;k)\right)=0\]
uniformly into all directions $\hat{\mx}=\frac{\mx}{\abs{\mx}}$. The determination of $u_{\mathrm{scat}}$ is a special case of the following problem
\begin{equation}\label{HelmSEQ}
\Delta u_{\mathrm{scat}}(\mx,\vt;k)+k^2 u_{\mathrm{scat}}(\mx,\vt;k)=0\quad\mbox{in}\quad\mathbb{R}^2\backslash\Gamma
\end{equation}
that satisfies the Dirichlet boundary condition
\begin{equation}\label{ExtDC}
u_{\mathrm{scat}}(\mx,\vt;k)=f(\mx,\vt;k)\quad\mbox{on}\quad\Gamma
\end{equation}
or the Neumann boundary condition
\begin{equation}\label{ExtNC}
{\frac{\p u_{\mathrm{scat}}(\mx,\vt;k)}{\p\vn(\mx)}=f(\mx,\vt;k)\quad\mbox{on}\quad\Gamma,}
\end{equation}
and the Sommerfeld radiation condition. Let us remind that from the boundary conditions (\ref{HelmBC}) and (\ref{HelmNC}), we can set $f(\mx,\vt;k)=-u_{\mathrm{inc}}(\mx,\vt;k)$ and $f(\mx,\vt;k)=-\nabla u_{\mathrm{inc}}(\mx,\vt;k)\cdot\vn(\mx)$ for the boundary conditions (\ref{ExtDC}) and (\ref{ExtNC}), respectively. Let us notice at this stage that the above  works only for \textit{smooth} arcs, and we should refer to \cite{N} for a broader and thoughtful coverage of electromagnetic scattering by general arcs, including ours.

\subsection{The far-field pattern}
Let us first consider the case of the Dirichlet boundary problem (TM polarization). The author in \cite{K} establishes the existence of a solution by searching it in the form of a single-layer potential
\begin{equation}\label{SLP}
u_{\mathrm{scat}}(\mx,\vt;k)=\int_{\Gamma}\Phi(\mx,\my;k)\varphi(\my,\vt;k)d\my\quad\mbox{for}\quad \mx\in\mathbb{R}^2\backslash\Gamma
\end{equation}
with the two-dimensional fundamental solution to the Helmholtz equation
\[\Phi(\mx,\my;k)=\frac{i}{4}\mathrm{H}_0^1(k\abs{\mx-\my})\quad\mbox{for}\quad \mx\ne\my,\]
expressed in terms of the Hankel function $\mathrm{H}_0^1$ of order zero and of the first kind. For the Neumann boundary problem, the author in \cite{M1} establishes the existence of a solution by searching it as a double-layer potential
\begin{equation}\label{DLP}
u_{\mathrm{scat}}(\mx,\vt;k)=\int_{\Gamma}\frac{\p\Phi(\mx,\my;k)}{\p\vn(\my)}\psi(\my,\vt;k)d\my\quad\mbox{for}\quad \mx\in\mathbb{R}^2\backslash\Gamma.
\end{equation}

Let us assume that for all $\mx\in\Gamma\backslash\set{\mz(-1),\mz(1)}$, the limits of the following quantities exist
\begin{align*}
u(\mx,\vt;k)&=\lim_{h\to+0}u(\mx\pm h\vn(\mx),\vt;k)\\
\frac{\p u_{\pm}(\mx,\vt;k)}{\p\vn(\mx)}&=\lim_{h\to+0}\vn(\mx)\cdot\nabla u(\mx\pm h\vn(\mx),\vt;k)\\
-\varphi(\mx,\vt;k)&=\frac{\p u_{+}(\mx,\vt;k)}{\p\vn(\mx)}-\frac{\p u_{-}(\mx,\vt;k)}{\p\vn(\mx)}\\
-\psi(\mx,\vt;k)&=u_{+}(\mx,\vt;k)-u_{-}(\mx,\vt;k).
\end{align*}

The far-field pattern $u_{\infty}$ of the scattered field $u_{\mathrm{scat}}$ is defined on the two-dimensional unit circle $\mathbb{S}^1$. It can be represented as
\[u_{\mathrm{scat}}(\mx,\vt;k)=\frac{e^{ik\abs{\mx}}}{\sqrt{\abs{\mx}}}\left(u_{\infty}(\hat{\mx},\vt;k)+O\left(\frac{1}{\abs{\mx}}\right)\right)\]
uniformly in all directions $\hat{\mx}=\mx/\abs{\mx}$ and $\abs{\mx}\longrightarrow\infty$. From the above representation and the asymptotic formula for the Hankel function, the far field pattern for Dirichlet boundary problem can be written as
\begin{align}
\begin{aligned}\label{FFPD}
u_\infty(\hat{\mx},\vt;k)&=-\frac{e^{i\frac{\pi}{4}}}{\sqrt{8\pi k}}\int_{\Gamma}e^{-ik\hat{\mx}\cdot\my}\left(\frac{\p u_{+}(\my,\vt;k)}{\p\vn(\my)}-\frac{\p u_{-}(\my,\vt;k)}{\p\vn(\my)}\right)d\my\\
&=\frac{e^{i\frac{\pi}{4}}}{\sqrt{8\pi k}}\int_{\Gamma}e^{-ik\hat{\mx}\cdot\my}\varphi(\my,\vt;k)d\my.
\end{aligned}
\end{align}

Similarly, the far-field pattern for the Neumann boundary problem can be expressed as
\begin{align}
\begin{aligned}\label{FFPN}
u_\infty(\hat{\mx},\vt;k)&=-\frac{e^{i\frac{\pi}{4}}}{\sqrt{8\pi k}}\int_{\Gamma}\frac{\p e^{-ik\hat{\mx}\cdot\my}}{\p\vn(\my)}\bigg(u_{+}(\my,\vt;k)-u_{-}(\my,\vt;k)\bigg)d\my\\
&=-\sqrt{\frac{k}{8\pi}}e^{-i\frac{\pi}{4}} \int_{\Gamma}\hat{\mx}\cdot\vn(\my)e^{-ik\hat{\mx}\cdot\my}\psi(\my,\vt;k)d\my.
\end{aligned}
\end{align}

\subsection{Introduction to multi-frequency subspace migration imaging functional}
In this section, we apply the far-field pattern formulas (\ref{FFPD}) and (\ref{FFPN}) in order to build up a subspace migration imaging functional. For that purpose, we use the eigenvalue structure of the Multi-Static Response (MSR) matrix
\[\mathbb{K}(k):=\bigg[K_{jl}(\hat{\mx}_j,\vt_l;k)\bigg]_{j,l=1}^{N}=\bigg[u_\infty(\hat{\mx}_j,\vt_l;k)\bigg]_{j,l=1}^{N}.\]

First, let us consider the Dirichlet boundary condition case. If the directions of incident and observation are coincide i.e., if $\hat{\mx}_j=-\vt_j$, then the MSR matrix $\mathbb{K}$ can be written as
\begin{equation}\label{MSRD}
\mathbb{K}(k)=\frac{e^{i\frac{\pi}{4}}}{\sqrt{8\pi k}}\int_{\Gamma}\mathbb{E}_{\mathrm{D}}(\hat{\mx},\my;k)\mathbb{F}_{\mathrm{D}}(\hat{\mx},\my;k)^Td\my,
\end{equation}
where $\mathbb{E}_{\mathrm{D}}(\hat{\mx},\my;k)$ is the illumination vector
\begin{equation}\label{VecED}
\mathbb{E}_{\mathrm{D}}(\hat{\mx},\my;k)=\bigg[e^{-ik\hat{\mx}_1\cdot\my},e^{-ik\hat{\mx}_2\cdot\my},\cdots,e^{-ik\hat{\mx}_N\cdot\my}\bigg]^T\bigg|_{\hat{\mx}_j=-\vt_j} =\bigg[e^{ik\vt_1\cdot\my},e^{ik\vt_2\cdot\my},\cdots,e^{ik\vt_N\cdot\my}\bigg]^T
\end{equation}
and where $\mathbb{F}_{\mathrm{D}}(\hat{\mx},\my;k)$ is the resulting density vector
\begin{equation}
\mathbb{F}_{\mathrm{D}}(\hat{\mx},\my;k)=\bigg[\varphi(\my,\vt_1;k),\varphi(\my,\vt_2;k),\cdots,\varphi(\my,\vt_N;k)\bigg]^T.
\end{equation}
Here, $\set{\hat{\mx}_j}_{j=1}^N\subset\mathbb{S}^1$ is a discrete finite set of observation directions and $\set{\vt_l}_{l=1}^N\subset\mathbb{S}^1$ is the same number of incident directions.

Formula (\ref{MSRD}) is a factorization of the MSR matrix that separates the known incoming wave information from the unknown information. The range of $\mathbb{K}(k)$ is determined by the span of the $\mathbb{E}_{\mathrm{D}}(\hat{\mx},\my;k)$ corresponding to the $\Gamma$, i.e., we can define a signal subspace by using a set of left singular vectors of $\mathbb{K}(k)$. We refer to \cite{HSZ,PL1} for a detailed discussion.

Assume that the crack is divided into $M$ different segments of size of order half the wavelength $\lambda/2$. Having in mind the Rayleigh resolution limit, any detail less than one-half of the wavelength cannot be probed, and only one point, say $\my_m$ for $m=1,2,\cdots,M$, at each segment is expected to contribute at the image space of the response matrix $\mathbb{K}(k)$, refer to \cite{ABC,AKLP,PL1,PL3}.

Since the coincide configuration of incident and observation directions, MSR matrix $\mathbb{K}(k)$ is complex symmetric, refer to \cite{AGKLS,HHSZ,PhDWKP,P3,P4,PL1,PL2,PL3}. Therefore, $\mathbb{K}(k)$ can be decomposed as
\begin{equation}\label{MatrixK}
\mathbb{K}(k)=\mathbb{H}(k)\mathbb{B}(k)\mathbb{H}(k)^T\approx\sum_{m=1}^{M}B_m(k)\mH_m(k)\mH_m^T(k).
\end{equation}
Having in mind of the structure of $\mathbb{E}_{\mathrm{D}}(\hat{\mx},\my;k)$ in (\ref{VecED}), we define a vector
\begin{equation}\label{VecD}
\mS_{\mathrm{D}}(\mx;k)=\bigg[e^{ik\vt_1\cdot\mx},e^{ik\vt_2\cdot\mx},\cdots,e^{ik\vt_N\cdot\mx}\bigg]^T
\end{equation}
and corresponding normalized vector $\hat{\mS}_{\mathrm{D}}(\mx;k)=\frac{\mS_{\mathrm{D}}(\mx;k)}{|\mS_{\mathrm{D}}(\mx;k)|}$. Then for $m=1,2,\cdots,M$, there exists some constants $\gamma_m$ such that (see \cite{HHSZ} for instance)
\begin{equation}\label{Approx1}
\mH_m(k)=e^{i\gamma_m}\hat{\mS}(\my_m;k).
\end{equation}
Since the first $M$ columns of the matrix $\mathbb{H}(k)$, $\set{\mH_1(k),\mH_2(k),\cdots,\mH_M(k)}$, are orthonormal, one can easily examine that
\begin{align}\label{Approx2}
\begin{aligned}
\hat{\mS}_{\mathrm{D}}(\mx;k)^*\mH_m(k)&\ne 0\quad\mbox{if}\quad \mx=\my_m\\
\hat{\mS}_{\mathrm{D}}(\mx;k)^*\mH_m(k)&\approx 0\quad\mbox{if}\quad \mx\ne\my_m
\end{aligned}
\end{align}
where $*$ is the mark of complex conjugate.

Based on above observations, we consider the following:
\begin{equation}\label{ImagingFunctionS}
\mathcal{I}_{\mathrm{D}}(\mx)=\left|\sum_{m=1}^{M}\left(\hat{\mS}_{\mathrm{D}}(\mx;k)^*\mH_m(k)\right) \left(\hat{\mS}_{\mathrm{D}}(\mx;k)^*\mH_m(k)\right)\right| =\left|\sum_{m=1}^{M}|\hat{\mS}_{\mathrm{D}}(\mx;k)^*\mH_m(k)|^2\right|.
\end{equation}
Then, from the relationships (\ref{Approx1}) and (\ref{Approx2}), $\mathcal{I}_{\mathrm{D}}(\mx)$ becomes $1$ at $\mx=\my_m\in\Gamma$ for $m=1,2,\cdots,M$, and $0$ at $\mx\in\mathbb{R}^2\backslash\Gamma$. With this, we can design a subspace migration imaging functional as follows: let us perform the Singular Value Decomposition (SVD) of $\mathbb{K}(k)$ as
\begin{equation}\label{SVD}
  \mathbb{K}(k)=\mathbb{U}(k)\Lambda(k)\mathbb{V}(k)^*\approx\sum_{m=1}^{M}\sigma_m(k)\mU_m(k)\mV_m^*(k),
\end{equation}
where $\sigma_m(k)$ denotes non-zero singular values, $\mU_m(k)$ and $\mV_m(k)$ are left- and right-singular vectors, respectively. Then, for several frequencies $\{k_f:f=1,2,\cdots,F\}$ we design a normalized imaging functional as
\begin{equation}\label{ImagingFunctionD}
\mathcal{I}_{\mathrm{D}}(\mx;F)=\frac{1}{F}\abs{\sum_{f=1}^{F}\sum_{m=1}^{M_f} \left(\hat{\mS}_{\mathrm{D}}(\mx;k_f)^*\mU_m(k_f)\right)\left(\hat{\mS}_{\mathrm{D}}(\mx;k_f)^*\overline{\mV}_m(k_f)\right)},
\end{equation}
where $M_f$ is number of nonzero singular values of MSR matrix at $k_f$ for $f=1,2,\cdots,F$. Then, from the relationships (\ref{Approx1}) and (\ref{Approx2}), $\mathcal{I}_{\mathrm{D}}(\mx;F)$ is expected to exhibit peaks of magnitude of $1$ at the location $\mx=\my_m$ for $m=1,2,\cdots,M_f$ and of small magnitude at $\mx\in\mathbb{R}^2\backslash\Gamma$. A suitable number of $M_f$ for each frequency can be found via careful thresholding (see \cite{HSZ,PL1,PL3}).

For the Neumann boundary condition case, the MSR matrix $\mathbb{K}(k)$ can be decomposed as
\begin{equation}\label{MSRN}
\mathbb{K}(k)=\sqrt{\frac{k}{8\pi}}e^{-i\frac{\pi}{4}}\int_{\Gamma}\mathbb{E}_{\mathrm{N}}(\hat{\mx},\my;k)\mathbb{F}_{\mathrm{N}}(\hat{\mx},\my;k)^Td\my,
\end{equation}
where $\mathbb{E}_{\mathrm{N}}(\hat{\mx},\my;k)$ is the illumination vector
\begin{align}
\begin{aligned}\label{VecEN}
\mathbb{E}_{\mathrm{N}}(\hat{\mx},\my;k)&=-\bigg[\hat{\mx}_1\cdot\vn(\my)e^{-ik\hat{\mx}_1\cdot \my},\hat{\mx}_2\cdot\vn(\my)e^{-ik\hat{\mx}_2\cdot \my},\cdots,\hat{\mx}_N\cdot\vn(\my)e^{-ik\hat{\mx}_N\cdot\my}\bigg]^T\bigg|_{\hat{\mx}_j=-\vt_j}\\
&=\bigg[\vt_1\cdot\vn(\my)e^{ik\vt_1\cdot\my},\vt_2\cdot\vn(\my)e^{ik\vt_2\cdot \my},\cdots,\vt_N\cdot\vn(\my)e^{ik\vt_N\cdot\my}\bigg]^T
\end{aligned}
\end{align}
and where $\mathbb{F}_{\mathrm{N}}(\hat{\mx},\my;k)$ is the corresponding density vector
\begin{equation}
\mathbb{F}_{\mathrm{N}}(\hat{\mx},\my;k)=\bigg[\psi(\my,\vt_1;k),\psi(\my,\vt_2;k),\cdots,\psi(\my,\vt_N;k)\bigg]^T.
\end{equation}

Formula (\ref{MSRN}) is a factorization of the MSR matrix that, like with the Dirichlet boundary condition case, separates the known incoming plane wave information from the unknown information. The range of $\mathbb{K}(k)$ is determined by the span of the $\mathbb{E}_{\mathrm{N}}(\hat{\mx},\my;k)$ corresponding to the $\Gamma$, i.e., we can define a signal subspace by using a set of left singular vectors of $\mathbb{K}(k)$.

The imaging algorithm for the Neumann boundary condition case is very similar to the Dirichlet boundary condition case. Based on the structure of (\ref{VecEN}), define a vector
\begin{equation}
\mS_{\mathrm{N}}(\mx;k)=\bigg[\vt_1\cdot\vn(\mx)e^{ik\vt_1\cdot\mx},\vt_2\cdot\vn(\mx)e^{ik\vt_2\cdot\mx},\cdots,\vt_N\cdot\vn(\mx)e^{ik\vt_N\cdot\mx}\bigg]^T
\end{equation}
and corresponding normalized vector $\hat{\mS}_{\mathrm{N}}(\mx;k)=\frac{\mS_{\mathrm{N}}(\mx;k)}{\norm{\mS_{\mathrm{N}}(\mx;k)}}$. Since the unit normal $\vn(\mx)$ is still unknown, for each point $\mx$ of the search domain, we use a set of directions $\vn_l(\mx)$ for $l=1,2,\cdots,L$, and we choose $\vn_l(\mx)$ which is to maximize the imaging functional among these directions at $\mx$. With this considerations, we suggest following normalized image functional with MSR matrices at several frequencies $\{k_f:f=1,2,\cdots,F\}$ as
\begin{equation}\label{ImagingFunctionN}
\mathcal{I}_{\mathrm{N}}(\mx;F)=\frac{1}{F}\abs{\sum_{f=1}^{F}\sum_{m=1}^{M_f}\max_{1\leq l\leq L} \bigg\{\left(\hat{\mS}_{\mathrm{N}}(\mx;k_f)^*\mU_m(k_f)\right)\left(\hat{\mS}_{\mathrm{N}}(\mx;k_f)^*\overline{\mV}_m(k_f)\right)\bigg\}},
\end{equation}
where $M_f$ is number of nonzero singular values of MSR matrix $\mathbb{K}$ at $k_f$ for $f=1,2,\cdots,F$.

\begin{rem}
  For the near-field data case, the illumination vector (\ref{VecED}) of MSR matrix $\mathbb{K}(k)$ (\ref{MSRD}) is (see \cite{HHSZ})
  \[\mathbb{E}_{\mathrm{D}}(\hat{\mx},\my;k)=\bigg[\Phi(-\vt_1,\my;k),\Phi(-\vt_2,\my;k),\cdots,\Phi(-\vt_N,\my;k)\bigg]^T.\]
  Therefore, it is natural to use the vector $\mS_{\mathrm{D}}(\mx;k)$ of (\ref{VecD}) for establishing imaging functional (\ref{ImagingFunctionD}) as follows:
  \begin{equation}\label{VecD2}
    \mS_{\mathrm{D}}(\mx;k)=\bigg[\Phi(-\vt_1,\mx;k),\Phi(-\vt_2,\mx;k),\cdots,\Phi(-\vt_N,\mx;k)\bigg]^T.
  \end{equation}
  Note that for sufficiently large $k$,
  \[\Phi(\mx,\my;k)=\frac{e^{i\frac{\pi}{4}}}{\sqrt{8k\pi}}e^{-ik\hat{\my}\cdot\mx}+o\left(\frac{1}{\sqrt{|\my|}}\right),\]
  hence the result will be similar to (\ref{ImagingFunctionD}).

  For the Neumann boundary case, the vector $\mS_{\mathrm{N}}(\mx;k)$ is selected as: \[\mS_{\mathrm{N}}(\mx;k)=\left[\frac{\p\Phi(-\vt_1,\mx;k)}{\p\vn_l(\mx)},\frac{\p\Phi(-\vt_2,\mx;k)}{\p\vn_l(\mx)}, \cdots,\frac{\p\Phi(-\vt_N,\mx;k)}{\p\vn_l(\mx)}\right]^T.\]
  Since for sufficiently large $k$,
  \[\frac{\p\Phi(\mx,\my;k)}{\p\vn(\my)}=-\frac{e^{i\frac{\pi}{4}}}{\sqrt{8k\pi}} \bigg(ik\hat{\my}\cdot\vn(\my)\bigg)e^{-ik\hat{\my}\cdot\mx}+o\left(\frac{1}{\sqrt{|\my|}}\right),\]
  the result will be similar to (\ref{ImagingFunctionN}).
\end{rem}

\section{Analysis of multi-frequency subspace migration imaging functionals: full-view case}\label{Sec3}
\subsection{Common features}\label{Sec3-1}
Based on the statistical hypothesis testing, image functionals (\ref{ImagingFunctionD}) and (\ref{ImagingFunctionN}) at single frequency offers an image with poor resolution hence one must applied sufficiently large $F$, i.e, application of multi-frequency is needed (see \cite{AGKPS,HSZ,P3,P4,PL2,PP} for instance). However, some phenomena such as appearance of unexpected ghost replicas cannot be explained. From now on, we carefully analyze the structure of (\ref{ImagingFunctionD}) and (\ref{ImagingFunctionN}). For this, we shall prove two useful identities in Theorem \ref{TheoremBessel}. Note that this result is derived in \cite[Lemma 4.1]{G} but our approach is different.

\begin{thm}\label{TheoremBessel}
  For sufficiently large $N$, following relations holds
  \begin{align*}
    &\frac{1}{N}\sum_{n=1}^{N}e^{ik\vt_n\cdot\mx}=\frac{1}{2\pi}\int_{\mathbb{S}^1}e^{ik\vt\cdot\mx}d\vt=J_0(k|\mx|)\\
    &\frac{1}{N}\sum_{n=1}^{N}\vt_n\cdot\vx e^{ik\vt_n\cdot\mx}=\frac{1}{2\pi}\int_{\mathbb{S}^1}\vt\cdot\vx e^{ik\vt\cdot\mx}d\vt=i\left(\frac{\mx}{|\mx|}\cdot\vx\right)J_1(k|\mx|),
  \end{align*}
  where $\vx\in\mathbb{R}^2$, and $\vt_n\in\mathbb{S}^1$, $n=1,2,\cdots,N$. Here, $J_n(\cdot)$ denotes the Bessel function of integer order $n$ of the first kind.
\end{thm}
\begin{proof}
  We write $\vt:=(\cos\theta,\sin\theta)$, $\vx=\rho(\cos\xi,\sin\xi)$ and $\mx=r(\cos\phi,\sin\phi)$. Then since following Jacobi-Anger expansion holds uniformly (see \cite{CK}),
  \begin{equation}\label{JAExp}
    e^{iz\cos\phi}=J_0(z)+2\sum_{n=1}^{\infty}i^nJ_n(z)\cos(n\phi),
  \end{equation}
  elementary calculus yields
  \begin{align*}
    \int_{\mathbb{S}^1}e^{ik\vt\cdot\mx}d\vt&=\int_{0}^{2\pi}e^{ikr\cos(\theta-\phi)}d\theta\\
    &\approx2\pi J_0(kr)+2\sum_{n=1}^{\infty}i^nJ_n(kr)\int_{0}^{2\pi}\cos n(\theta-\phi)d\theta\\
    &=2\pi J_0(kr)+4\sum_{n=1}^{\infty}\frac{i^n}{n}J_n(kr)\cos\frac{n(2\pi-2\phi)}{2}\sin\frac{2n\pi}{2}=2\pi J_0(kr).
  \end{align*}
  Therefore,
  \[\frac{1}{2\pi}\int_{\mathbb{S}^1}e^{ik\vt\cdot\mx}d\vt=J_0(kr)=J_0(k|\mx|).\]

  Similarly, we can write
  \begin{align*}
    \int_{\mathbb{S}^1}\vt\cdot\vx e^{ik\vt\cdot\mx}d\vt&=\int_0^{2\pi}\rho\cos(\theta-\xi)e^{ikr\cos(\theta-\phi)}d\theta\\
    &\approx\int_0^{2\pi}\rho\cos(\theta-\xi)\left(J_0(kr)+2\sum_{n=1}^{\infty}i^nJ_n(kr)\cos n(\theta-\phi)\right)d\theta\\
    &=2\rho\sum_{n=1}^{\infty}i^nJ_n(kr)\int_0^{2\pi}\cos(\theta-\xi)\cos n(\theta-\phi)d\theta.
  \end{align*}
  If $n=1$ then elementary calculus yields
  \begin{align*}
    \int_0^{2\pi}\cos(\theta-\xi)\cos(\theta-\phi)d\theta&=\left[\frac{\theta}{2}\cos(\phi-\xi)+\frac14\sin(2\theta-\phi-\xi)\right]_0^{2\pi}=\pi\cos(\phi-\xi).
  \end{align*}
  If $n\geq2$ then since $\sin(1-n)\pi=\sin(1+n)\pi=0$,
  \begin{align*}
    \int_0^{2\pi}\cos(\theta-\xi)\cos n(\theta-\phi)d\theta=&\left[\frac{\sin\{(1-n)\theta+n\phi-\xi\}}{2(1-n)}
    +\frac{\sin\{(1+n)\theta-n\phi-\xi\}}{2(1+n)}\right]_0^{2\pi}\\
    =&\frac{\sin\{(1-n)\pi\}\cos\{(1-n)\pi+n\phi-\xi\}}{1-n}\\
    &+\frac{\sin\{(1+n)\pi\}\cos\{(1+n)\pi-n\phi-\xi\}}{1-n}=0.
  \end{align*}
  Therefore,
  \[\frac{1}{2\pi}\int_{\mathbb{S}^1}\vt\cdot\vx e^{ik\vt\cdot\mx}d\vt=i\rho\cos(\phi-\xi)J_1(kr)=i\frac{r\rho}{r}\cos(\phi-\xi)J_1(kr)
    =i\left(\frac{\mx}{|\mx|}\cdot\vx\right)J_1(k|\mx|).\]
  This completes the proof.
\end{proof}

\subsection{Dirichlet boundary condition (TM) case}
First, we consider the imaging functional (\ref{ImagingFunctionD}). Applying Lemma \ref{TheoremBessel}, we can obtain following results:
\begin{thm}\label{StructureImagingFunctional}
  For sufficiently large $N$ and $F$, (\ref{ImagingFunctionD}) can be written as follows:
  \begin{enumerate}
    \item If $k_F<+\infty$ then
      \begin{multline}\label{StructureD1}
        \mathcal{I}_{\mathrm{D}}(\mx;F)=\left|\sum_{m=1}^{M}\left[\frac{k_F}{k_F-k_1}\bigg(J_0(k_F|\mx-\my_m|)^2+J_1(k_F|\mx-\my_m|)^2\bigg)\right.\right.\\
        \left.\left.-\frac{k_1}{k_F-k_1}\bigg(J_0(k_1|\mx-\my_m|)^2+J_1(k_1|\mx-\my_m|)^2\bigg)+\int_{k_1}^{k_F}J_1(k|\mx-\my_m|)^2dk\right]\right|.
      \end{multline}
      Moreover, if $F=1$ then (\ref{ImagingFunctionD}) becomes
      \[\mathcal{I}_{\mathrm{D}}(\mx;1)=\sum_{m=1}^{M}J_0(k_f|\mx-\my_m|)^2.\]
  \item If $k_F\longrightarrow+\infty$ then
      \begin{equation}\label{StructureD2}
        \mathcal{I}_{\mathrm{D}}(\mx;F)=\chi(\Gamma),
      \end{equation}
      where $\chi$ denotes the characteristic function.
  \end{enumerate}
\end{thm}
\begin{proof}
  We assume that for every $f$, number of non-zero singular values $M_f$ is almost equal to $M$. Then, based on (\ref{Approx1}), (\ref{Approx2}), and (\ref{SVD}), we can observe that
  \begin{align*}
    \mathcal{I}_{\mathrm{D}}(\mx;F)\approx&\frac{1}{F}\abs{\sum_{f=1}^{F}\sum_{m=1}^{M}  \left(\hat{\mS}_{\mathrm{D}}(\mx;k_f)^*\overline{\hat{\mS}_{\mathrm{D}}(\my_m;k_f)}\right)\left(\hat{\mS}_{\mathrm{D}}(\mx;k_f)^*\overline{\hat{\mS}_{\mathrm{D}}(\my_m;k_f)}\right)}\\
    =&\frac{1}{F}\abs{\sum_{f=1}^{F}\sum_{m=1}^{M}\left(\sum_{s=1}^{N}e^{ik_f\vt_s\cdot(\mx-\my_m)}\right)\left(\sum_{t=1}^{N}e^{ik_f\vt_t\cdot(\mx-\my_m)}\right)}\\
    =&\frac{1}{4\pi^2F}\abs{\sum_{f=1}^{F}\sum_{m=1}^{M}\left(\int_{\mathbb{S}^1}e^{ik_f\vt\cdot(\mx-\my_m)}\right)^2d\vt} =\frac{1}{F}\abs{\sum_{f=1}^{F}\sum_{m=1}^{M}J_0(k_f|\mx-\my_m|)^2}\\
    =&\frac{1}{k_F-k_1}\abs{\sum_{m=1}^{M}\int_{k_1}^{k_F}J_0(k|\mx-\my_m|)^2dk}.
  \end{align*}
  \begin{enumerate}
    \item Since $k_F<+\infty$, based on following indefinite integral (see \cite[Page 35]{R})
      \[\int J_0(t)^2dt=t\bigg(J_0(t)^2+J_1(t)^2\bigg)+\int J_1(t)^2dt\]
      with change of variable $t=k|\mx-\my_m|$ yields
      \begin{multline*}
        \mathcal{I}_{\mathrm{D}}(\mx;F)\approx\frac{1}{k_F-k_1}\left|\sum_{m=1}^{M}\left[k_F\bigg(J_0(k_F|\mx-\my_m|)^2+J_1(k_F|\mx-\my_m|)^2\bigg)\right.\right.\\
        \left.\left.-k_1\bigg(J_0(k_1|\mx-\my_m|)^2+J_1(k_1|\mx-\my_m|)^2\bigg)+\int_{k_1}^{k_F}J_1(k|\mx-\my_m|)^2dk\right]\right|.
      \end{multline*}
      Note that this proof can be found in \cite{JKHP}.
    \item If $\mx=\my_m$ then it is clear that $\mathcal{I}_{\mathrm{D}}(\mx;F)=1$. Suppose that $\mx\ne\my_m$ then
    \begin{multline*}
      \lim_{k_F\to+\infty}\int_{k_1}^{k_F}J_0(k|\mx-\my_m|)^2dk\\
      \approx\lim_{k_F\to+\infty}\left[k_F\bigg(J_0(k_F|\mx-\my_m|)^2+J_1(k_F|\mx-\my_m|)^2\bigg)+\int_{k_1}^{k_F}J_1(k|\mx-\my_m|)^2dk\right].
    \end{multline*}
    Since following asymptotic form holds for $k|\mx-\my_m|\gg|n^2-0.25|$,
    \begin{equation}\label{ApproximationBesselFunction}
      J_{n}(k|\mx-\my_m|)\approx\sqrt{\frac{2}{k\pi|\mx-\my_m|}}\cos\left(k|\mx-\my_m|-\frac{n\pi}{2}-\frac{\pi}{4} +O\left(\frac{1}{k|\mx-\my_m|}\right)\right),
    \end{equation}
    we can observe that
    \begin{align*}
      \lim_{k_F\to+\infty}\frac{1}{k_F-k_1}&\bigg(k_FJ_0(k_F|\mx-\my_m|)^2\bigg)\\
      &\approx\lim_{k_F\to+\infty}\frac{1}{k_F-k_1}\left[\frac{2}{\pi|\mx-\my_m|}\cos^2\left(k_F|\mx-\my_m|-\frac{\pi}{4}\right)\right]=0,\\
      \lim_{k_F\to+\infty}\frac{1}{k_F-k_1}&\bigg(k_FJ_1(k_F|\mx-\my_m|)^2\bigg)\\
      &\approx\lim_{k_F\to+\infty}\frac{1}{k_F-k_1}\left[\frac{2}{\pi|\mx-\my_m|}\cos^2\left(k_F|\mx-\my_m|-\frac{3\pi}{4}\right)\right]=0,\\
      \lim_{k_F\to+\infty}\frac{1}{k_F-k_1}&\int_{k_1}^{k_F}J_1(k|\mx-\my_m|)^2dk\\
      &\approx\lim_{k_F\to+\infty}\frac{1}{k_F-k_1}\int_{k_1}^{k_F}\left[\frac{2}{k\pi|\mx-\my_m|}\cos^2\left(k_F|\mx-\my_m|-\frac{3\pi}{4}\right)\right]dk\\
      &\leq\frac{2}{\pi|\mx-\my_m|}\lim_{k_F\to+\infty}\frac{\ln k_F-\ln k_1}{k_F-k_1}=0.
    \end{align*}
    Hence, (\ref{StructureD2}) can be derived.
  \end{enumerate}
\end{proof}

\begin{rem}
  Note that the last term of (\ref{StructureD1}) does not significantly contribute the imaging performance because \[\mathcal{I}_{\mathrm{D}}(\mx;F)=O(1)\quad\mbox{and}\quad\frac{1}{k_F-k_1}\int_{k_1}^{k_F}J_1(k|\mx-\my_m|)^2dk\ll O(1).\]
  Hence, Theorem \ref{StructureImagingFunctional} tells us that multi-frequency subspace migration (i.e., map of $\mathcal{I}_{\mathrm{D}}(\mx;F)$) yields better images owing to less oscillation than single-frequency subspace migration (i.e., map of $\mathcal{I}_{\mathrm{D}}(\mx;1)$) does. A detailed discussion can be found in \cite{JKHP}.
\end{rem}

\begin{rem}[Kirchhoff migration]
  Based on \cite{AGKPS}, classical Kirchhoff migration at single frequency can be written as
  \[\mathcal{I}_{\mathrm{KM}}(\mx):=\hat{\mS}_{\mathrm{D}}(\mx;k)^*\mathbb{K}(k)\overline{\hat{\mS}_{\mathrm{D}}(\mx;k)}.\]
  Then it can be represented as
  \begin{align*}
    \mathcal{I}_{\mathrm{KM}}(\mx)=&\hat{\mS}_{\mathrm{D}}(\mx;k)^*\left(\sum_{m=1}^{N}\sigma_m(k)\mU_m(k)\mV_m^*(k)\right)\overline{\hat{\mS}_{\mathrm{D}}(\mx;k)}\\
    \approx&\sum_{m=1}^{M}\sigma_m(k)\left(\hat{\mS}_{\mathrm{D}}(\mx;k)^*\overline{\hat{\mS}_{\mathrm{D}}(\my_m;k)}\right)\left(\hat{\mS}_{\mathrm{D}}(\mx;k)^*\overline{\hat{\mS}_{\mathrm{D}}(\my_m;k)}\right)\\
    &+\sum_{m=M+1}^{N}\sigma_m(k)\left(\hat{\mS}_{\mathrm{D}}(\mx;k)^*\overline{\hat{\mS}_{\mathrm{D}}(\my_m;k)}\right)\left(\hat{\mS}_{\mathrm{D}}(\mx;k)^*\overline{\hat{\mS}_{\mathrm{D}}(\my_m;k)}\right)\\
    =&\sum_{m=1}^{M}\sigma_m(k)J_0(k|\mx-\my_m|)^2+\sum_{m=M+1}^{N}\sigma_m(k)\bigg(1-J_0(k|\mx-\my_m|)^2\bigg),
  \end{align*}
  Hence, we can observe that $\mathcal{I}_{\mathrm{KM}}(\mx)$ plots magnitude $\sigma_m(k)$, $m=1,2,\cdots,M$, at the location $\my_m\in\Gamma$. However, it also produces (unexpected) magnitude $\sigma_m(k)(1-J_0(k|\mx-\my_m|)^2)$ at $\my_m\in\mathbb{R}^2\backslash\overline{\Gamma}$. This is the reason why subspace migration is an improved version of Kirchhoff migration. Multi-frequency Kirchhoff migration can be treated by a similar manner.
\end{rem}

Let us denote $N_{inc}$ and $N_{obs}$ be the number of incident and observation directions. At this moment, we assume that only $N_{obs}$ is sufficiently large. Then, we can obtain following result.
\begin{thm}\label{StructureImagingFunctionalS}
  For sufficiently large $N_{obs}$ and $k_F$, (\ref{ImagingFunctionD}) can be written as follows:
  \begin{enumerate}
    \item If $k_F<+\infty$ then
    \begin{align}
    \begin{aligned}\label{StructureDS1}
      \mathcal{I}_{\mathrm{D}}(\mx;F)\simeq&\left|\sum_{m=1}^{M}\sum_{s=1}^{N_{inc}}\left[\frac{k_F}{k_F-k_1}\bigg(J_0(k_F|\mx-\my_m|)^2+J_1(k_F|\mx-\my_m|)^2\bigg)\right.\right.\\
    &\left.\left.-\frac{k_1}{k_F-k_1}\bigg(J_0(k_1|\mx-\my_m|)^2+J_1(k_1|\mx-\my_m|)^2\bigg)+\Lambda_1(k_F,k_1,|\mx-\my_m|;\vt_s)\right]\right|.
    \end{aligned}
    \end{align}
    Here, $\Lambda_1(k_F,k_1,|\mx-\my_m|;\vt_s)$ is given by
    \[\Lambda_1(k_F,k_1,|\mx-\my_m|;\vt_s):=\frac{2}{k_F-k_1}\sum_{n=1}^{\infty}i^n\cos(n\hat{\theta}_s)\int_{k_1}^{k_F}J_0(k|\mx-\my_m|)J_n(k|\mx-\my_m|)dk\ll O(1),\]
    where
    \begin{equation}\label{thetahat}
      \hat{\theta}_s=\cos^{-1}\left(\frac{\vt_s\cdot(\mx-\my_m)}{|\vt_s\cdot(\mx-\my_m)|}\right).
    \end{equation}
    \item If $k_F\longrightarrow+\infty$ then
    \begin{equation}\label{StructureDS2}
      \mathcal{I}_{\mathrm{D}}(\mx;F)\simeq\sum_{m=1}^{M}\sum_{s=1}^{N_{inc}}\frac{1}{\sqrt{ |\mx-\my_m|^2-(\vt_s\cdot(\mx-\my_m))^2}},
  \end{equation}
  \end{enumerate}
  where $A\simeq B$ means that there exists a constant $C$ such that $A=BC$.
\end{thm}
\begin{proof}
  Same as above, we assume that $M_f$ is almost equal to $M$ for every $f=1,2,\cdots,F$. Then, (\ref{ImagingFunctionD}) becomes
  \begin{align*}
    \mathcal{I}_{\mathrm{D}}(\mx;F)\approx&\frac{1}{F}\abs{\sum_{f=1}^{F}\sum_{m=1}^{M}  \left(\hat{\mS}_{\mathrm{D}}(\mx;k_f)^*\overline{\hat{\mS}_{\mathrm{D}}(\my_m;k_f)}\right)\left(\hat{\mS}_{\mathrm{D}}(\mx;k_f)^*\overline{\hat{\mS}_{\mathrm{D}}(\my_m;k_f)}\right)}\\
    =&\frac{1}{F}\abs{\sum_{f=1}^{F}\sum_{m=1}^{M}\left(\sum_{s=1}^{N_{inc}}e^{ik_f\vt_s\cdot(\mx-\my_m)}\right)\left(\sum_{t=1}^{N_{obs}}e^{ik_f\vt_t\cdot(\mx-\my_m)}\right)}\\
    =&\frac{1}{F}\abs{\sum_{f=1}^{F}\sum_{m=1}^{M}\sum_{s=1}^{N_{inc}}e^{ik_f\vt_s\cdot(\mx-\my_m)}J_0(k_f|\mx-\my_m|)}\\
    =&\frac{1}{k_F-k_1}\abs{\sum_{m=1}^{M}\sum_{s=1}^{N_{inc}}\int_{k_1}^{k_F}e^{ik\vt_s\cdot(\mx-\my_m)}J_0(k|\mx-\my_m|)dk}.
  \end{align*}
  Note that if $\mx=\my_m$, $\mathcal{I}_{\mathrm{D}}(\mx;F)\approx1$. Hence, we assume that $\mx\ne\my_m$.
  \begin{enumerate}
    \item Similar to the proof of Theorem \ref{TheoremBessel}, we let $\vt_s:=(\cos\theta_s,\sin\theta_s)$, $\mx-\my_m=r_m(\cos\phi,\sin\phi)$. Then applying Jacobi-Anger expansion (\ref{JAExp}), we can write
        \begin{align}
          \int_{k_1}^{k_F}e^{ik\vt_s\cdot(\mx-\my_m)}&J_0(k|\mx-\my_m|)dk=\int_{k_1}^{k_F}J_0(k|\mx-\my_m|)^2dk\label{StructureDS1-1}\\
          &+2\sum_{n=1}^{\infty}i^n\cos(n\hat{\theta}_s)\int_{k_1}^{k_F}J_0(k|\mx-\my_m|)J_n(k|\mx-\my_m|)dk\label{StructureDS1-2}.
        \end{align}
        Note that integration of (\ref{StructureDS1-1}) is derived in Theorem \ref{StructureImagingFunctional} hence we consider (\ref{StructureDS1-2}). Since all terms of (\ref{StructureDS1-1}) are convergent, (\ref{StructureDS1-2}) converges. For $x\in\mathbb{R}$, since following relation holds
        \begin{equation}\label{BoundednessofBesselFunction}
          J_n(x)\leq\frac{|x|^n}{2^n n!},
        \end{equation}
        applying H\"{o}lder's inequality, we can obtain
        \[\int_{k_1}^{k_F}J_0(k|\mx-\my_m|)J_n(k|\mx-\my_m|)dk\leq\frac{(k_F^{n+1}-k_1^{n+1})|\mx-\my_m|^n}{2^n(n+1)!}.\]
        Since
        \[\lim_{n\to\infty}i^n\cos(n\hat{\theta}_s)\int_{k_1}^{k_F}J_0(k|\mx-\my_m|)J_n(k|\mx-\my_m|)dk=0,\]
        for sufficiently large number $\cL$, (\ref{StructureDS1-2}) can be represented as
        \[2\sum_{n=1}^{\cL}i^n\cos(n\hat{\theta}_s)\int_{k_1}^{k_F}J_0(k|\mx-\my_m|)J_n(k|\mx-\my_m|)dk.\]
        Now, we assume that $\mx$ is close enough to $\my_m$ such that $k_F|\mx-\my_m|\ll\sqrt{\cL+1}$ then since
        \[\int_{k_1}^{k_F}J_0(k|\mx-\my_m|)J_n(k|\mx-\my_m|)dk\leq\frac{(k_F^{n+1}-k_1^{n+1})|\mx-\my_m|^n}{2^n(n+1)!}\ll\frac{k_F\sqrt{\cL+1}}{2^n(n+1)!},\]
        we can conclude that
        \begin{equation}\label{Boundednesso1}
          \left|\frac{2}{k_F-k_1}\sum_{n=1}^{\infty}i^n\cos(n\hat{\theta}_s)\int_{k_1}^{k_F}J_0(k|\mx-\my_m|)J_n(k|\mx-\my_m|)dk\right|\ll O(1).
        \end{equation}
        Suppose that $\mx$ is located away from $\my_m$ such that $k_F|\mx-\my_m|\gg|\cL^2-0.25|$, then for $n=1,2,\cdots,L$, applying asymptotic form (\ref{ApproximationBesselFunction}) yields
        \begin{multline*}
          \int_{k_1}^{k_F}J_0(k|\mx-\my_m|)J_{n}(k|\mx-\my_m|)dk \leq\int_{k_1}^{k_F}\sqrt{\frac{2}{k\pi|\mx-\my_m|}}\cos\left(k|\mx-\my_m|-\frac{n\pi}{2}-\frac{\pi}{4}\right)dk\\
          \leq\sqrt{\frac{2}{\pi|\mx-\my_m|}}(\sqrt{k_F}-\sqrt{k_1})\leq k_F\sqrt{\frac{2}{\pi k_F|\mx-\my_m|}}\ll k_F\sqrt{\frac{2}{\pi|\cL^2-0.25|}},
        \end{multline*}
        we can obtain (\ref{Boundednesso1}). Hence, (\ref{StructureDS1}) is derived.
    \item Since $\vt_s\in\mathbb{S}^1$, we can observe following relation
      \begin{equation}\label{VecRelation}
        |\mx-\my_m|^2-(\vt_s\cdot(\mx-\my_m))^2=|\mx-\my_m|^2\left[1-\left(\vt_s\cdot\frac{\mx-\my_m}{|\mx-\my_m|}\right)^2\right]\geq0.
      \end{equation}
      Applying (\ref{VecRelation}) with the following identity (see \cite{GR})
      \begin{equation}\label{BesselRelation}
        \int_0^\infty e^{iat}J_\nu(bt)dt=\frac{1}{\sqrt{b^2-a^2}}\bigg[\cos\left(\nu\sin^{-1}\frac{a}{b}\right)+i\sin\left(\nu\sin^{-1}\frac{a}{b}\right)\bigg]\quad\mbox{for}\quad a<b,
      \end{equation}
      we can evaluate following
      \begin{align*}
       \int_{k_1}^{k_F}e^{ik\vt_s\cdot(\mx-\my_m)}J_0(k|\mx-\my_m|)dk&\approx\int_0^{\infty}e^{ik\vt_s\cdot(\mx-\my_m)}J_0(k|\mx-\my_m|)dk\\
       &=\frac{1}{|\mx-\my_m|\sqrt{\displaystyle 1-\left(\vt_s\cdot\frac{\mx-\my_m}{|\mx-\my_m|}\right)^2}}
      \end{align*}
      Hence, (\ref{StructureDS2}) is obtained.
  \end{enumerate}
\end{proof}

\begin{rem}
  Theorems \ref{StructureImagingFunctional} and \ref{StructureImagingFunctionalS} tell us some properties of (\ref{ImagingFunctionD}) summarized as follows:

  \begin{enumerate}\renewcommand{\theenumi}{(D\arabic{enumi})}
    \item Application of multiple frequencies should guarantees an accurate shape of $\Gamma$ via (\ref{ImagingFunctionD}). Note that this fact can be identified via Statistical Hypothesis Testing \cite{AGKPS}.
    \item Based on (\ref{StructureDS2}), (\ref{ImagingFunctionD}) plots a large magnitude at $\mx$ satisfying
        \[\mx=\my_m\in\Gamma\quad\mbox{and}\quad\vt_s=\pm\frac{\mx-\my_m}{|\mx-\my_m|}.\]
        This means that (\ref{ImagingFunctionD}) produces not only the shape of $\Gamma$ but also unexpected ghost replicas. Hence, sufficiently large number of $N_{inc}$ and $N_{obs}$ is required for producing a good result.
  \end{enumerate}
\end{rem}

\subsubsection{Weighted multi-frequency imaging: an improvement}\label{subsec}
At this moment, we consider the following multi-frequency subspace migration imaging functional weighted by given frequency:
\begin{equation}\label{ImagingFunctionW}
\mathcal{I}_{\mathrm{W}}(\mx;F,p)=\frac{1}{F}\abs{\sum_{f=1}^{F}\sum_{m=1}^{M_f} (k_f)^p\left(\hat{\mS}_{\mathrm{D}}(\mx;k_f)^*\mU_m(k_f)\right)\left(\hat{\mS}_{\mathrm{D}}(\mx;k_f)^*\overline{\mV}_m(k_f)\right)},
\end{equation}
where $p$ is a positive integer. Based on recent work \cite{P3}, (\ref{ImagingFunctionW}) is an improved version of (\ref{ImagingFunctionD}) when $p=1$. We briefly introduce the structure of (\ref{ImagingFunctionW}) as follows.

\begin{thm}\label{TheoremStructureW1}
  For sufficiently large $N$ and $k_F$, (\ref{ImagingFunctionW}) can be written
  \begin{multline*}
      \mathcal{I}_{\mathrm{W}}(\mx;F,1)=\frac{1}{k_F-k_1}\left|\sum_{m=1}^{M}\frac{(k_F)^2}{2}\bigg(J_0(k_F|\mx-\my_m|)^2+J_1(k_F|\mx-\my_m|)^2\bigg)\right.\\
      -\left.\frac{(k_1)^2}{2}\bigg(J_0(k_1|\mx-\my_m|)^2+J_1(k_1|\mx-\my_m|)^2\bigg)\right|.
  \end{multline*}
\end{thm}
\begin{proof}
  See \cite[Theorem 4]{P3}.
\end{proof}

\begin{thm}\label{TheoremStructureW2}
  For sufficiently large $N_{obs}$ and $k_F$, (\ref{ImagingFunctionW}) can be written
    \begin{enumerate}
      \item For $k_F<+\infty$,
        \begin{multline*}
          \mathcal{I}_{\mathrm{W}}(\mx;F,1)=\frac{1}{k_F-k_1}\left|\sum_{m=1}^{M}\frac{(k_F)^2}{2}\bigg(J_0(k_F|\mx-\my_m|)^2+J_1(k_F|\mx-\my_m|)^2\bigg)\right.\\
          -\left.\frac{(k_1)^2}{2}\bigg(J_0(k_1|\mx-\my_m|)^2+J_1(k_1|\mx-\my_m|)^2\bigg)+\Lambda_5(k_F,k_1,|\mx-\my_m|;\vt_s)\right|,
        \end{multline*}
        where $\Lambda_5(k_F,k_1,|\mx-\my_m|;\vt_s)$ is represented as follows
        \[\Lambda_5(k_F,k_1,|\mx-\my_m|;\vt_s)=2\sum_{n=1}^{\infty}i^n\cos(n\hat{\theta}_s)\int_{k_1}^{k_F}kJ_0(k|\mx-\my_m|)J_n(k|\mx-\my_m|)dk.\]
        Here, $\hat{\theta}_s$ is given by (\ref{thetahat}).
      \item For $k_F\longrightarrow+\infty$,
        \begin{align*}
          \mathcal{I}_{\mathrm{W}}(\mx;F,1)=\left|\sum_{m=1}^{M}\sum_{s=1}^{N_{inc}}\frac{\vt_s\cdot(\mx-\my_m)}{\left(\displaystyle |\mx-\my_m|^2-(\vt_s\cdot(\mx-\my_m))^2\right)^{3/2}}\right|.
        \end{align*}
    \end{enumerate}
\end{thm}
\begin{proof}Similar to the proof of Theorem \ref{StructureImagingFunctionalS}, (\ref{ImagingFunctionW}) can be written as
\[\mathcal{I}_{\mathrm{W}}(\mx;F,1)=\frac{1}{k_F-k_1}\abs{\sum_{m=1}^{M}\sum_{s=1}^{N_{inc}}\int_{k_1}^{k_F}ke^{ik\vt_s\cdot(\mx-\my_m)}J_0(k|\mx-\my_m|)dk}.\]
\begin{enumerate}
  \item Borrowing the polar coordinate in the proof of Theorem \ref{StructureImagingFunctionalS} and Jacobi-Anger expansion (\ref{JAExp}), we can write
      \begin{multline*}
        \int_{k_1}^{k_F}ke^{ik\vt_s\cdot(\mx-\my_m)}J_0(k|\mx-\my_m|)dk\\
        =\int_{k_1}^{k_F}kJ_0(k|\mx-\my_m|)^2dk+2\sum_{n=1}^{\infty}i^n\cos(n\hat{\theta}_s)\int_{k_1}^{k_F}kJ_0(k|\mx-\my_m|)J_n(k|\mx-\my_m|)dk.
      \end{multline*}
      Following well-known indefinite integral (see \cite{AS})
      \[\int_0^xtJ_0(t)^2dt=\frac{x^2}{2}\bigg(J_0(x)^2+J_1(x)^2\bigg)\]
      yields
      \begin{multline*}
        \int_{k_1}^{k_F}kJ_0(k|\mx-\my_m|)^2dk=\frac{(k_F)^2}{2}\bigg(J_0(k_F|\mx-\my_m|)^2+J_1(k_F|\mx-\my_m|)^2\bigg)\\
        -\frac{(k_1)^2}{2}\bigg(J_0(k_1|\mx-\my_m|)^2+J_1(k_1|\mx-\my_m|)^2\bigg).
      \end{multline*}
      And similar to the proof of Theorem \ref{StructureImagingFunctionalS}, it is easy to observe that $\mathcal{I}_{\mathrm{W}}(\mx;F,1)=O(k_F)$ and
      \[\frac{1}{k_F-k_1}\Lambda_5(k_F,k_1,|\mx-\my_m|;\vt_s)\ll O(k_F).\]
  \item Let $k_F\longrightarrow+\infty$. Since following identify holds for $b>0$, $p>-q-2$ (see \cite[Formula 6.621-4]{GR})
  \[\int_0^\infty x^{q+1}e^{-ax}J_p(bx)=(-1)^{q+1}b^{-p}\frac{d^{q+1}}{da^{q+1}}\left[\frac{(\sqrt{a^2+b^2}-a)^p}{\sqrt{a^2+b^2}}\right],\]
  elementary calculus yields
  \[\int_0^{\infty}xe^{-ax}J_0(bx)=\frac{a}{(a^2+b^2)^{3/2}}.\]
  Hence, we can obtain
  \[\int_0^{\infty}ke^{ik\vt_s\cdot(\mx-\my_m)}J_0(k|\mx-\my_m|)dk=\frac{i\vt_s\cdot(\mx-\my_m)}{|\mx-\my_m|^3\left[\displaystyle 1-\left(\vt_s\cdot\frac{\mx-\my_m}{|\mx-\my_m|}\right)^2\right]^{3/2}}.\]
\end{enumerate}
\end{proof}

\subsection{Neumann boundary condition (TE) case}
Next, we consider the imaging functional (\ref{ImagingFunctionN}). It is worth mentioning that if we have \textit{a priori} information of $\Gamma$ (specially, $\vn(\mx)$ at $\mx\in\Gamma$) we can obtain same result in Theorem \ref{StructureImagingFunctional}. However, due to the fact that since the normal direction to $\Gamma$ is unknown, mapping of (\ref{ImagingFunctionN}) consumes large computational costs. Unfortunately, based on the results in section \ref{Sec4-3}, the results are still poor. Hence, we consider the following alternative subspace migration imaging functional:
\begin{equation}\label{ImagingFunctionA}
  \mathcal{I}_{\mathrm{A}}(\mx;F)=\frac{1}{F}\abs{\sum_{f=1}^{F}\sum_{m=1}^{M_f}\left(\hat{\mS}_{\mathrm{D}}(\mx;k_f)^*\mU_m(k_f)\right)\left(\hat{\mS}_{\mathrm{D}}(\mx;k_f)^*\overline{\mV}_m(k_f)\right)},
\end{equation}
where $\mS_{\mathrm{D}}(\mx;k_f)$ is defined in (\ref{VecD}). Then we can obtain following result.

\begin{thm}\label{StructureImagingFunctionalA}
  For sufficiently large $N$ and $F$, (\ref{ImagingFunctionA}) can be written as follows:
  \[\mathcal{I}_{\mathrm{A}}(\mx;F)=\frac{1}{k_F-k_1}\abs{\sum_{m=1}^{M}\int_{k_1}^{k_F}\left[\left(\frac{\mx-\my_m}{|\mx-\my_m|}\cdot\vn(\my_m)\right)J_1(k|\mx-\my_m|)\right]^2dk}.\]
  Furthermore,
  \begin{enumerate}
    \item If $\mx$ is close enough to $\my_m$, then
      \begin{equation}\label{PropertyIA}
        \mathcal{I}_{\mathrm{A}}(\mx;F)=\frac{(k_F)^3-(k_1)^3}{12(k_F-k_1)}\abs{\sum_{m=1}^{M}\bigg((\mx-\my_m)\cdot\vn(\my_m)\bigg)^2}.
      \end{equation}
    \item If $\mx$ is far away from $\my_m$, then
      \[\mathcal{I}_{\mathrm{A}}(\mx;F)\leq\frac{2}{\pi(k_F-k_1)}\abs{\sum_{m=1}^{M}\frac{((\mx-\my_m)\cdot\vn(\my_m))^2}{\sqrt{2k_F}|\mx-\my_m|^4}}.\]
  \end{enumerate}
\end{thm}
\begin{proof}
  Same as the proof of Theorem \ref{StructureImagingFunctionalS}, we assume that for every $f$, number of non-zero singular values $M_f$ is almost equal to $M$. Then
  \begin{align*}
    \mathcal{I}_{\mathrm{A}}(\mx;F)&=\frac{1}{F}\abs{\sum_{f=1}^{F}\sum_{m=1}^{M_f}\left(\hat{\mS}_{\mathrm{D}}(\mx;k_f)^*\mU_m(k_f)\right)\left(\hat{\mS}_{\mathrm{D}}(\mx;k_f)^*\overline{\mV}_m(k_f)\right)}\\
    &=\frac{1}{F}\abs{\sum_{f=1}^{F}\sum_{m=1}^{M}\left(\sum_{s=1}^{N}\vt_s\cdot\vn(\my_m)e^{ik_f\vt_s\cdot(\mx-\my_m)}\right)\left(\sum_{t=1}^{N}\vt_t\cdot\vn(\my_m)e^{ik_f\vt_t\cdot(\mx-\my_m)}\right)}\\
    &=\frac{1}{4\pi^2F}\abs{\sum_{f=1}^{F}\sum_{m=1}^{M}\left(\int_{\mathbb{S}^1}\vt\cdot\vn(\my_m)e^{ik_f\vt\cdot(\mx-\my_m)}d\vt\right)^2}\\
    &=\frac{1}{F}\abs{\sum_{f=1}^{F}\sum_{m=1}^{M}\bigg[\left(\frac{\mx-\my_m}{|\mx-\my_m|}\cdot\vn(\my_m)\right)J_1(k_f|\mx-\my_m|)\bigg]^2}\\
    &=\frac{1}{k_F-k_1}\abs{\sum_{m=1}^{M}\int_{k_1}^{k_F}\left[\left(\frac{\mx-\my_m}{|\mx-\my_m|}\cdot\vn(\my_m)\right)J_1(k|\mx-\my_m|)\right]^2dk}.
  \end{align*}
  Unfortunately, there is no finite representation of the integral $\int J_1(x)^2 dx$. Therefore we cannot go further. In order to observe some properties of (\ref{ImagingFunctionA}), we consider the following two cases.
  \begin{enumerate}
    \item Assume that $\mx$ is close to $\my_m$ such that $0<k|\mx-\my_m|\ll\sqrt2$. Then applying asymptotic form of Bessel function
    \[J_1(k|\mx-\my_m|)\approx\frac{k|\mx-\my_m|}{2}\quad\mbox{for}\quad0<|\mx-\my_m|\ll\frac{\sqrt{2}}{k},\]
    we can observe that
    \begin{align*}
      \mathcal{I}_{\mathrm{A}}(\mx;F)&=\frac{1}{k_F-k_1}\abs{\sum_{m=1}^{M}\int_{k_1}^{k_F}\left[\left(\frac{\mx-\my_m}{|\mx-\my_m|}\cdot\vn(\my_m)\right)J_1(k|\mx-\my_m|)\right]^2dk}\\
      &=\frac{1}{k_F-k_1}\abs{\sum_{m=1}^{M}\bigg((\mx-\my_m)\cdot\vn(\my_m)\bigg)^2\int_{k_1}^{k_F}k^2dk}\\
      &=\frac{(k_F)^3-(k_1)^3}{12(k_F-k_1)}\abs{\sum_{m=1}^{M}\bigg((\mx-\my_m)\cdot\vn(\my_m)\bigg)^2}.
    \end{align*}
    \item Assume that $\mx$ is far away from $\my_m$ such that $k|\mx-\my_m|\gg|1-0.25|$. Then since (\ref{ApproximationBesselFunction}) can be approximated as follows
     \[J_1(k|\mx-\my_m|)\approx\sqrt{\frac{2}{k\pi|\mx-\my_m|}}\cos\left(k|\mx-\my_m|-\frac{3\pi}{4}\right)\quad\mbox{for}\quad |\mx-\my_m|\gg\frac{3}{4k},\]
    we can observe that
    \begin{align*}
      \mathcal{I}_{\mathrm{A}}(\mx;F)&=\frac{1}{k_F-k_1}\abs{\sum_{m=1}^{M}\int_{k_1}^{k_F}\left[\left(\frac{\mx-\my_m}{|\mx-\my_m|}\cdot\vn(\my_m)\right)J_1(k|\mx-\my_m|)\right]^2dk}\\
      &=\frac{2}{\pi(k_F-k_1)}\abs{\sum_{m=1}^{M}\frac{((\mx-\my_m)\cdot\vn(\my_m))^2}{|\mx-\my_m|^3}\int_{k_1}^{k_F}\frac{1}{\sqrt{k}}\cos\left(k|\mx-\my_m|-\frac{3\pi}{4}\right)dk}.
    \end{align*}
    Let $k|\mx-\my_m|=t$. Then since $t$ is sufficiently large, based on following asymptotic behavior (see \cite{GR})
    \begin{align*}
      \frac{1}{\sqrt{2\pi}}\int_0^x\frac{\sin t}{\sqrt{t}}dt&=\frac12-\frac{1}{\sqrt{2\pi x}}\cos x+O\left(\frac{1}{x}\right)\\
      \frac{1}{\sqrt{2\pi}}\int_0^x\frac{\cos t}{\sqrt{t}}dt&=\frac12+\frac{1}{\sqrt{2\pi x}}\sin x+O\left(\frac{1}{x}\right),
    \end{align*}
    we can obtain
    \begin{align*}
      \int_{k_1}^{k_F}&\frac{1}{\sqrt{k}}\cos\left(k|\mx-\my_m|-\frac{3\pi}{4}\right)dk=\frac{1}{\sqrt{2|\mx-\my_m|}}\int_{k_1|\mx-\my_m|}^{k_F|\mx-\my_m|}\frac{\sin t-\cos t}{\sqrt{t}}dt\\
      &=\frac{\cos(k_1|\mx-\my_m|)+\sin(k_1|\mx-\my_m|)}{\sqrt{2k_1}|\mx-\my_m|}-\frac{\cos(k_F|\mx-\my_m|)+\sin(k_F|\mx-\my_m|)}{\sqrt{2k_F}|\mx-\my_m|}.
    \end{align*}
    Hence,
    \[\mathcal{I}_{\mathrm{A}}(\mx;F)\leq\frac{2}{\pi(k_F-k_1)}\abs{\sum_{m=1}^{M}\frac{((\mx-\my_m)\cdot\vn(\my_m))^2}{\sqrt{2k_F}|\mx-\my_m|^4}}.\]
  \end{enumerate}
\end{proof}

Above result tells us that $\mathcal{I}_{\mathrm{A}}(\mx;F)=0$ at $\mx=\my_m\in\Gamma$ so that $\mathcal{I}_{\mathrm{A}}(\mx;F)$ should plots $0$ (or small values) along the crack(s). Moreover, based on (\ref{PropertyIA}), map of $\mathcal{I}_{\mathrm{A}}(\mx;F)$ gives two curves in the neighborhood of true crack(s). This means that although $\mathcal{I}_{\mathrm{A}}(\mx;F)$ does not produces image of crack(s), an approximate shape of crack(s) can be recognized from the images two-curves. See Figure \ref{FigIA} and various numerical examples in Section \ref{Sec4-3}.

\begin{figure}[!ht]
\begin{center}
\includegraphics[width=0.9\textwidth]{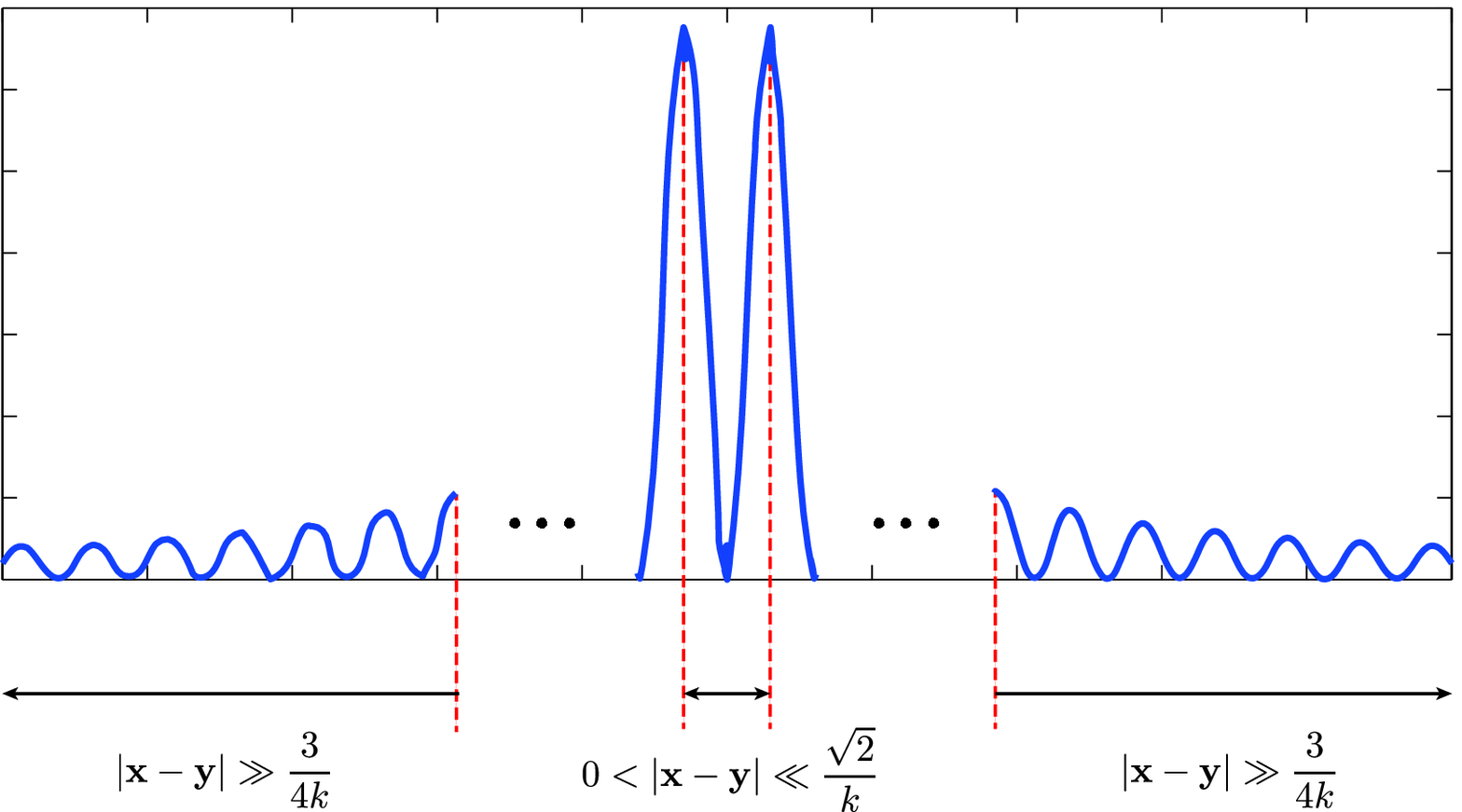}
\caption{\label{FigIA}1-D plot of $\mathcal{I}_{\mathrm{A}}(\mx;F)$.}
\end{center}
\end{figure}

Throughout a similar argument of Theorem \ref{StructureImagingFunctionalS}, we can obtain following result.
\begin{thm}\label{StructureImagingFunctionalSN}
  For sufficiently large $N_{obs}$ and $k_F$, (\ref{ImagingFunctionA}) can be written as follows:
  \begin{enumerate}
    \item If $k_F<+\infty$ then
    \begin{equation}
      \mathcal{I}_{\mathrm{A}}(\mx;F)=\frac{1}{k_F-k_1}\abs{\sum_{m=1}^{M}\sum_{s=1}^{N_{inc}}\bigg(\frac{\vt_s}{|\mx-\my_m|}\cdot\vn(\my_m)\bigg)\bigg(\frac{\mx-\my_m}{|\mx-\my_m|}\cdot\vn(\my_m)\bigg)\Lambda_2(\mx,\my_m;\vt_s)},
    \end{equation}
    where
    \begin{equation}\label{FunctionLambda2}
      \Lambda_2(\mx,\my_m;\vt_s)=\frac{1}{2}\bigg(J_0(k_1|\mx-\my_m|)^2-J_0(k_F|\mx-\my_m|)^2\bigg)+\Lambda_3(k_F,k_1,|\mx-\my_m|;\vt_s).
    \end{equation}
    Here, $\Lambda_3(k_F,k_1,|\mx-\my_m|;\vt_s)$ satisfies
    \[\Lambda_3(k_F,k_1,|\mx-\my_m|;\vt_s):=2\sum_{n=1}^{\infty}i^n\cos(n\hat{\theta}_s)\int_{k_1}^{k_F}J_1(k|\mx-\my_m|)J_n(k|\mx-\my_m|)dk,\]
    where $\hat{\theta}_s$ is given by (\ref{thetahat}).
    \item If $k_F\longrightarrow+\infty$ then
    \begin{equation}\label{StructureSN2}
      \mathcal{I}_{\mathrm{A}}(\mx;F)=\frac{1}{k_F-k_1}\abs{\sum_{m=1}^{M}\sum_{s=1}^{N_{inc}}\bigg(\vt_s\cdot\vn(\my_m)\bigg)\bigg(\frac{\mx-\my_m}{|\mx-\my_m|}\cdot\vn(\my_m)\bigg)\Lambda_4(\mx,\my_m;\vt_s)},
    \end{equation}
    where $\Lambda_4(\mx,\my_m;\vt_s)$ is
    \[\Lambda_4(\mx,\my_m;\vt_s)=\frac{1}{|\mx-\my_m|}\left(1+i\frac{\vt_s\cdot(\mx-\my_m)}{\sqrt{|\mx-\my_m|^2-(\vt_s\cdot(\mx-\my_m))^2}}\right).\]
  \end{enumerate}
\end{thm}
\begin{proof}
  Assume that $M_f$ is almost equal to $M$ for every $f=1,2,\cdots,F$. Then, (\ref{ImagingFunctionA}) becomes
  \begin{align*}
    \mathcal{I}_{\mathrm{A}}(\mx;F)&=\frac{1}{F}\abs{\sum_{f=1}^{F}\sum_{m=1}^{M_f}\left(\hat{\mS}_{\mathrm{D}}(\mx;k_f)^*\mU_m(k_f)\right)\left(\hat{\mS}_{\mathrm{D}}(\mx;k_f)^*\overline{\mV}_m(k_f)\right)}\\
    &=\frac{1}{F}\abs{\sum_{f=1}^{F}\sum_{m=1}^{M}\left(\sum_{s=1}^{N_{inc}}\vt_s\cdot\vn(\my_m)e^{ik_f\vt_s\cdot(\mx-\my_m)}\right)\left(\sum_{t=1}^{N_{obs}}\vt_t\cdot\vn(\my_m)e^{ik_f\vt_t\cdot(\mx-\my_m)}\right)}\\
    &=\frac{1}{2\pi F}\abs{\sum_{f=1}^{F}\sum_{m=1}^{M}\sum_{s=1}^{N_{inc}}\vt_s\cdot\vn(\my_m)e^{ik_f\vt_s\cdot(\mx-\my_m)}\int_{\mathbb{S}^1}\vt\cdot\vn(\my_m)e^{ik_f\vt\cdot(\mx-\my_m)}d\vt}\\
    &=\frac{1}{F}\abs{\sum_{f=1}^{F}\sum_{m=1}^{M}\sum_{s=1}^{N_{inc}}\bigg(\vt_s\cdot\vn(\my_m)\bigg)\bigg(\frac{\mx-\my_m}{|\mx-\my_m|}\cdot\vn(\my_m)\bigg)e^{ik_f\vt_s\cdot(\mx-\my_m)}J_1(k_f|\mx-\my_m|)}\\
    &=\abs{\sum_{m=1}^{M}\sum_{s=1}^{N_{inc}}\left(\frac{\vt_s\cdot\vn(\my_m)}{k_F-k_1}\right)\left(\frac{\mx-\my_m}{|\mx-\my_m|}\cdot\vn(\my_m)\right)\int_{k_1}^{k_F}e^{ik\vt_s\cdot(\mx-\my_m)}J_1(k|\mx-\my_m|)dk}.
  \end{align*}
  \begin{enumerate}
    \item Assume that $k_F<+\infty$. Considering the polar coordinate $\vt_s:=(\cos\theta_s,\sin\theta_s)$, $\mx-\my_m=r_m(\cos\phi,\sin\phi)$, and applying Jacobi-Anger expansion (\ref{JAExp}), we can write
        \begin{multline*}
          \int_{k_1}^{k_F}e^{ik\vt_s\cdot(\mx-\my_m)}J_1(k|\mx-\my_m|)dk=\int_{k_1}^{k_F}e^{ik|\mx-\my_m|\cos(\theta_s-\phi)}J_1(k|\mx-\my_m|)dk\\
          =\int_{k_1}^{k_F}J_0(k|\mx-\my_m|)J_1(k|\mx-\my_m|)dk+2\sum_{n=1}^{\infty}i^n\cos(n\hat{\theta}_s)\int_{k_1}^{k_F}J_1(k|\mx-\my_m|)J_n(k|\mx-\my_m|)dk.
        \end{multline*}
        Hence, applying well-known indefinite integral
        \begin{equation}\label{IntegralBessel01}
          \int J_0(x)J_1(x)dx=-\frac12 J_0(x)^2
        \end{equation}
        yields (\ref{FunctionLambda2}). Moreover, through the similar process of the proof of Theorem \ref{StructureImagingFunctionalS}, we can observe that the term
        \[\frac{1}{k_F-k_1}\Lambda_3(k_F,k_1,|\mx-\my_m|;\vt_s)\]
        can be negligible.
    \item Suppose that $k_F\longrightarrow+\infty$. Then, applying (\ref{BesselRelation}) yields
      \begin{align*}
        \lim_{k_F\to\infty}&\int_{k_1}^{k_F}e^{ik\vt_s\cdot(\mx-\my_m)}J_1(k|\mx-\my_m|)dk\approx\int_{0}^{\infty}e^{ik\vt_s\cdot(\mx-\my_m)}J_1(k|\mx-\my_m|)dk\\
        &=\frac{\displaystyle\cos\left(\sin^{-1}\frac{\vt_s\cdot(\mx-\my_m)}{|\mx-\my_m|}\right)+i\sin\left(\sin^{-1}\frac{\vt_s\cdot(\mx-\my_m)}{|\mx-\my_m|}\right)} {\displaystyle|\mx-\my_m|\sqrt{\displaystyle 1-\left(\vt_s\cdot\frac{\mx-\my_m}{|\mx-\my_m|}\right)^2}}.
      \end{align*}
      Therefore, from the elementary calculus
      \[\cos(\sin^{-1}\phi)=\sqrt{1-\phi^2}\quad\mbox{and}\quad\sin(\sin^{-1}\phi)=\phi,\]
      structure (\ref{StructureSN2}) can be obtained.
  \end{enumerate}
\end{proof}
The result in Theorem \ref{StructureImagingFunctionalSN} shows that when the number of incident directions are small, the property of $\mathcal{I}_{\mathrm{A}}(\mx;F)$ is similar to the one in Theorem \ref{StructureImagingFunctionalA} but due to the remaining terms (for example, $\Lambda_3(k_F,k_1,|\mx-\my_m|;\vt_s)$), produces results should be poor.

\section{Analysis of multi-frequency subspace migration imaging functionals: limited-view case}\label{Sec4}
\subsection{Common features}
We now turn our attention to the limited-view problems. We assume that the unit circle divided into the two-disjoint connected sets $\mathbb{S}^1=\mathbb{S}_+^1\cup\mathbb{S}_-^1$ and every $\vt_n$ are elements of $\mathbb{S}_+^1$ such that
\begin{equation}\label{IODir}
  \vt_n=\left(\cos\theta_n,\sin\theta_n\right),\quad\theta_n=\alpha+(\beta-\alpha)\frac{n-1}{N-1},
\end{equation}
where $0<\alpha<\beta<2\pi$. In this case, we let $\mathcal{I}_{\mathrm{L}}(\mx;F)$ be either (\ref{ImagingFunctionD}) or (\ref{ImagingFunctionA}). In order to explore the structure of $\mathcal{I}_{\mathrm{L}}(\mx;F)$, we must evaluate following integrals
\[\int_{\mathbb{S}_+^1}e^{ik\vt\cdot\mx}d\vt\quad\mbox{and}\quad\int_{\mathbb{S}_+^1}\vt\cdot\vx e^{ik\vt\cdot\mx}d\vt.\]
In our knowledge, there is no finite representation of above integrals so at this moment, we cannot conclude any properties of imaging functionals (\ref{ImagingFunctionD}) and (\ref{ImagingFunctionA}). Hence, we find approximations of above integrals and consequently discover certain properties of $\mathcal{I}_{\mathrm{L}}(\mx;F)$.

\begin{thm}\label{TheoremBesselLimited}
  Let $\mx=r(\cos\phi,\sin\phi)$ and $\vx=(\cos\xi,\sin\xi)$. Then for sufficiently large $N$, following relations holds
  \begin{align*}
    \int_{\mathbb{S}_+^1}e^{ik\vt\cdot\mx}d\vt=&(\beta-\alpha)J_0(k|\mx|)+4\sum_{n=1}^{\infty}\Lambda_{\mathrm{D}}(\alpha,\beta,k|\mx|;n)\\
    \int_{\mathbb{S}_+^1}\vt\cdot\vx e^{ik\vt\cdot\mx}d\vt=&2J_0(k|\mx|)\sin\frac{\beta-\alpha}{2}\cos\frac{\beta+\alpha-2\xi}{2}\\
    &+iJ_1(k|\mx|)\bigg[(\beta-\alpha)\bigg(\frac{\mx}{|\mx|}\cdot\vx\bigg)+\sin(\beta-\alpha)\cos(\beta+\alpha-\xi-\phi)\bigg]\\
    &+2\sum_{n=2}^{\infty}\Lambda_{\mathrm{N}}(\alpha,\beta,k|\mx|;n),
  \end{align*}
  where
  \[\Lambda_{\mathrm{D}}(\alpha,\beta,k|\mx|;n)=\frac{i^n}{n}J_n(k|\mx|)\cos\frac{n(\beta+\alpha-2\phi)}{2}\sin\frac{n(\beta-\alpha)}{2}\]
  and
  \begin{align*}
    \Lambda_{\mathrm{N}}(\alpha,\beta,k|\mx|;n)=&i^nJ_n(k_f|\mx-\my_m|)\bigg[\frac{1}{1-n}\sin\frac{(1-n)(\beta-\alpha)}{2}\cos\frac{(1-n)(\beta+\alpha)+2n\phi-2\xi}{2}\\
    &+\frac{1}{1+n}\sin\frac{(1+n)(\beta-\alpha)}{2}\cos\frac{(1+n)(\beta+\alpha)-2n\phi-2\xi}{2}\bigg].
  \end{align*}
\end{thm}
\begin{proof}
  Similar to the proof of Theorem \ref{TheoremBessel}, we let $\vt=(\cos\theta,\sin\theta)$, $\mx=r(\cos\phi,\sin\phi)$, and $\vx=(\cos\xi,\sin\xi)$. Then applying Jacobi-Anger expansion (\ref{JAExp}), we can obtain
  \begin{align*}
    \int_{\mathbb{S}_+^1}e^{ik\vt\cdot\mx}d\vt&=\int_{\alpha}^{\beta}e^{ikr\cos(\theta-\phi)}d\theta\\
    &\approx(\beta-\alpha)J_0(kr)+2\sum_{n=1}^{\infty}i^nJ_n(kr)\int_{\alpha}^{\beta}\cos n(\theta-\phi)d\theta\\ &=(\beta-\alpha)J_0(kr)+4\sum_{n=1}^{\infty}\frac{i^n}{n}J_n(kr)\cos\frac{n(\beta+\alpha-2\phi)}{2}\sin\frac{n(\beta-\alpha)}{2}\\
    &=(\beta-\alpha)J_0(k|\mx|)+4\sum_{n=1}^{\infty}\Lambda_{\mathrm{D}}(\alpha,\beta,|\mx|;n).
  \end{align*}
  Similarly, since
  \begin{align*}
    \int_{\mathbb{S}_+^1}\vt\cdot\vx e^{ik\vt\cdot\mx}d\vt&=\int_{\alpha}^{\beta}\cos(\theta-\xi)e^{ikr\cos(\theta-\phi)}d\theta\\ &=\int_{\alpha}^{\beta}\cos(\theta-\xi)\left[J_0(k|\mx|)+2\sum_{n=1}^{\infty}i^nJ_n(k|\mx|)\cos n(\theta-\phi)\right]d\theta.
  \end{align*}
  Elementary calculus yields
  \[\int_{\alpha}^{\beta}J_0(k|\mx|)\cos(\theta-\xi)d\theta=2J_0(k|\mx|)\sin\frac{\beta-\alpha}{2}\cos\frac{\beta+\alpha-2\xi}{2}.\]
  In order to evaluate remaining terms, we recall following indefinite integral (see \cite[Formula 2.532-3,6]{GR})
  \begin{align*}
    \int\cos(ax+b)\cos(ax+d)&=\frac{x\cos(b-d)}{2}+\frac{\sin(2ax+b+d)}{4a}\\
    \int\cos(ax+b)\cos(cx+d)&=\frac{\sin[(a-c)x+b-d]}{2(a-c)}+\frac{\sin[(a+c)x+b+d]}{2(a+c)}\quad\mbox{for}\quad a^2\ne c^2.
  \end{align*}
  Then
  \begin{multline*}
    \int_{\alpha}^{\beta}2iJ_1(k|\mx|)\cos(\theta-\xi)\cos(\theta-\phi)d\theta\\ =iJ_1(k|\mx|)\bigg[(\beta-\alpha)\bigg(\frac{\mx}{|\mx|}\cdot\vx\bigg)+\sin(\beta-\alpha)\cos(\beta+\alpha-\xi-\phi)\bigg]
  \end{multline*}
  and for $n\ne1$,
  \begin{align*}
    \int_{\alpha}^{\beta}&2i^nJ_n(k|\mx|)\cos(\theta-\xi)\cos n(\theta-\phi)d\theta\\ =&i^nJ_n(k_f|\mx-\my_m|)\bigg[\frac{1}{1-n}\sin\frac{(1-n)(\beta-\alpha)}{2}\cos\frac{(1-n)(\beta+\alpha)+2n\phi-2\xi}{2}\\
    &+\frac{1}{1+n}\sin\frac{(1+n)(\beta-\alpha)}{2}\cos\frac{(1+n)(\beta+\alpha)-2n\phi-2\xi}{2}\bigg].
  \end{align*}
\end{proof}

\subsection{Dirichlet boundary condition (TM) case}
First, we consider the TM case. Applying Theorem \ref{TheoremBesselLimited}, we can obtain following results.
\begin{thm}\label{StructureImagingFunctionalLimitedTM}
  For sufficiently large $N$ and $F$, (\ref{ImagingFunctionD}) can be written as follows:
  \begin{multline*}
    \mathcal{I}_{\mathrm{L}}(\mx;F)\approx\left|\sum_{m=1}^{M}\left[\frac{k_F}{k_F-k_1}\bigg(J_0(k_F|\mx-\my_m|)^2+J_1(k_F|\mx-\my_m|)^2\bigg)\right.\right.\\
        \left.\left.-\frac{k_1}{k_F-k_1}\bigg(J_0(k_1|\mx-\my_m|)^2+J_1(k_1|\mx-\my_m|)^2\bigg)\right]\right|.
  \end{multline*}
\end{thm}
\begin{proof}
Similar to the proof of Theorem \ref{StructureImagingFunctional}, we assume that for every $f$, number of non-zero singular values $M_f$ is almost equal to $M$. Then
\begin{align*}
  \mathcal{I}_{\mathrm{L}}(\mx;F)\approx&\frac{1}{F}\abs{\sum_{f=1}^{F}\sum_{m=1}^{M}  \left(\hat{\mS}_{\mathrm{D}}(\mx;k_f)^*\overline{\hat{\mS}_{\mathrm{D}}(\my_m;k_f)}\right)\left(\hat{\mS}_{\mathrm{D}}(\mx;k_f)^*\overline{\hat{\mS}_{\mathrm{D}}(\my_m;k_f)}\right)}\\
  =&\frac{1}{F}\abs{\sum_{f=1}^{F}\sum_{m=1}^{M}\left(\sum_{s=1}^{N}e^{ik_f\vt_s\cdot(\mx-\my_m)}\right)\left(\sum_{t=1}^{N}e^{ik_f\vt_t\cdot(\mx-\my_m)}\right)}\\
  =&\frac{1}{(\beta-\alpha)^2F}\abs{\sum_{f=1}^{F}\sum_{m=1}^{M}\left(\int_{\mathbb{S}_+^1}e^{ik_f\vt\cdot(\mx-\my_m)}d\vt\right)^2}\\
  =&\frac{1}{(\beta-\alpha)^2(k_F-k_1)}\abs{\sum_{m=1}^{M}\int_{k_1}^{k_F}\left(\int_{\mathbb{S}_+^1}e^{ik\vt\cdot(\mx-\my_m)}d\vt\right)^2dk}.
\end{align*}
Letting $\vt=(\cos\theta,\sin\theta)$ and $\mx=r(\cos\phi,\sin\phi)$, and applying Theorem \ref{TheoremBesselLimited}, we can evaluate the following
\begin{align*}
  \left(\int_{\mathbb{S}_+^1}e^{ik\vt\cdot(\mx-\my_m)}d\vt\right)^2 =&\left[(\beta-\alpha)J_0(k|\mx-\my_m|)+4\sum_{n=1}^{\infty}\Lambda_{\mathrm{D}}(\alpha,\beta,k|\mx-\my_m|;n)\right]^2\\
  =&(\beta-\alpha)^2J_0(k|\mx-\my_m|)^2\\
  &+8(\beta-\alpha)\sum_{n=1}^{\infty}J_0(k|\mx-\my_m|)\Lambda_{\mathrm{D}}(\alpha,\beta,k|\mx-\my_m|;n)\\
  &+16\left(\sum_{n=1}^{\infty}\Lambda_{\mathrm{D}}(\alpha,\beta,k|\mx-\my_m|;n)\right)^2.
\end{align*}
Similar to the derivation of (\ref{Boundednesso1}) in Theorem \ref{StructureImagingFunctionalS}, for sufficiently large $k$,
\begin{align*}
  &\frac{1}{k_F-k_1}\sum_{n=1}^{\infty}J_0(k|\mx-\my_m|)\Lambda_{\mathrm{D}}(\alpha,\beta,k|\mx-\my_m|;n)\ll O(1)\\
  &\frac{1}{k_F-k_1}\left(\sum_{n=1}^{\infty}\Lambda_{\mathrm{D}}(\alpha,\beta,k|\mx-\my_m|;n)\right)^2\ll O(1).
\end{align*}
Hence,
\begin{multline*}
  \mathcal{I}_{\mathrm{L}}(\mx;F)\approx\left|\sum_{m=1}^{M}\left[\frac{k_F}{k_F-k_1}\bigg(J_0(k_F|\mx-\my_m|)^2+J_1(k_F|\mx-\my_m|)^2\bigg)\right.\right.\\
        \left.\left.-\frac{k_1}{k_F-k_1}\bigg(J_0(k_1|\mx-\my_m|)^2+J_1(k_1|\mx-\my_m|)^2\bigg)\right]\right|.
\end{multline*}
\end{proof}

Above result shows that the terms $\Lambda_{\mathrm{D}}(\alpha,\beta,|\mx-\my_m|;n)$ will disturb the shape identification of $\Gamma$ i.e., imaging performance of $\mathcal{I}_{\mathrm{L}}(\mx;F)$ highly depends on the range of incident and observation directions, and applied wavenumber $k_f$. Note that if one can find $\alpha$ and $\beta$ such that
\[\beta\ne\alpha\quad\mbox{and}\quad\sin\frac{n(\beta-\alpha)}{2}=0,\]
an accurate shape of $\Gamma$ can be obtained. Notice that this is only for $\beta-\alpha=2\pi$, i.e., full-view case. Hence, for obtaining a good result in the limited-view problem, sufficiently large $k_f$ must be applied.

\subsection{Neumann boundary condition (TM) case}
Now, we consider the TE case. With the same configuration of previous subsection, we can obtain following result.
\begin{thm}\label{StructureImagingFunctionalLimitedTE}
  For sufficiently large $N$ and $F$, (\ref{ImagingFunctionN}) can be written as follows:
  \begin{align*}
    \mathcal{I}_{\mathrm{L}}(\mx;F)\approx\frac{1}{(\beta-\alpha)^2}&\left|\sum_{m=1}^{M}\left[\frac{k_F(C_1)^2}{k_F-k_1}\bigg(J_0(k_F|\mx-\my_m|)^2+J_1(k_F|\mx-\my_m|)^2\bigg)\right.\right.\\
        &-\frac{k_1(C_1)^2}{k_F-k_1}\bigg(J_0(k_1|\mx-\my_m|)^2+J_1(k_1|\mx-\my_m|)^2\bigg)\\
        &+\frac{(C_1)^2-(C_2)^2}{k_F-k_1}\int_{k_1}^{k_F}J_1(k|\mx-\my_m|)^2dk\\
        &\left.\left.+i\frac{C_1C_2}{2(k_F-k_1)|\mx-\my_m|}\bigg(J_0(k_1|\mx-\my_m|)^2-J_0(k_F|\mx-\my_m|)^2\bigg)\right]\right|,
  \end{align*}
  where constants $C_1$ and $C_2$ are defined in (\ref{ConstantC}).
\end{thm}
\begin{proof}
Suppose that for every $f$, number of non-zero singular values $M_f$ is almost equal to $M$. Then (\ref{ImagingFunctionA}) becomes
\begin{align*}
  \mathcal{I}_{\mathrm{L}}(\mx;F)&=\abs{\sum_{f=1}^{F}\sum_{m=1}^{M_f}\left(\hat{\mS}_{\mathrm{D}}(\mx;k_f)^*\mU_m(k_f)\right)\left(\hat{\mS}_{\mathrm{D}}(\mx;k_f)^*\overline{\mV}_m(k_f)\right)}\\
  &=\frac{1}{(\beta-\alpha)^2F}\abs{\sum_{f=1}^{F}\sum_{m=1}^{M}\left(\int_{\mathbb{S}_+^1}\vt\cdot\vn(\my_m)e^{ik_f\vt\cdot(\mx-\my_m)}d\vt\right)^2}\\
  &=\frac{1}{(\beta-\alpha)^2(k_F-k_1)}\abs{\sum_{m=1}^{M}\int_{k_1}^{k_F}\left(\int_{\mathbb{S}_+^1}\vt\cdot\vn(\my_m)e^{ik\vt\cdot(\mx-\my_m)}d\vt\right)^2dk}.
\end{align*}
Then, by setting $\vt:=(\cos\theta,\sin\theta)$, $\vn(\my_m)=(\cos\nu_m,\sin\nu_m)$, and $\mx-\my_m:=r_m(\cos\phi_m,\sin\phi_m)$, applying Jacobi-Anger expansion yields
\begin{align*}
  \int_{\mathbb{S}_+^1}&\vt\cdot\vn(\my_m)e^{ik\vt\cdot(\mx-\my_m)}d\vt=\int_{\alpha}^{\beta}\cos(\theta-\nu_m)e^{ik|\mx-\my_m|\cos(\theta-\phi_m)}d\theta\\
  =&\int_{\alpha}^{\beta}\cos(\theta-\nu_m)\left[J_0(k|\mx-\my_m|)+2\sum_{n=1}^{\infty}i^nJ_n(k|\mx-\my_m|)\cos n(\theta-\phi_m)\right]d\theta\\
  =&C_1J_0(k|\mx-\my_m|)+iC_2J_1(k|\mx-\my_m|)+2\sum_{n=2}^{\infty}\Lambda_{\mathrm{N}}(\alpha,\beta,k|\mx-\my_m|;n),
\end{align*}
where
\begin{align}
\begin{aligned}\label{ConstantC}
  C_1&=2\sin\frac{\beta-\alpha}{2}\cos\frac{\beta+\alpha-2\nu_m}{2}\\
  C_2&=(\beta-\alpha)\bigg(\frac{\mx-\my_m}{|\mx-\my_m|}\cdot\vn(\my_m)\bigg)+\sin(\beta-\alpha)\cos(\beta+\alpha-\nu_m-\phi_m).
\end{aligned}
\end{align}
Same as the derivation of (\ref{Boundednesso1}) in Theorem \ref{StructureImagingFunctionalS}, the following holds for sufficiently large $k$:
\begin{align*}
  &\frac{C_1}{k_F-k_1}\sum_{n=2}^{\infty}J_0(k|\mx-\my_m|)\Lambda_{\mathrm{N}}(\alpha,\beta,k|\mx-\my_m|;n)\ll O(1)\\
  &\frac{C_2}{k_F-k_1}\sum_{n=2}^{\infty}J_1(k|\mx-\my_m|)\Lambda_{\mathrm{N}}(\alpha,\beta,k|\mx-\my_m|;n)\ll O(1)\\
  &\frac{1}{k_F-k_1}\left(\sum_{n=2}^{\infty}\Lambda_{\mathrm{N}}(\alpha,\beta,k|\mx-\my_m|;n)\right)^2\ll O(1).
\end{align*}
Hence, we can evaluate
\begin{align*}
  \mathcal{I}_{\mathrm{L}}&(\mx;F)=\frac{1}{(\beta-\alpha)^2(k_F-k_1)}\left|(C_1)^2\sum_{m=1}^{M}\int_{k_1}^{k_F}J_0(k|\mx-\my_m|)^2dk\right.\\
  &\left.-(C_2)^2\sum_{m=1}^{M}\int_{k_1}^{k_F}J_1(k|\mx-\my_m|)^2dk+iC_1C_2\sum_{m=1}^{M}\int_{k_1}^{k_F}J_0(k|\mx-\my_m|)J_1(k|\mx-\my_m|)dk\right|.
\end{align*}
Finally, applying Theorem \ref{StructureImagingFunctional} and (\ref{IntegralBessel01}), we can obtain desired result.
\end{proof}
It is interesting to observe that opposite to the (\ref{ImagingFunctionA}) in Theorem \ref{StructureImagingFunctionalA}, $\mathcal{I}_{\mathrm{L}}(\mx;F)$ will plots large magnitude at $\mx=\my_m\in\Gamma$ since $J_0(k|\mx-\my_m|)$ exists. But it also produces unexpected points of large magnitude at $\mx\notin\Gamma$. Note that in the full-view case, i.e., $\beta-\alpha=2\pi$, then
\[C_1=\Lambda_{\mathrm{N}}(\alpha,\beta,k|\mx-\my_m|;n)=0\quad\mbox{and}\quad C_2=2\pi\bigg(\frac{\mx-\my_m}{|\mx-\my_m|}\cdot\vn(\my_m)\bigg).\]
Hence, we can obtain Theorem \ref{StructureImagingFunctionalA}.

Now, we end up this subsection with the following conclusion: if the range of incident and observation directions is sufficiently wide, produced image should be acceptable but if it is narrow, one cannot obtain a good result when each applied wavenumber $k_f$ is sufficiently large enough.

\section{Numerical examples}\label{Sec4}
\subsection{Common features}\label{Sec4-1}
In this section, we present some numerical examples for imaging arc-like cracks satisfying the Dirichlet (TM polarization) or Neumann boundary condition (TE polarization). Throughout this section, the applied wave number is taken of the form $k_f=\frac{2\pi}{\lambda_f}$; here $\lambda_f$, $f=1,2,\cdots,F$, is the given wavelength. In this paper, the wavenumbers $k_f$ are always equi-distributed in the interval $[k_1,k_f]$.

Four $\Gamma_j$ are chosen for illustration:
\begin{align*}
\Gamma_1&=\set{(s,0.3):s\in[-0.5,0.5]}\\
\Gamma_2&=\set{\left(s,\frac{1}{2}\cos\frac{s\pi}{2}+\frac{1}{5}\sin\frac{s\pi}{2}-\frac{1}{10}\cos\frac{3s\pi}{2}\right):s\in[-1,1]}\\
\Gamma_3&=\set{\left(2\sin\frac{s}{2},\sin s\right):s\in\left[\frac{\pi}{4},\frac{7\pi}{4}\right]}\\
\Gamma_4&=\Gamma_4^{(1)}\cup\Gamma_4^{(2)}
\end{align*}
where
\begin{align*}
\Gamma_4^{(1)}&=\set{\left(s-0.2,-0.5s^2+0.6\right):s\in[-0.5,0.5]}\\
\Gamma_4^{(2)}&=\set{\left(s+0.2,s^3+s^2-0.6\right):s\in[-0.5,0.5]}.
\end{align*}
The search domain $\Omega$ is illustrated in Table \ref{Configuration} for $\Gamma_j$, $j=1,2,3$ and $4$. For each $\mx\in\Omega$, the step size of $\mx$ is taken of the order of 0.02. As for the observation directions $\hat{\mx}_j$, same as (\ref{IODir}), they are taken as
\[\hat{\mx}_j=\left(\cos\theta_j,\sin\theta_j\right),\quad\theta_j=\alpha+(\beta-\alpha)\frac{j-1}{N-1}\]
where $\alpha=0$ and $\beta=2\pi$ for the full-view case, and $\alpha=\pi/6$ and $\beta=5\pi/6$ for the limited-view case.

It is worth mentioning that, since the reliable and efficient solution of the direct scattering problem indicated previously is very important (for example, avoiding inverse crime, etc.), all numerical data in this section (the elements $u_{\infty}(\hat{\mx}_j,\vt_l;k_f)$ for $j,l=1,2,\cdots,N$ of the dataset $\mathbb{K}(k_f)$) are generated by the Nystr{\"o}m method for both the Dirichlet and Neumann boundary conditions as presented in \cite{K} and \cite{M2}, respectively. After obtaining the dataset, a 15dB Gaussian random noise is added to the unperturbed data to show the robustness of the proposed algorithm. In order to obtain the number of nonzero singular values $M_f$ for each frequency, a $0.01$-threshold scheme (choosing first $M$ singular values $s_m$ such that $s_m/s_1\geq0.01$) is adopted. A more detailed discussion of thresholding can be found in \cite{PL1,PL3} (see \cite{HSZ} for volumetric extended target case).


\begin{table}
\begin{center}
\begin{tabular}{c||c|c|c|c|c|c||c}
\hline \multirow{2}*{Crack}&\multicolumn{2}{c|}{TM case (section \ref{Sec4-2})}&\multicolumn{2}{c|}{TE case (section \ref{Sec4-3})}&\multirow{2}*{$\lambda_1$}&\multirow{2}*{$\lambda_F$}&search domain\cr\cline{2-5}
&$N$&$F$&$N$&$F$& & &$\Omega$\\
\hline\hline
$\Gamma_1$&$~~~~~~16~~~~~~$&$10$&$~~~~~~16~~~~~~$&$10$&$0.5$&$0.4$&$[-1,1]\times[-1,1]$\\
$\Gamma_2$&$28$&$12$&$36$&$12$&$0.6$&$0.3$&$[-2,2]\times[-2,2]$\\
$\Gamma_3$&$40$&$16$&$64$&$16$&$0.5$&$0.3$&$[-2,2]\times[-1,3]$\\
$\Gamma_4$&$32$&$24$&$64$&$24$&$0.4$&$0.2$&$[-1,1]\times[-1,1]$\\
\hline
\end{tabular}
\caption{\label{Configuration}Test configuration for $\Gamma_j$, $j=1,2,3$ and $4$.}
\end{center}
\end{table}

\subsection{Dirichlet boundary condition case - TM}\label{Sec4-2}
In this case, we consider the imaging of crack with Dirichlet boundary condition. First, let us consider the $\Gamma_1$. Throughout many references, \cite{P4,PL1,PL3}, one can easily notice that when a crack is straight line, it can be successfully retrieved. In this result, an expected result is appeared, refer to Figure \ref{FigGamma1D}.

\begin{figure}[!ht]
\begin{center}
\subfloat[Map of $\mathcal{I}_{\mathrm{D}}(\mx;F)$]{\label{FigGamma1DF}\includegraphics[width=0.32\textwidth]{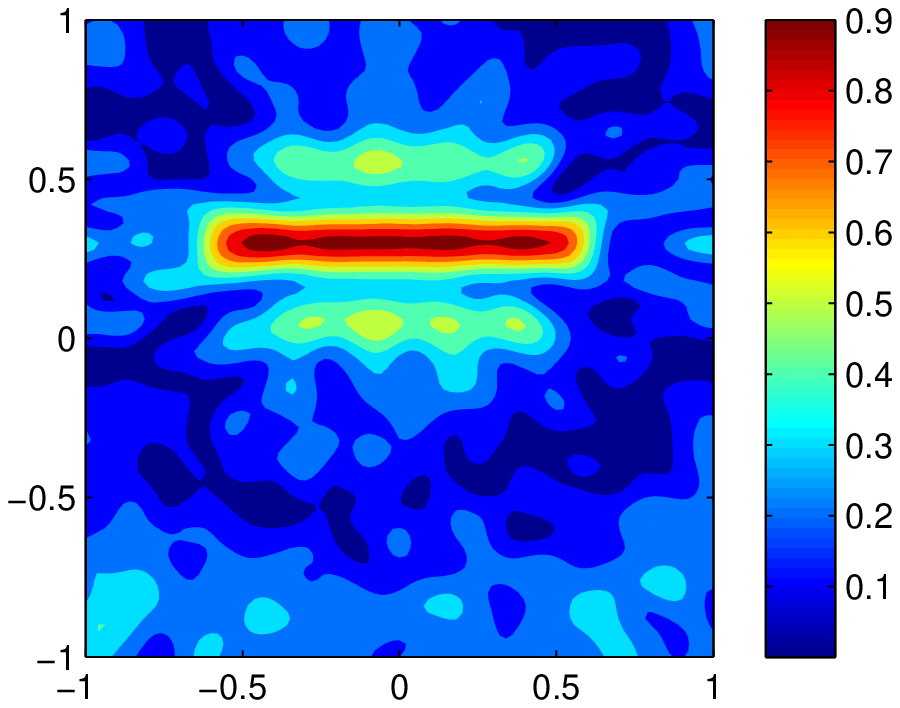}}
\subfloat[Map of $\mathcal{I}_{\mathrm{L}}(\mx;F)$]{\label{FigGamma1DN}\includegraphics[width=0.32\textwidth]{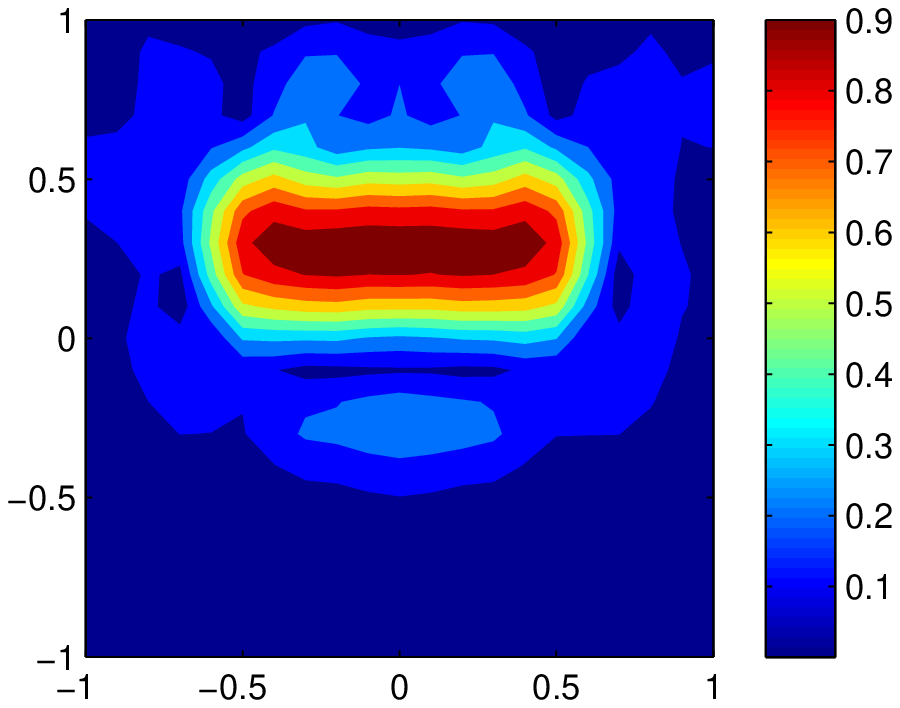}}
\subfloat[True shape]{\label{FigGamma1DT}\includegraphics[width=0.32\textwidth]{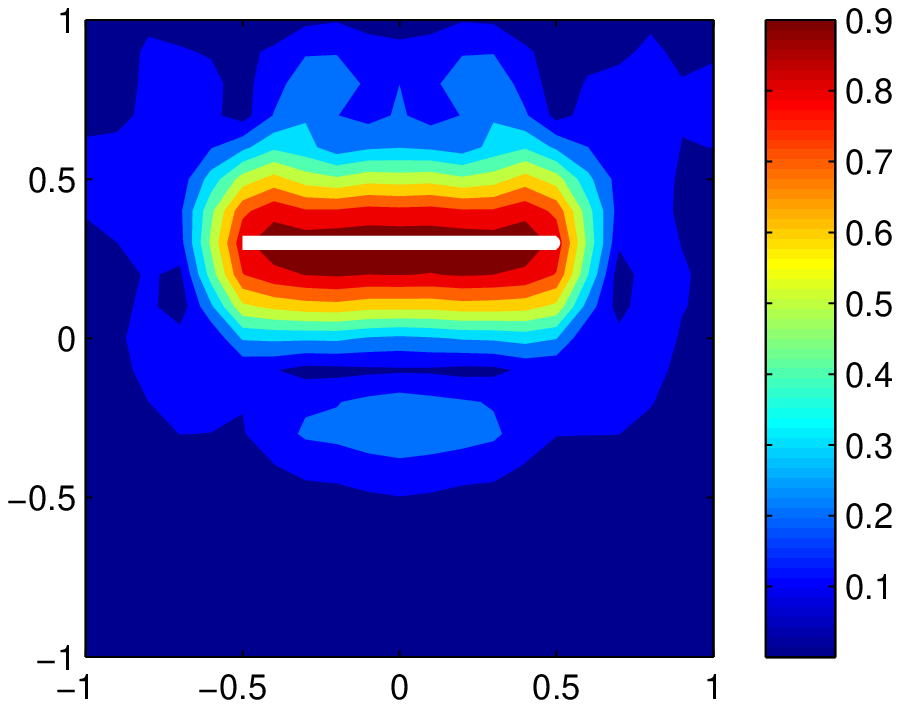}}
\caption{\label{FigGamma1D}Maps of $\mathcal{I}_{\mathrm{D}}(\mx;F)$ and $\mathcal{I}_{\mathrm{L}}(\mx;F)$ for $\Gamma_1$.}
\end{center}
\end{figure}

Similarly with the penetrable inclusion case as dealt with in \cite{P3,P4,PL1,PL2,PL3}, when the crack is not anymore a straight line, for example, the image of $\Gamma_2$, poor results are observed. Fortunately, in this case, the location of the end-points of $\Gamma_2$ is well identified. That is, connected by a straight line, it should provide a good initial guess for an iterative solution algorithm, refer to Figure \ref{FigGamma2D}.

\begin{figure}[!ht]
\begin{center}
\subfloat[Map of $\mathcal{I}_{\mathrm{D}}(\mx;F)$]{\label{FigGamma2DF}\includegraphics[width=0.32\textwidth]{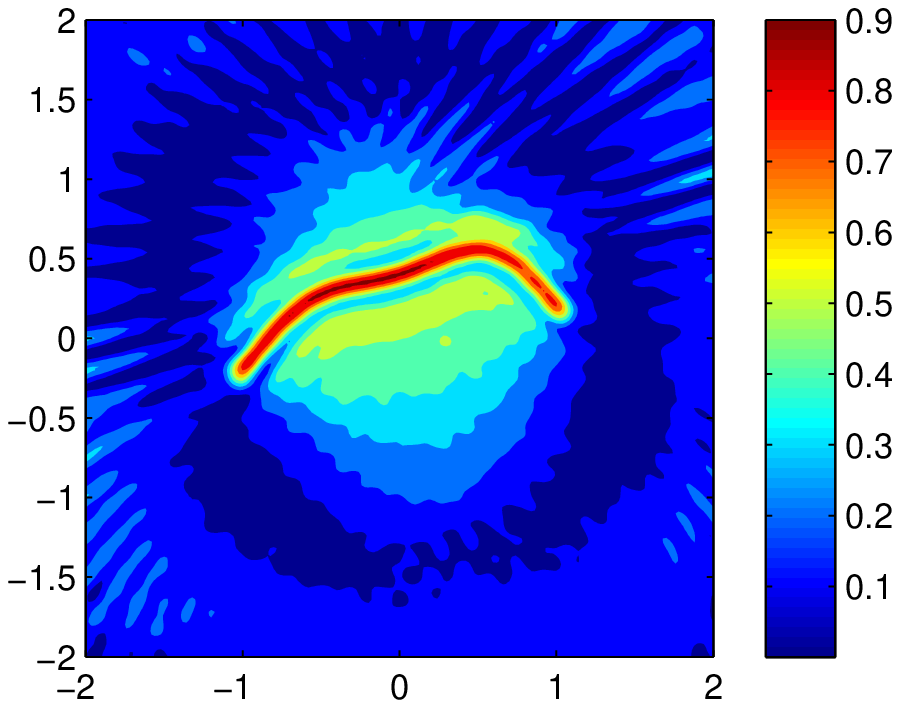}}
\subfloat[Map of $\mathcal{I}_{\mathrm{L}}(\mx;F)$]{\label{FigGamma2DN}\includegraphics[width=0.32\textwidth]{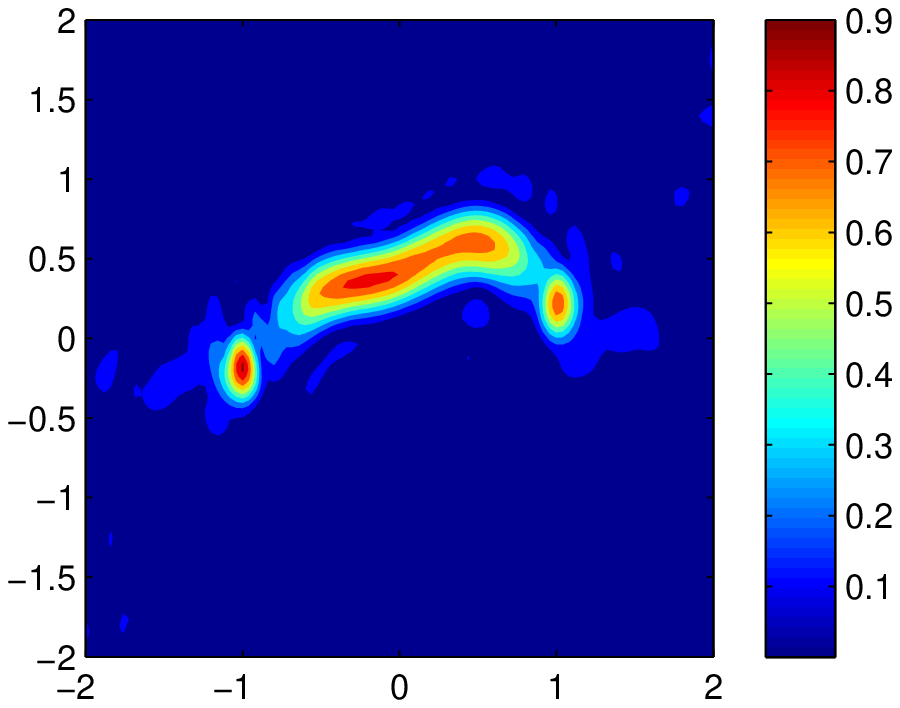}}
\subfloat[True shape]{\label{FigGamma2DT}\includegraphics[width=0.32\textwidth]{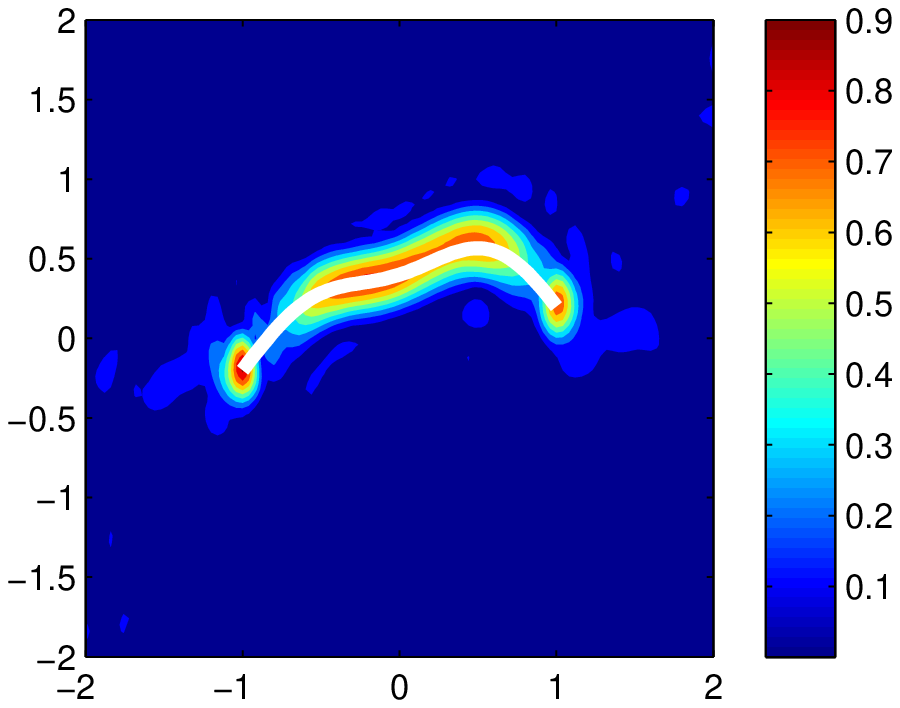}}
\caption{\label{FigGamma2D}Same as Figure \ref{FigGamma1D} except the crack is $\Gamma_2$.}
\end{center}
\end{figure}

In addition, for a complicated crack case $\Gamma_3$, only limited part of crack can be imaged, refer to Figure \ref{FigGamma3D}. In order to detect the remaining part of $\Gamma_3$, one must change the observation (and also incident) directions. For example, if one wants to detect the right-hand side of $\Gamma_3$, $\alpha=-\frac{\pi}{6}$ and $\beta=\frac{\pi}{6}$ of (\ref{IODir}) will be a good choice. In Figure \ref{FigGamma3D-2}, corresponding results are exhibited. For the imaging of extended targets, similar phenomenon can be found in \cite{HHSZ}.

\begin{figure}[!ht]
\begin{center}
\subfloat[Map of $\mathcal{I}_{\mathrm{D}}(\mx;F)$]{\label{FigGamma3DF}\includegraphics[width=0.32\textwidth]{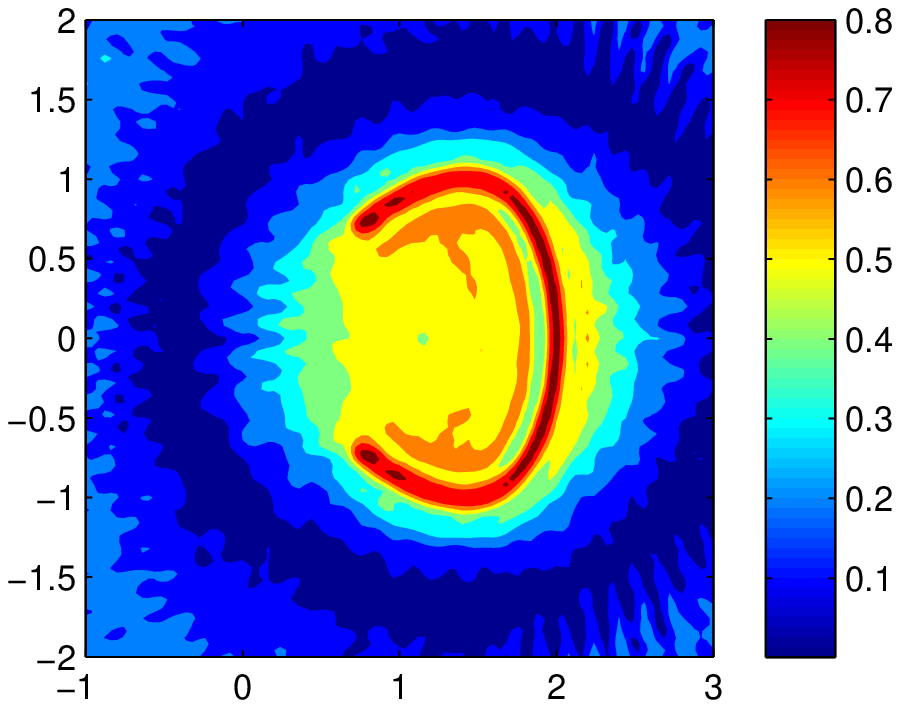}}
\subfloat[Map of $\mathcal{I}_{\mathrm{L}}(\mx;F)$]{\label{FigGamma3DN}\includegraphics[width=0.32\textwidth]{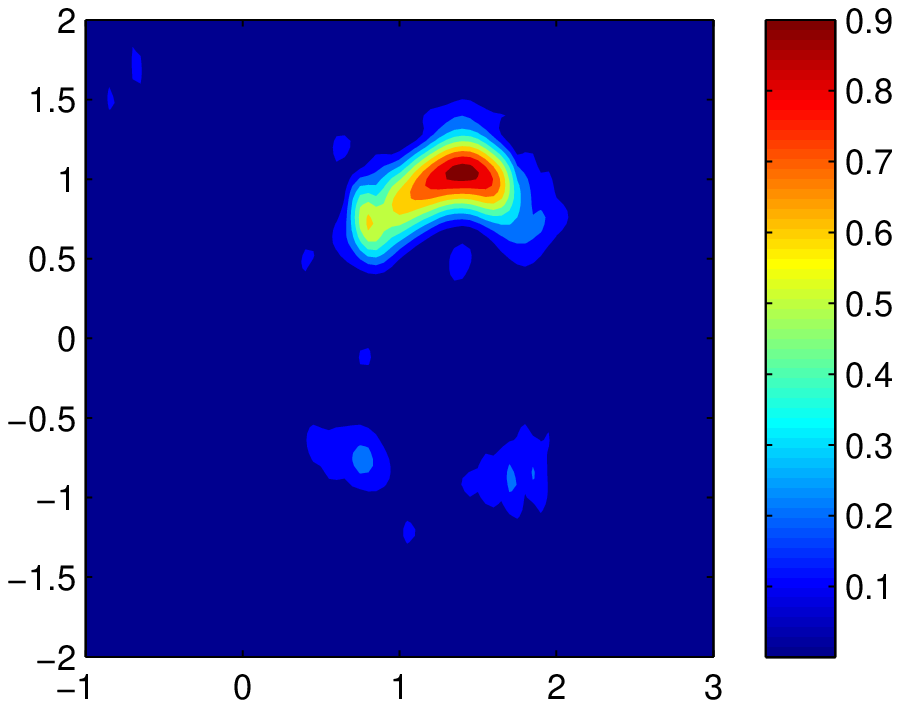}}
\subfloat[True shape]{\label{FigGamma3DT}\includegraphics[width=0.32\textwidth]{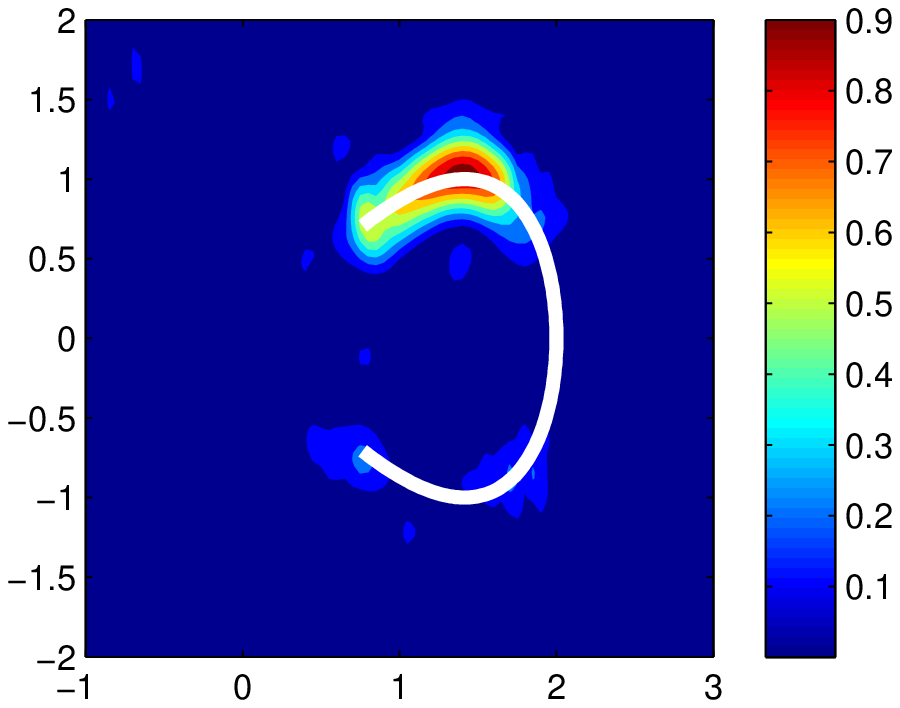}}
\caption{\label{FigGamma3D}Same as Figure \ref{FigGamma1D} except the crack is $\Gamma_3$.}
\end{center}
\end{figure}

\begin{figure}[!ht]
\begin{center}
\subfloat[$\alpha=\frac{5}{6}\pi$ and $\beta=\frac{7}{6}\pi$]{\label{FigGamma3Db}\includegraphics[width=0.32\textwidth]{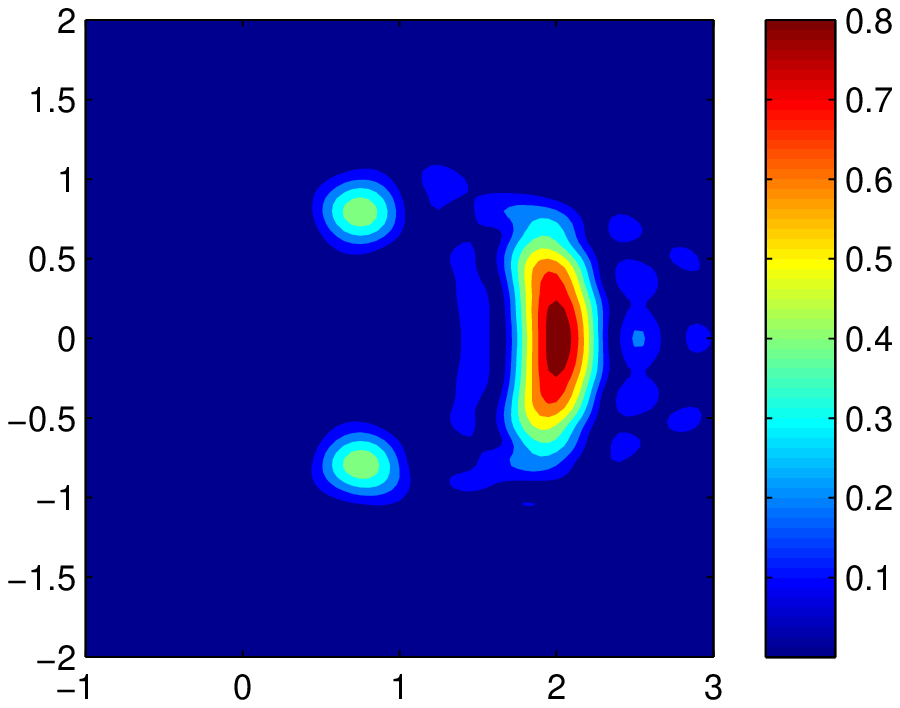}}
\subfloat[$\alpha=\frac{7}{6}\pi$ and $\beta=\frac{11}{6}\pi$]{\label{FigGamma3DNb}\includegraphics[width=0.32\textwidth]{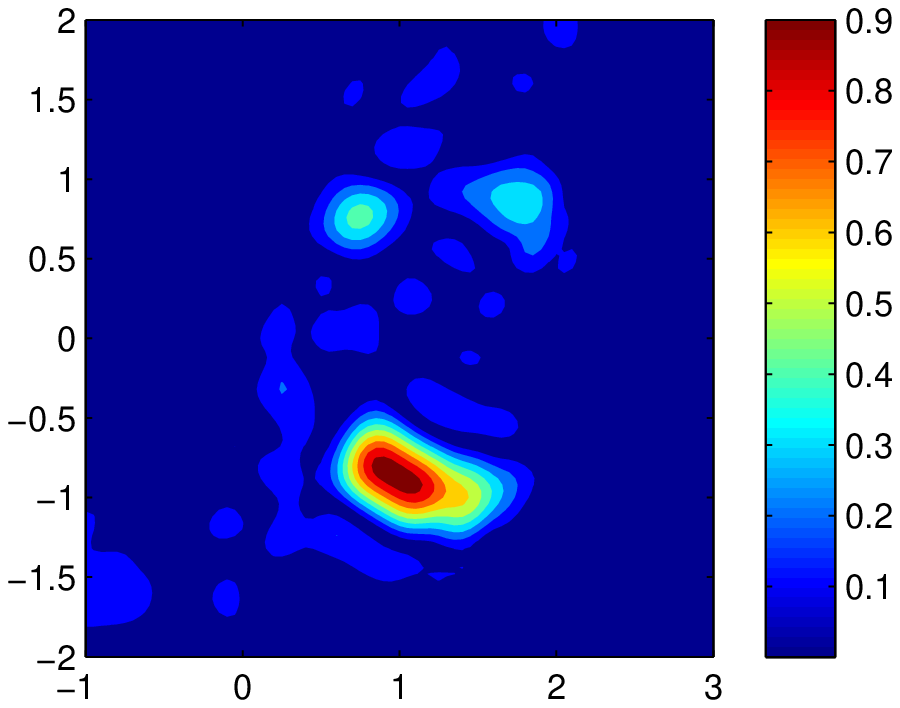}}
\subfloat[$\alpha=-\frac{1}{6}\pi$ and $\beta=\frac{1}{6}\pi$]{\label{FigGamma3DTb}\includegraphics[width=0.32\textwidth]{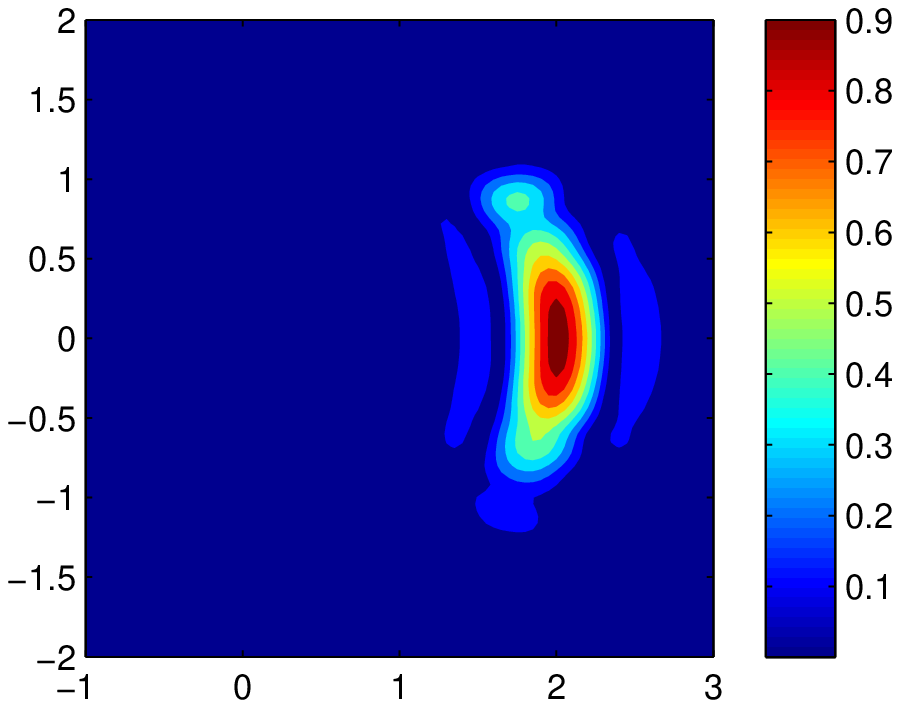}}
\caption{\label{FigGamma3D-2}Maps of $\mathcal{I}_{\mathrm{L}}(\mx;F)$ for $\Gamma_3$.}
\end{center}
\end{figure}

\begin{rem}
Instead of the scattered field dataset generated from the Nystr{\"o}m method, some authors introduce a similar formulation involving the solution of a second-kind Fredholm integral equation along the crack, refer to \cite{N}. Numerical experimentation shows that images of a crack from far-field data acquired by the Nystr{\"o}m method or from the ones calculated via this alternative formulation are almost indistinguishable (see Figure \ref{FigGammaD-Naz}). Using near-field data instead of far-field one yields similar results.
\end{rem}

\begin{figure}[!ht]
\begin{center}
\subfloat[Map of $\mathcal{I}_{\mathrm{L}}(\mx;F)$ for $\Gamma_1$]{\label{FigGamma1DNaz}\includegraphics[width=0.32\textwidth]{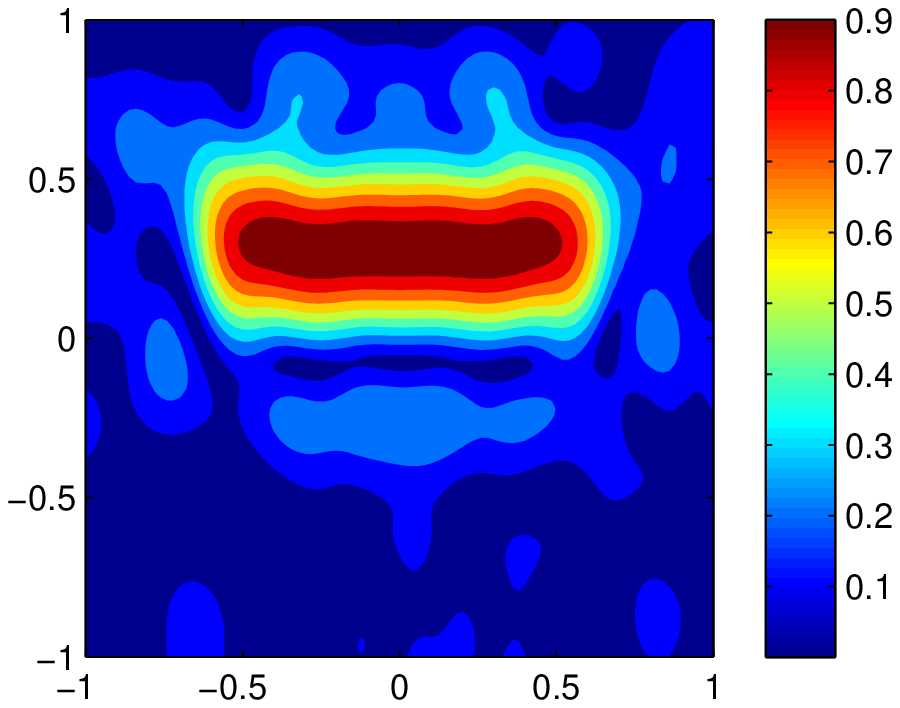}}
\subfloat[Map of $\mathcal{I}_{\mathrm{L}}(\mx;F)$ for $\Gamma_2$]{\label{FigGamma2DNaz}\includegraphics[width=0.32\textwidth]{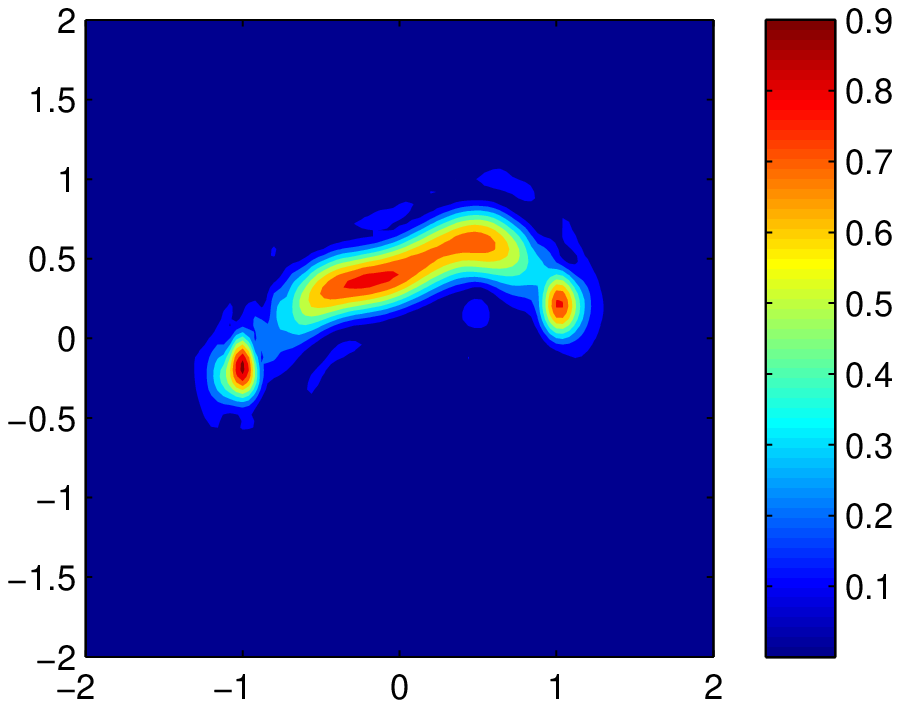}}
\subfloat[Map of $\mathcal{I}_{\mathrm{L}}(\mx;F)$ for $\Gamma_3$]{\label{FigGamma3DNaz}\includegraphics[width=0.32\textwidth]{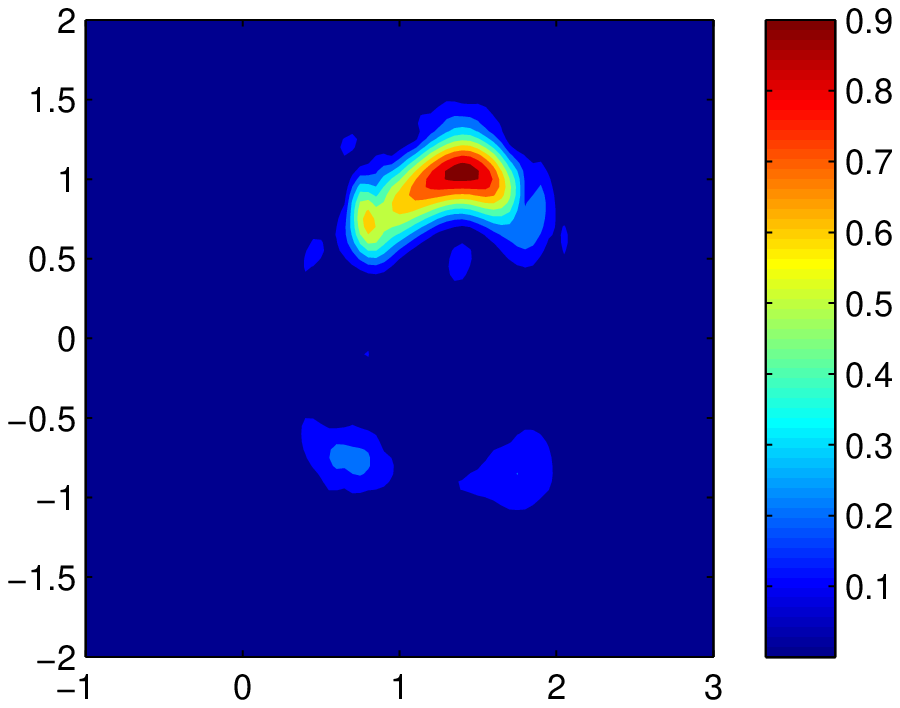}}
\caption{\label{FigGammaD-Naz}Map of $\mathcal{I}_{\mathrm{L}}(\mx;F)$ for $\Gamma_j$, $j=1,2,3$, where dataset generated via method in \cite{N}.}
\end{center}
\end{figure}

Both the mathematical configuration and the numerical analysis could be extended in somewhat straightforward fashion to the case of non-overlapped multiple cracks. We will not present the derivation herein and simply illust some examples. Let us notice that the elements $u_{\infty}(\hat{\mx}_j,\vt_l;k_f)$ for $j,l=1,2,\cdots,N$ of dataset $\mathbb{K}(k_f)$ are generated from the Nystr\"om method now applied to scattering by more than one crack (see \cite{N,PhDWKP}).

Now, let us work with $\Gamma_4$. Maps of $\mathcal{I}_{\mathrm{D}}(\mx;F)$ are displayed in Figure \ref{FigGamma4D}. Similarly with the previous example, we can only identify the $\Gamma_4^{(1)}$ with $\alpha=\frac{\pi}{6}$ and $\beta=\frac{5\pi}{6}$ of (\ref{IODir}). For retrieving $\Gamma_4^{(2)}$, one must choose another incident (and observation) direction setting for example, $\alpha=\frac{7\pi}{6}$ and $\beta=\frac{11\pi}{6}$ is a good choice.

\begin{figure}[!ht]
\begin{center}
\subfloat[Map of $\mathcal{I}_{\mathrm{D}}(\mx;F)$]{\label{FigGamma4Da}\includegraphics[width=0.32\textwidth]{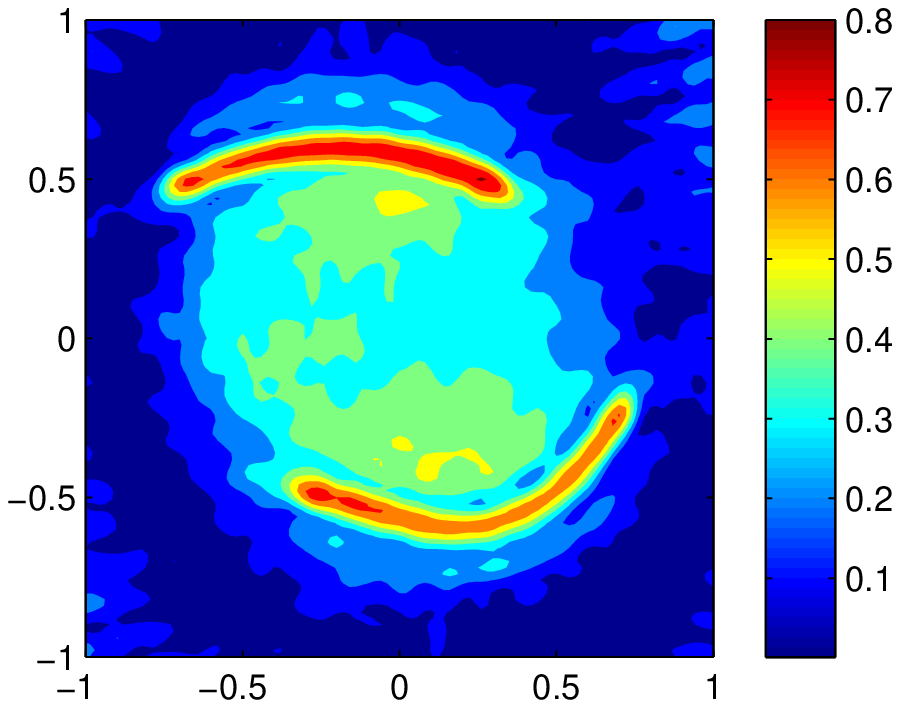}}
\subfloat[Map of $\mathcal{I}_{\mathrm{L}}(\mx;F)$]{\label{FigGamma4DNa}\includegraphics[width=0.32\textwidth]{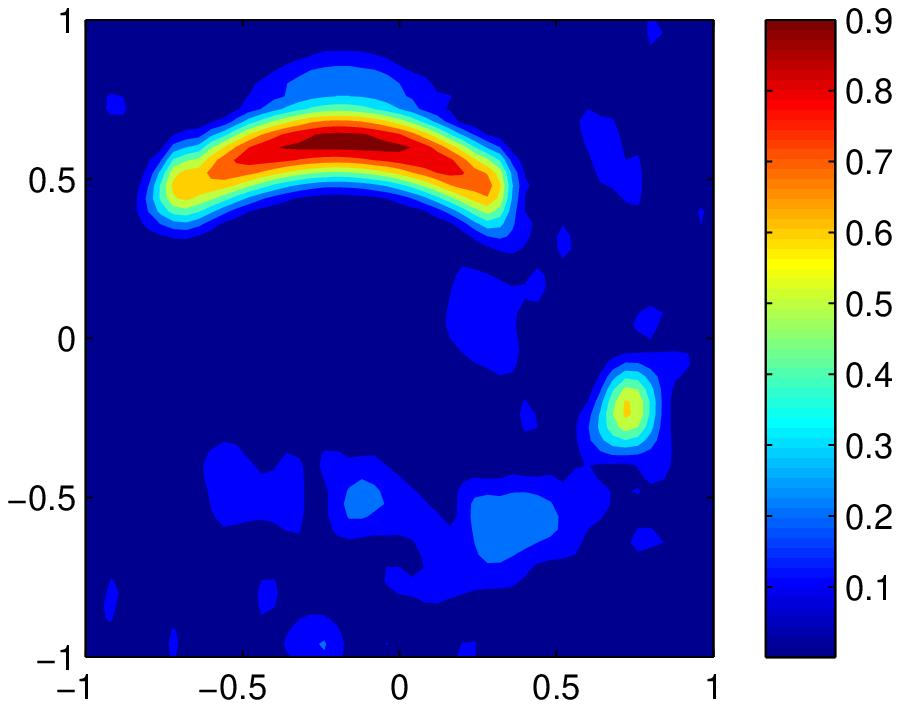}}
\subfloat[True shape]{\label{FigGamma4DTa}\includegraphics[width=0.32\textwidth]{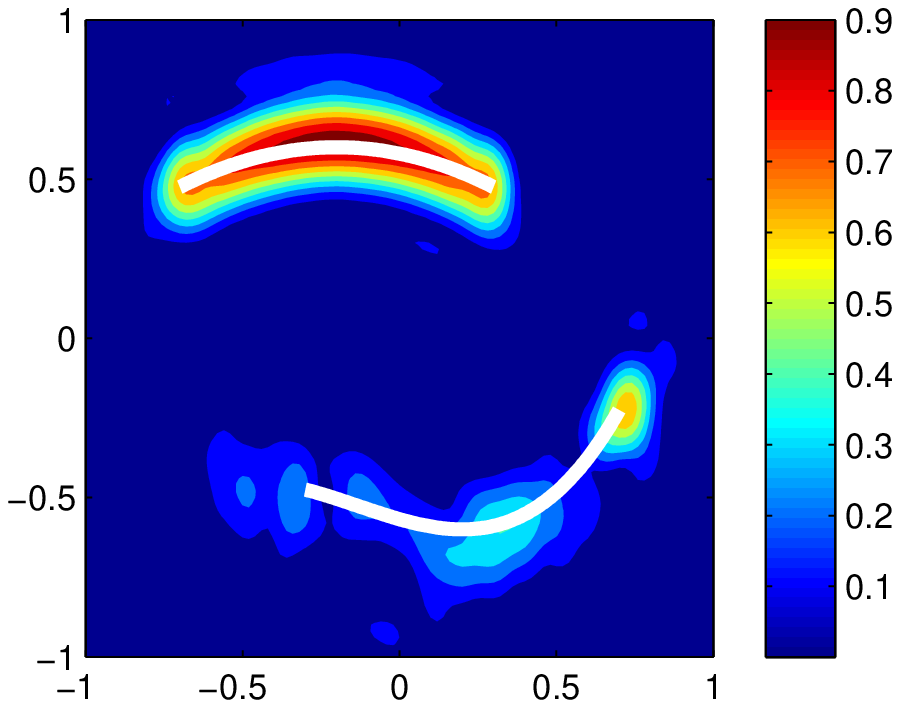}}
\caption{\label{FigGamma4D}Same as Figure \ref{FigGamma1D} except the crack is $\Gamma_4$.}
\end{center}
\end{figure}
%
%
\subsection{Neumann boundary condition case - TE}\label{Sec4-3}
In this case, we present some imaging results for the Neumann boundary condition. The configuration is the same as previously and we use a set of fixed directions $\vn_l$ as
\[\vn_l=\left(\cos\frac{2\pi l}{L},\sin\frac{2\pi l}{L}\right)\quad\mbox{for}\quad l=1,2,\cdots,L.\]
Throughout this section, we use $L=8$ directions for the imaging of $\Gamma_1$ and $L=24$ for $\Gamma_2$, $\Gamma_3$ and $\Gamma_4$. Let us emphasize here that the fact that the direction normal to $\Gamma$ is not known of us results in the slowness of the imaging, i.e., some computational costs are needed (about 2 minutes are required to obtain Figure \ref{FigGamma1N} and 15 minutes to obtain Figures \ref{FigGamma2N}, \ref{FigGamma3N}, \ref{FigGammaN-Naz} and \ref{FigGamma4N} on a personal computer with 2.44 GHz dual-core pentium processor).

Let us consider the imaging of $\Gamma_1$. Similarly with the Dirichlet boundary condition case, although a blurred image appeared, it can be successfully retrieved, refer to Figure \ref{FigGamma1N}.

\begin{figure}[!ht]
\begin{center}
\subfloat[Map of $\mathcal{I}_{\mathrm{N}}(\mx;F)$]{\label{FigGamma1NF}\includegraphics[width=0.32\textwidth]{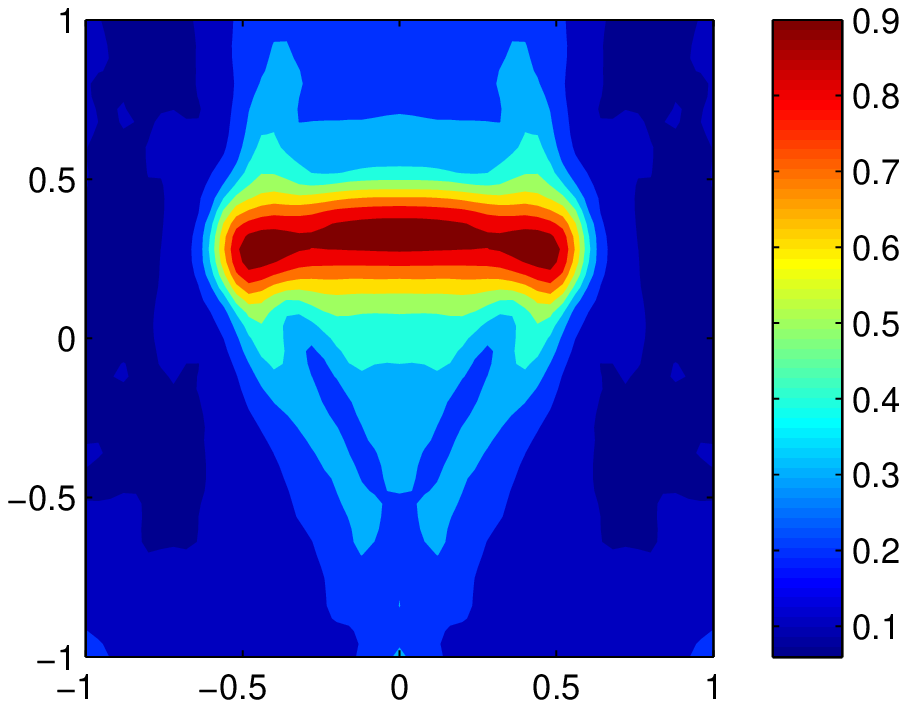}}
\subfloat[Map of $\mathcal{I}_{\mathrm{A}}(\mx;F)$]{\label{FigGamma1NN}\includegraphics[width=0.32\textwidth]{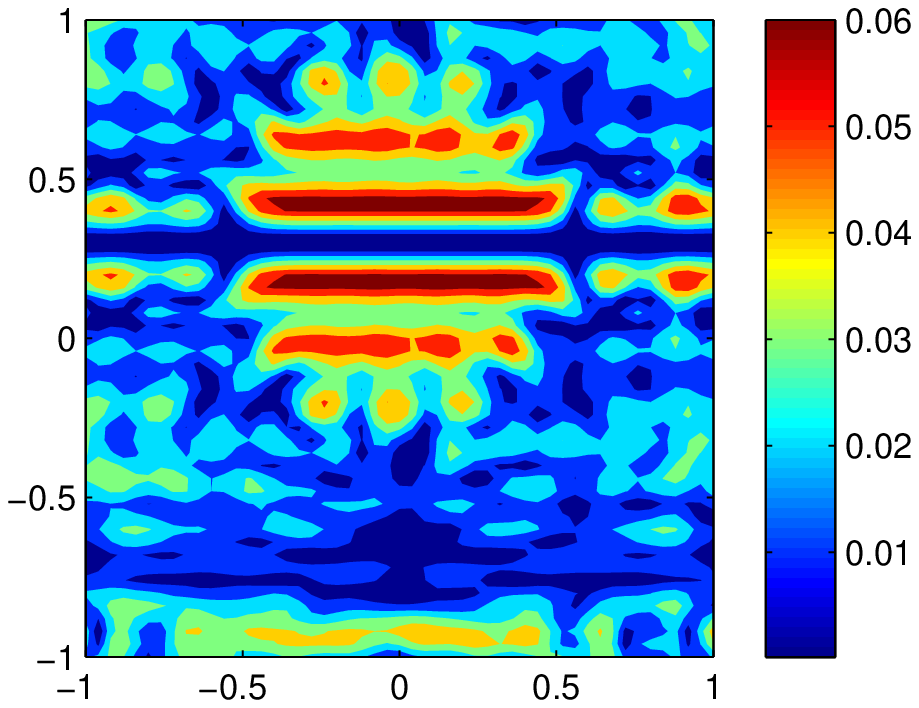}}
\subfloat[Map of $\mathcal{I}_{\mathrm{L}}(\mx;F)$]{\label{FigGamma1NT}\includegraphics[width=0.32\textwidth]{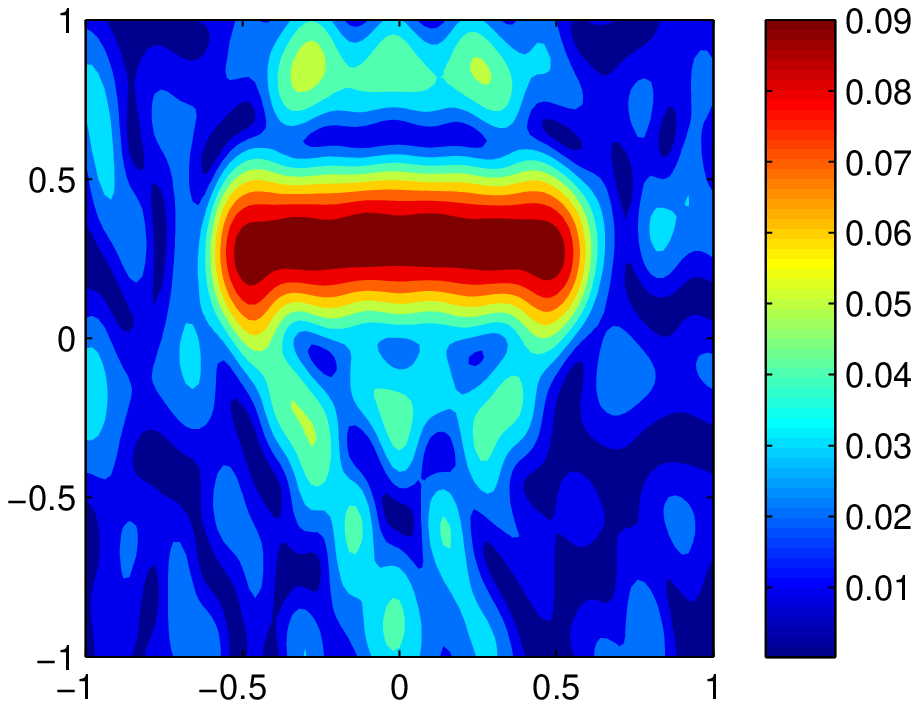}}
\caption{\label{FigGamma1N}Maps of $\mathcal{I}_{\mathrm{N}}(\mx;F)$, $\mathcal{I}_{\mathrm{A}}(\mx;F)$, and $\mathcal{I}_{\mathrm{L}}(\mx;F)$ for $\Gamma_1$.}
\end{center}
\end{figure}

In contrast with the Dirichlet boundary condition case, see Figure \ref{FigGamma2N}, when the crack is not anymore a straight line, a few ghost replicas appear and the location of end-points cannot be identified.

\begin{figure}[!ht]
\begin{center}
\subfloat[Map of $\mathcal{I}_{\mathrm{N}}(\mx;F)$]{\label{FigGamma2NF}\includegraphics[width=0.32\textwidth]{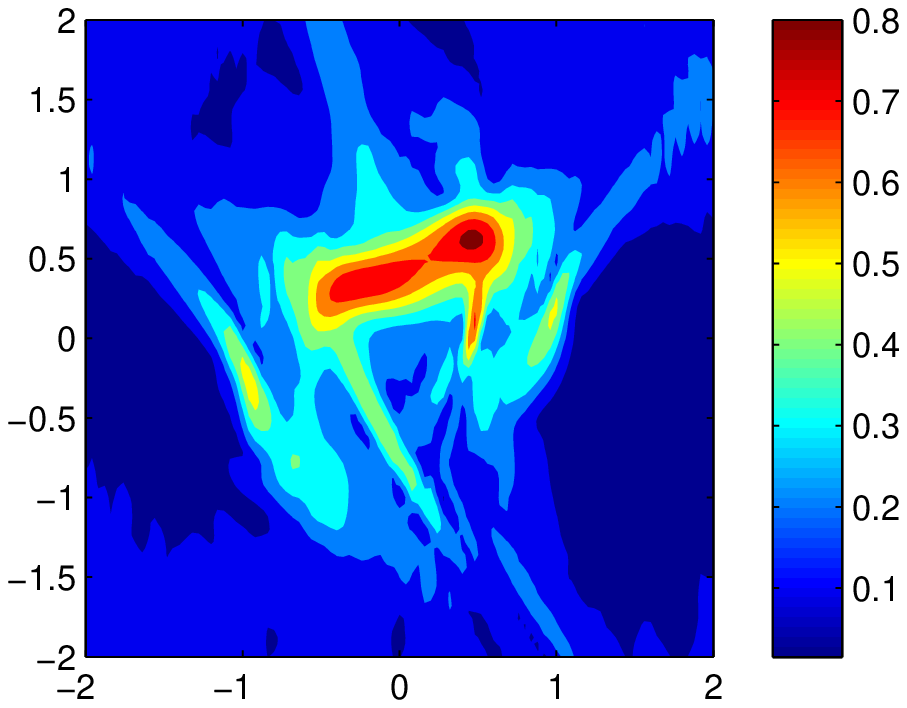}}
\subfloat[Map of $\mathcal{I}_{\mathrm{A}}(\mx;F)$]{\label{FigGamma2NN}\includegraphics[width=0.32\textwidth]{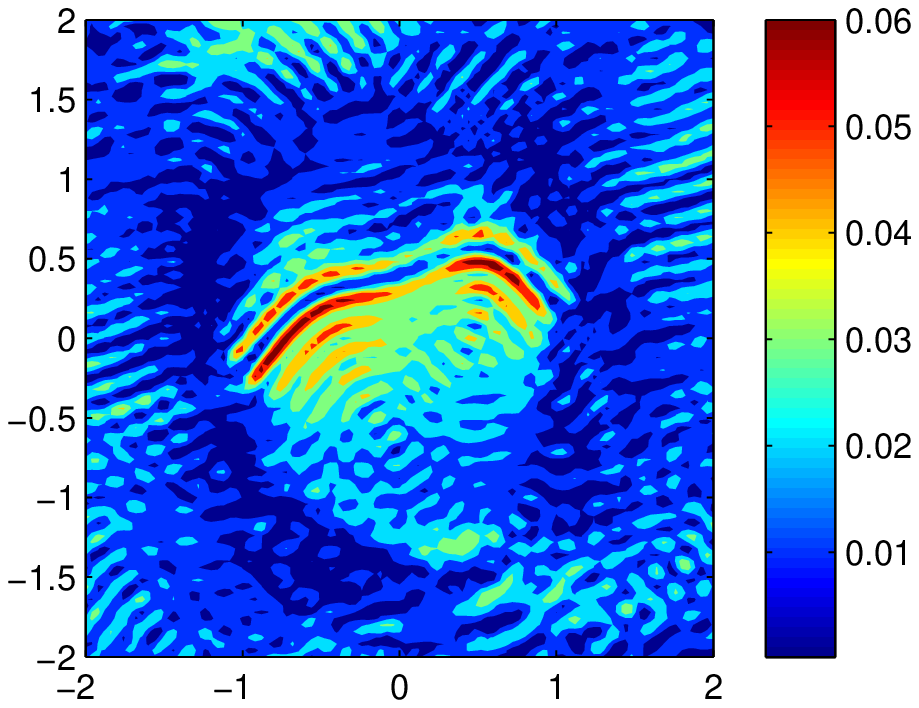}}
\subfloat[Map of $\mathcal{I}_{\mathrm{L}}(\mx;F)$]{\label{FigGamma2NT}\includegraphics[width=0.32\textwidth]{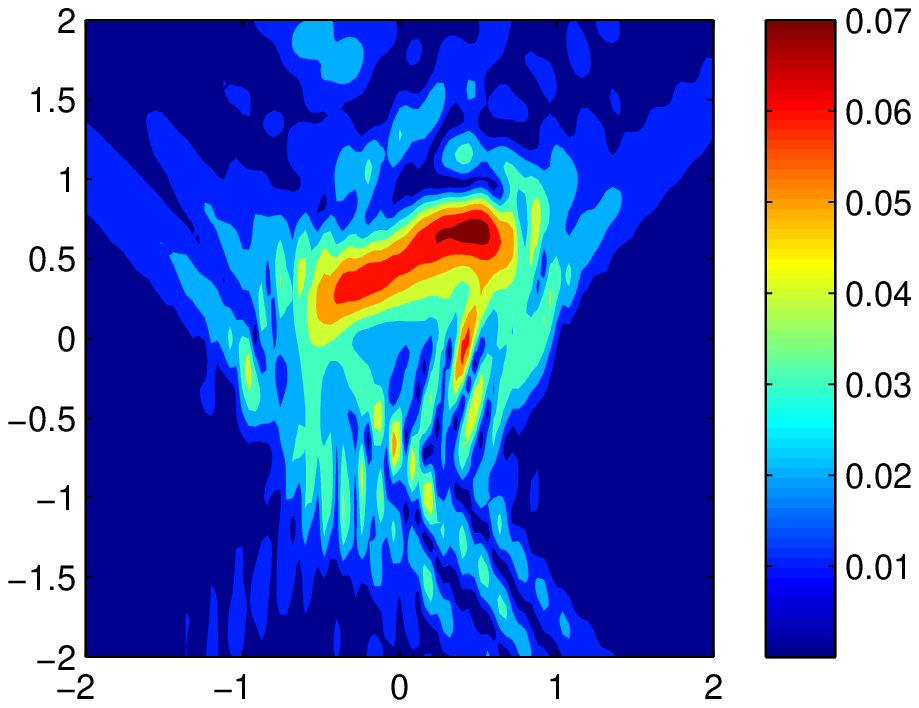}}
\caption{\label{FigGamma2N}Same as Figure \ref{FigGamma1N} except the crack is $\Gamma_2$.}
\end{center}
\end{figure}

Let us consider the imaging of $\Gamma_3$. Typical results are in Figure \ref{FigGamma3N}. It is interesting to observe that opposite to the Dirichlet boundary condition case, two parts of $\Gamma_3$ can be retrieved. However, the result is still bad. It means that if one wants to detect the remaining parts of $\Gamma_3$, more incident (and observation) directions are needed.

\begin{figure}[!ht]
\begin{center}
\subfloat[Map of $\mathcal{I}_{\mathrm{N}}(\mx;F)$]{\label{FigGamma3NF}\includegraphics[width=0.32\textwidth]{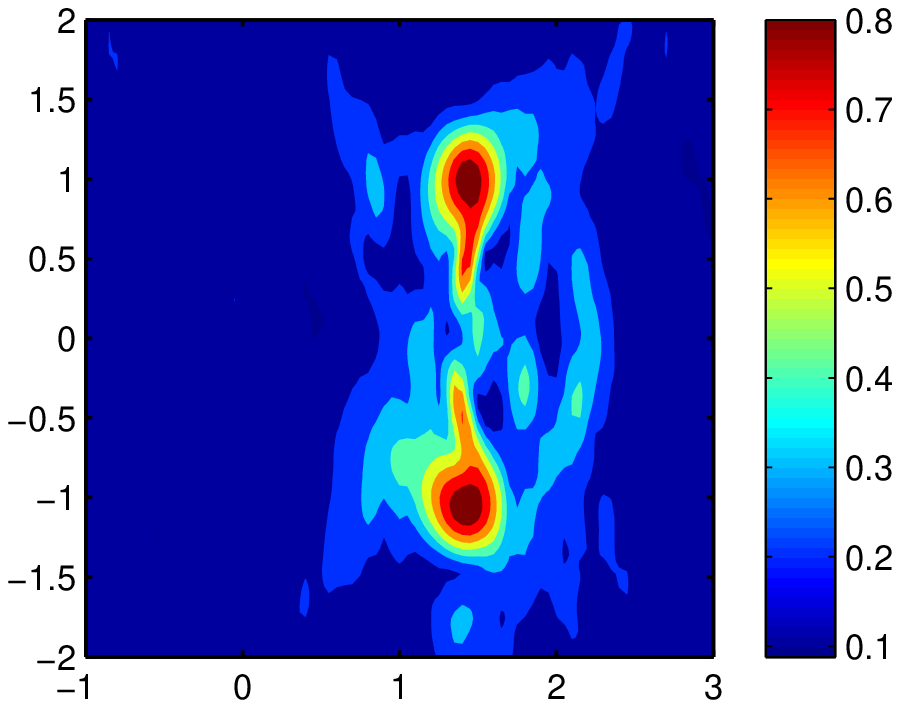}}
\subfloat[Map of $\mathcal{I}_{\mathrm{A}}(\mx;F)$]{\label{FigGamma3NN}\includegraphics[width=0.32\textwidth]{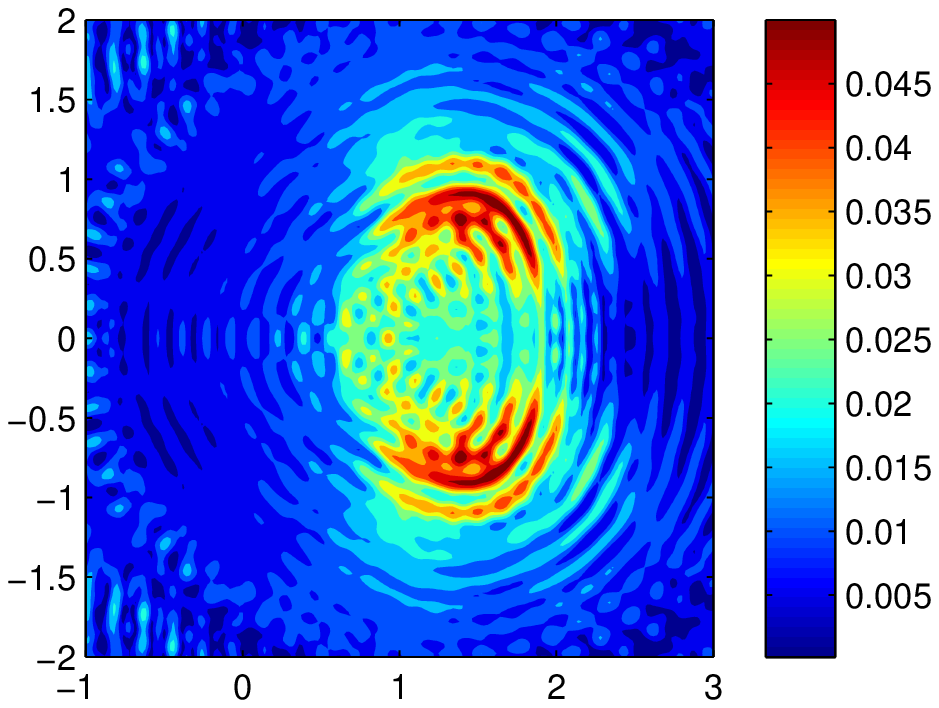}}
\subfloat[Map of $\mathcal{I}_{\mathrm{L}}(\mx;F)$]{\label{FigGamma3NT}\includegraphics[width=0.32\textwidth]{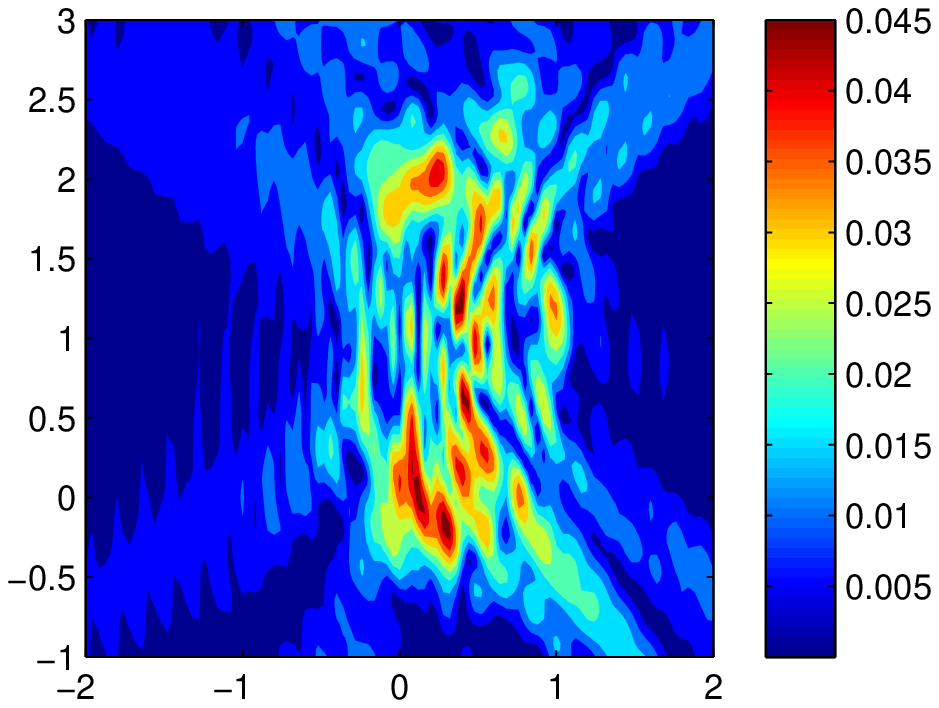}}
\caption{\label{FigGamma3N}Same as Figure \ref{FigGamma1N} except the crack is $\Gamma_3$.}
\end{center}
\end{figure}

\begin{rem}
Similarly with the Dirichlet boundary condition case, we can generate images from far-field data computed by the algorithm presented in \cite{N}, refer to Figure \ref{FigGammaN-Naz}. Numerical experimentation shows that images of a crack from data acquired by the Nystr\"om method or from the ones calculated via this alternative formulation are almost same. Using near-field data instead of far-field one yields similar results.
\end{rem}

\begin{figure}[!ht]
\begin{center}
\subfloat[Map of $\mathcal{I}_{\mathrm{A}}(\mx;F)$ for $\Gamma_1$]{\label{FigGamma1NNaz}\includegraphics[width=0.32\textwidth]{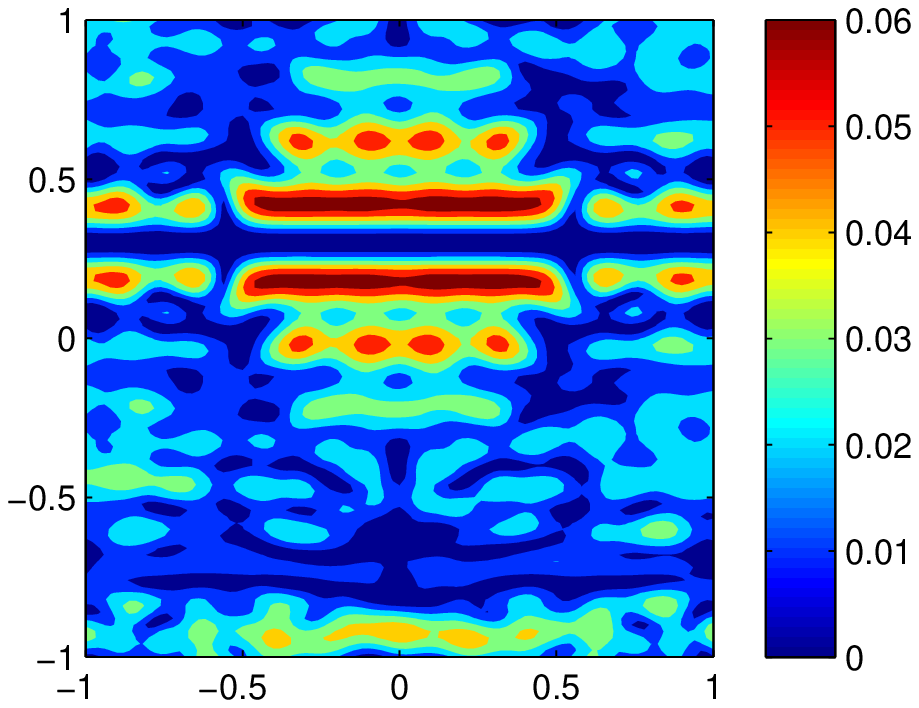}}
\subfloat[Map of $\mathcal{I}_{\mathrm{A}}(\mx;F)$ for $\Gamma_2$]{\label{FigGamma2NNaz}\includegraphics[width=0.32\textwidth]{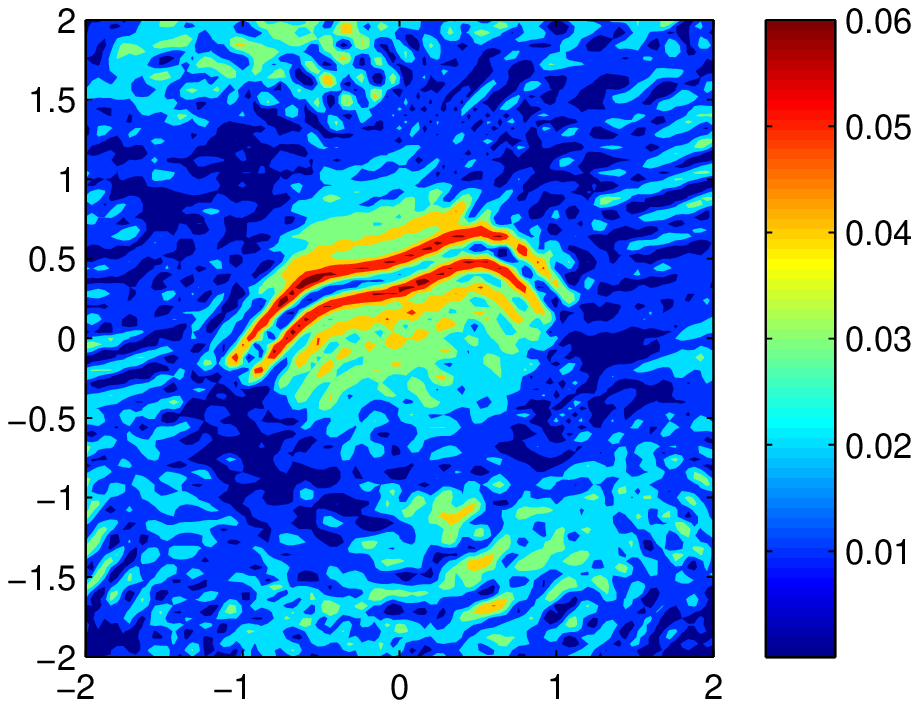}}
\subfloat[Map of $\mathcal{I}_{\mathrm{A}}(\mx;F)$ for $\Gamma_3$]{\label{FigGamma3NNaz}\includegraphics[width=0.32\textwidth]{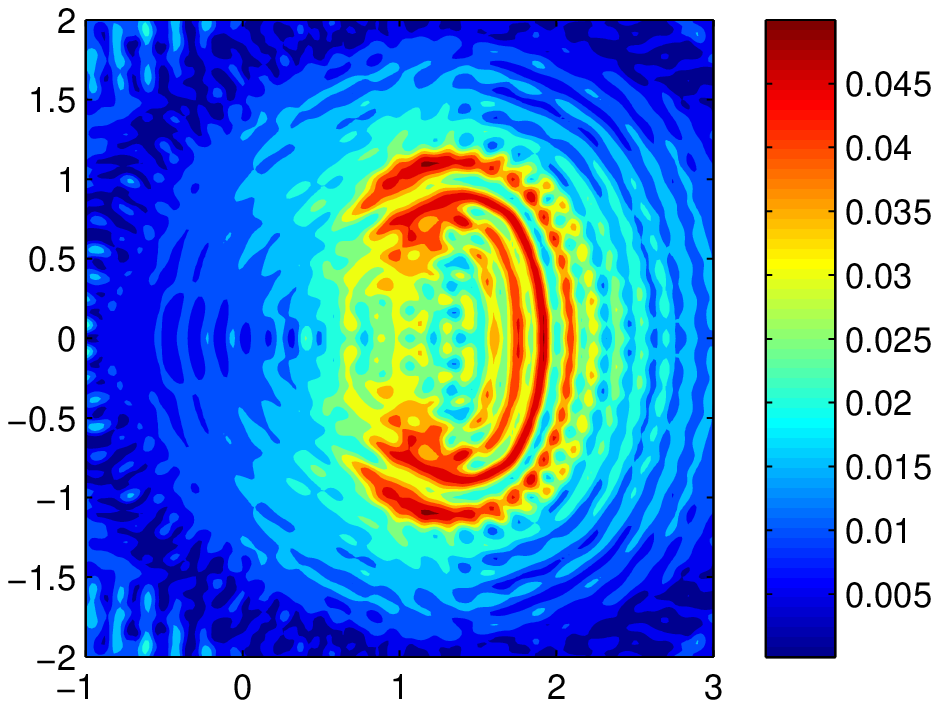}}
\caption{\label{FigGammaN-Naz}Map of $\mathcal{I}_{\mathrm{A}}(\mx;F)$ for $\Gamma_j$, $j=1,2,3$, where dataset generated via method in \cite{N}.}
\end{center}
\end{figure}

For the final example, imaging of multiple cracks is illustrated in Figure \ref{FigGamma4N}. Unlike the Dirichlet boundary condition case, some parts of not only $\Gamma_4^{(1)}$ but also $\Gamma_4^{(2)}$ are retrieved but the result is still poor.

\begin{figure}[!ht]
\begin{center}
\subfloat[Map of $\mathcal{I}_{\mathrm{N}}(\mx;F)$]{\label{FigGamma4NF}\includegraphics[width=0.32\textwidth]{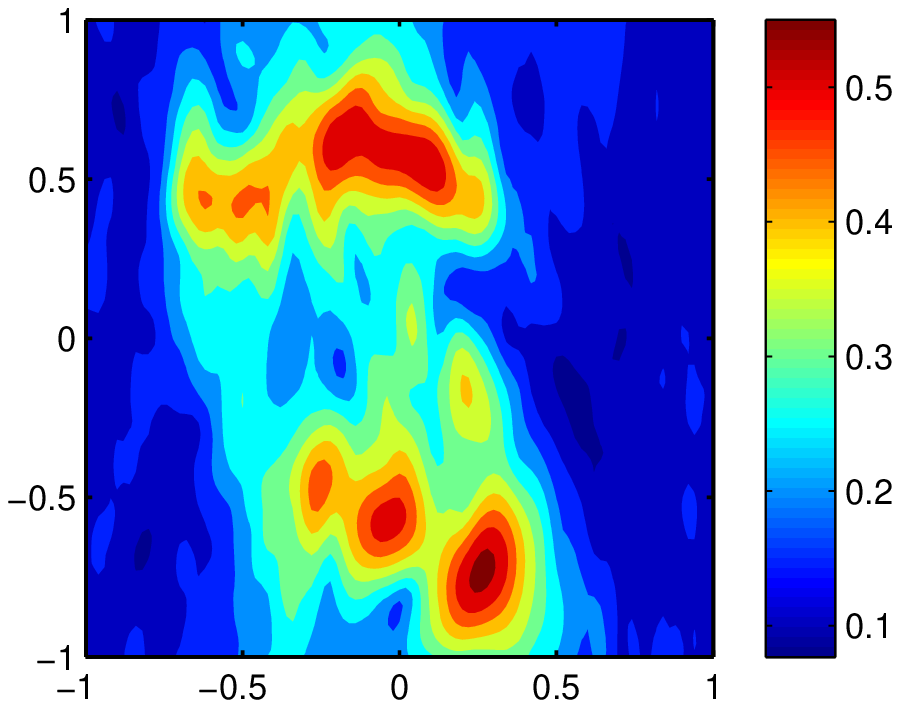}}
\subfloat[Map of $\mathcal{I}_{\mathrm{A}}(\mx;F)$]{\label{FigGamma4NN}\includegraphics[width=0.32\textwidth]{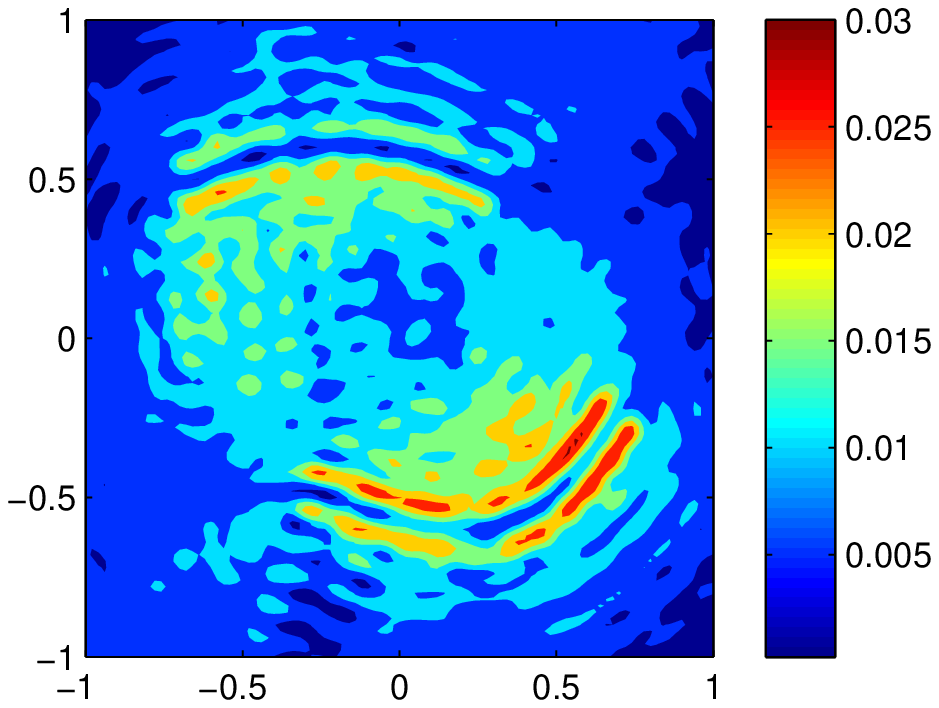}}
\subfloat[Map of $\mathcal{I}_{\mathrm{L}}(\mx;F)$]{\label{FigGamma4NT}\includegraphics[width=0.32\textwidth]{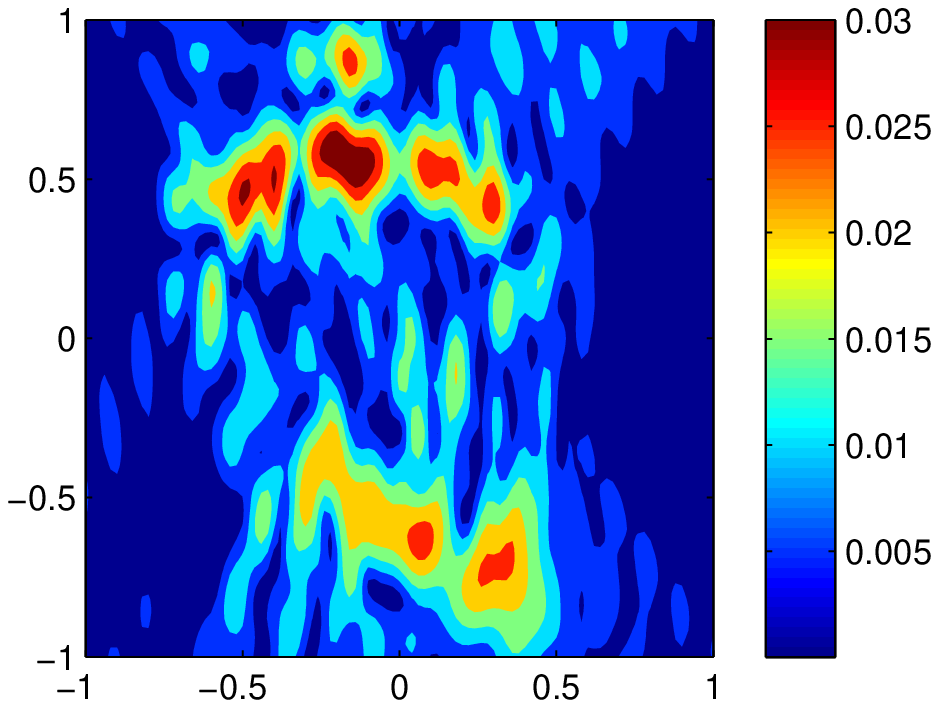}}
\caption{\label{FigGamma4N}Same as Figure \ref{FigGamma1N} except the crack is $\Gamma_4$.}
\end{center}
\end{figure}

\subsection{Complete shape reconstruction - TM case only}\label{Sec4-4}
Now, we consider the complete shape reconstruction of perfectly conducting crack with Dirichlet boundary condition (\ref{HelmBC}). For that purpose, we proceed the Newton method introduced in \cite[Section 7]{K} in order to reconstruct $\Gamma_2$ from far-field measurement with $k=\frac{2\pi}{0.5}$, $\hat{\mx}_j$ of (\ref{IODir}) with $N=8$, $\alpha=\pi/6$, and $\beta=5\pi/6$.

In order to perform the Newton method, we denote $\Gamma_2^{(n)}$ be the crack after $n-$th iteration, i.e., $\Gamma_2^{(0)}$ is the initial guess. Throughout this section, we assume that $\Gamma_2^{(n)}$ can be represented as follows:
\[\Gamma_2^{(n)}=\set{\mz^{(n)}(s):s\in[-1,1]},\]
where $\mz^{(n)}:[-1,1]\longrightarrow\mathbb{R}^2$ is of the form
\[\mz^{(n)}(s)=\left(s,\sum_{j=0}^{p}a_j^{(n)} T_j(s)\right),\quad s\in[-1,1].\]
Here $T_j(s)$ denotes the Chebyshev polynomials of the first kind defined by the recurrence relation
\begin{align*}
  T_0(s)&=1\\
  T_1(s)&=s\\
  T_{j+1}(s)&=2sT_j(s)-T_{j-1}(s).
\end{align*}
Based on the numerical experience in \cite{K}, we use $p=5$ polynomials to reconstruct $\Gamma_2$. Note that true crack $\Gamma_2=\set{\mz(s):s\in[-1,1]}$ is represented as
\[\mz(s)\approx\bigg(s,0.26T_0(s)+0.23T_1(s)-0.22T_2(s)-0.03T_3(s)-0.06T_4(s)\bigg).\]

From the identified parts of $\Gamma_2$ in Figure \ref{FigGamma2DF}, we can evaluate the coefficients $a_j^{(0)}$, $j=0,1,\cdots,5$. With this good initial guess (see Figure \ref{Gamma2iter0}), we apply Newton's method until the value of discrete least square functional in two consecutive steps
\begin{equation}\label{leastsquare}
\mathcal{R}(n):=\frac{1}{2}\sum_{j=1}^{N}\abs{u_\infty^{\mbox{\tiny true}}(\hat{\mx}_j;\vt)-u_\infty^{\mbox{\tiny comp}}(\hat{\mx}_j;\vt)}^2
\end{equation}
was less than a tolerance $0.001$, i.e., we stop this at $n-$th iteration procedure when $|\mathcal{R}(n)-\mathcal{R}(n-1)|<0.001$. In this experiment, only 4 iterations yield a good shape reconstruction of $\Gamma_2$. Obtained values of $a_j^{(n)}$ and corresponding shape of crack $\Gamma_2^{(4)}$ are illustrated in Table \ref{table:Gamma2} and Figure \ref{FigGamma2Iteration}, respectively.

We believe that the results in section \ref{Sec4-3} could be good initial guesses for Neumann boundary condition problem \cite{M2}. It is worth mentioning that if one proceed Newton method with an initial guess that does not close to the true crack (or arbitrary blind initial guess), it is very hard to obtain a desired result even with more iterations, refer to \cite{PL4}.

\begin{table}[!ht]
\begin{center}
\begin{tabular}{c|c|c|c|c|c|c|c}\hline
iterations&$a_0^{(n)}$&$a_1^{(n)}$&$a_2^{(n)}$&$a_3^{(n)}$&$a_4^{(n)}$&$a_5^{(n)}$&value of $\mathcal{R}(n)$\\
\hline\hline
$0$&$0.2741$&$0.2267$&$-0.2062$&$-0.0276$&$-0.0678$&$0.0009$&$0.1203$\\
$1$&$0.2702$&$0.2274$&$-0.2052$&$-0.0277$&$-0.0639$&$0.0007$&$0.0935$\\
$2$&$0.2630$&$0.2274$&$-0.2085$&$-0.0281$&$-0.0614$&$0.0006$&$0.0388$\\
$3$&$0.2622$&$0.2275$&$-0.2110$&$-0.0282$&$-0.0611$&$0.0006$&$0.0344$\\
$4$&$0.2619$&$0.2276$&$-0.2114$&$-0.0282$&$-0.0610$&$0.0006$&$0.0337$\\
\hline
true&$0.2600$&$0.2300$&$-0.2200$&$-0.0300$&$-0.0600$&$0.0000$&~\\
\hline
\end{tabular}
\caption{\label{table:Gamma2}Numerical results for $\Gamma_2$.}
\end{center}
\end{table}

\begin{figure}[!ht]
\begin{center}
\subfloat[Initial guess $\Gamma_2^{(0)}$]{\label{Gamma2iter0}\includegraphics[width=0.32\textwidth]{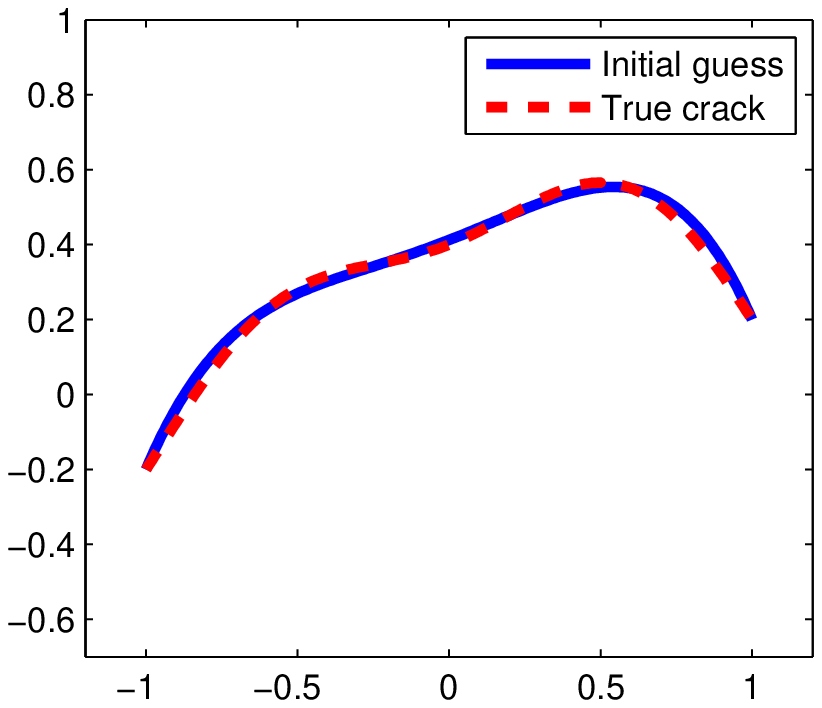}}
\subfloat[After 4 iterations $\Gamma_2^{(4)}$]{\label{Gamma2iter4}\includegraphics[width=0.32\textwidth]{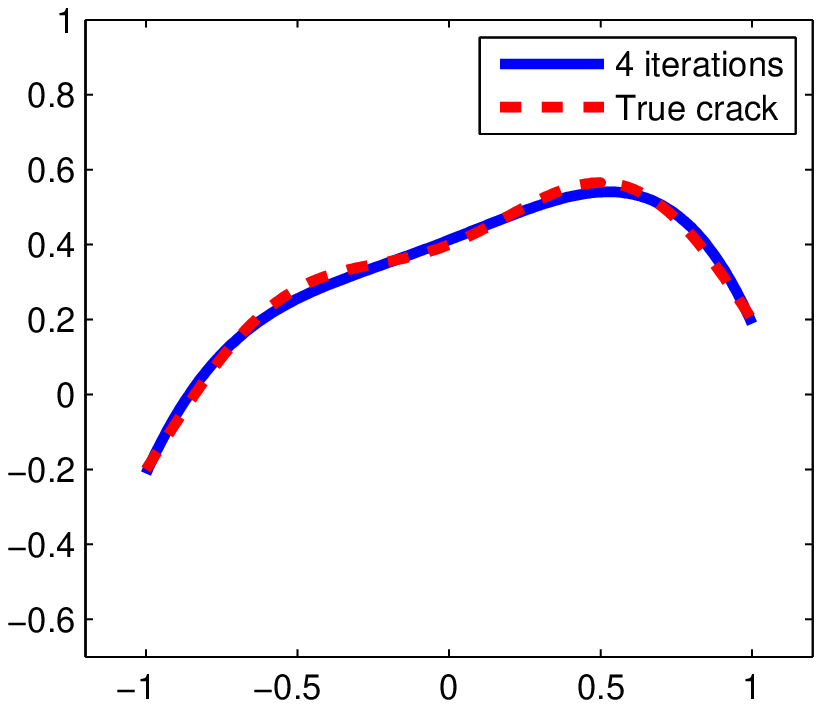}}
\subfloat[Value of least square (\ref{leastsquare})]{\label{Gamma2residue}\includegraphics[width=0.32\textwidth]{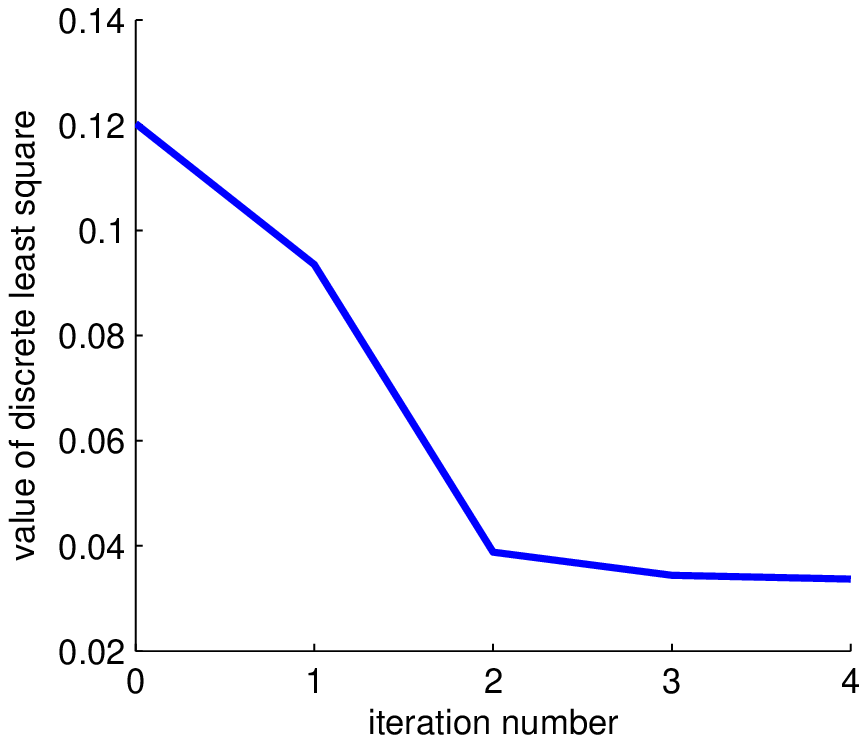}}
\caption{\label{FigGamma2Iteration}Shape reconstruction of $\Gamma_2$ via Newton method in \cite{K}.}
\end{center}
\end{figure}

\section{Conclusion}\label{Sec5}
In this paper, subspace migration imaging technique has been considered to image perfectly conducting, arc-like cracks modeled via a Dirichlet or Neumann boundary condition (TM and TE polarization in two-dimensional electromagnetics) in the two-dimensional full- and limited-view inverse scattering problems. It is based on the factorization of collected Multi-Static Response (MSR) matrix at multi-frequencies of operation and the structure of singular vectors associated to the nonzero singular values.

Throughout rigorous derivation of various definite integrations of Bessel function, we have examined that subspace migration imaging functional can be represented as the combination of Bessel function of integer order of the first kind, and this investigation presents certain properties, limitations in TM and TE polarization cases, and a way of improvements of imaging in TM case.

Presented various numerical simulations from synthetic data computed by rigorous solution methods, it has been shown that the subspace migration imaging technique is very fast, effective and robust with respect to noise for imaging of perfectly conducting cracks. Moreover, it can be easily applied to the imaging of non-overlapped multiple cracks. Nevertheless, some improvements are still required, e.g., when the crack is of large curvature or highly oscillating shaped, and the choice of the normal direction on the crack and method of implementation in TE polarization case.

It is needless to say that such results are obtained at low computational cost. So, though they do not guarantee complete shaping of the cracks, they could be a good initial guess of a level-set evolution or of a standard iterative algorithm \cite{ADIM,AGJKLY,CR,DL,GH,K,KS,M2,PL4,RLLZ,VXB}.

Finally, we mention a point which is interesting but the proof is unsolved and numerical examples are left out in this paper: Opposite to the improvement in section \ref{subsec}, we introduce the following multi-frequency imaging functional weighted by log of given wavenumber $k_f$:
\begin{equation}\label{ImagingFunctionLog}
\mathcal{I}_{\mathrm{WL}}(\mx;F)=\frac{1}{F}\abs{\sum_{f=1}^{F}\sum_{m=1}^{M_f} \ln(k_f)\left(\hat{\mS}_{\mathrm{D}}(\mx;k_f)^*\mU_m(k_f)\right)\left(\hat{\mS}_{\mathrm{D}}(\mx;k_f)^*\overline{\mV}_m(k_f)\right)}.
\end{equation}
Note that if $\mx\ne\my_m$, then since $0\leq\ln(k)J_0(k|\mx-\my_m|)^2\leq kJ_0(k|\mx-\my_m|)^2$,
\begin{align*}
  0\leq\mathcal{I}_{\mathrm{WL}}(\mx;F)&=\frac{1}{k_F-k_1}\left|\sum_{m=1}^{M}\int_{k_1}^{k_F}\ln(k)J_0(k|\mx-\mz|)^2dk\right|\\
  &\leq\frac{1}{k_F-k_1}\left|\sum_{m=1}^{M}\int_{k_1}^{k_F}k J_0(k|\mx-\mz|)^2dk\right|=\mathcal{I}_{\mathrm{W}}(\mx;F,1),
\end{align*}
we can observe that $\mathcal{I}_{\mathrm{WL}}(\mx;F)$ is an improved version of $\mathcal{I}_{\mathrm{W}}(\mx;F,1)$ due to the less oscillation property. Throughout several numerical results, we identified that this imaging functional is also an improvement of multi-frequency subspace migration and offers better results than (\ref{ImagingFunctionW}). In order to identify the structure of (\ref{ImagingFunctionLog}) one must evaluate the following definite integration of Bessel function combined with the natural logarithmic function:
\[\int \ln(x) J_0(x)^2 dx.\]
But in our knowledge, there is no finite representation of this integration. Hence, derivation of this integration and examination of structure of (\ref{ImagingFunctionLog}) should be an interesting and remarkable research topic. Moreover, throughout several numerical experiments, it turns out that if a function $\zeta$ satisfies for sufficiently large $x$
\[1<\zeta(x)<x,\]
then following imaging functional successfully improves $\mathcal{I}_{\mathrm{W}}(\mx;F,1)$
\begin{equation}\label{ImagingFunctionFuture}
\mathcal{I}_{\mathrm{WF}}(\mx;F)=\frac{1}{F}\abs{\sum_{f=1}^{F}\sum_{m=1}^{M_f} \zeta(k_f)\left(\hat{\mS}_{\mathrm{D}}(\mx;k_f)^*\mU_m(k_f)\right)\left(\hat{\mS}_{\mathrm{D}}(\mx;k_f)^*\overline{\mV}_m(k_f)\right)}.
\end{equation}
For example, $\zeta(k_f)=\ln(k_f)$, $\zeta(k_f)=\sqrt{k_f}$, and so on. Hence, finding an optimal function $\zeta$ will be an interesting task.

\end{document}